\documentclass[11pt,onecolumn,letterpaper]{quantumarticle}

\pdfoutput=1

\usepackage{macros}
\usepackage{enumitem}
\usepackage{tabularx}

\allowdisplaybreaks

\title{Conditional contraction coefficients and their applications to quantum networks}

\author{Christoph Hirche}
\affiliation{Institute for Information Processing (tnt/L3S), Leibniz Universit\"at Hannover, Germany}
\email{christoph.hirche@gmail.com}
\author{Ian George}
\affiliation{Centre for Quantum Technologies, National University of Singapore, Singapore 117543, Singapore}
\email{qit.george@gmail.com}
\author{Theshani Nuradha}
\affiliation{Department of Mathematics and
Illinois Quantum Information Science and Technology (IQUIST) Center,\\ 
University of Illinois Urbana-Champaign, Urbana, IL 61801, USA}
\email{theshani.gallage@gmail.com}
\author{Mark M.~Wilde}
\affiliation{School of Electrical and Computer Engineering, Cornell University, Ithaca, New York 14850, USA}
\email{wilde@cornell.edu}

\date{August 21, 2026}

\begin{document}

\maketitle

\begin{abstract}
Contraction coefficients quantify the loss of distinguishability induced by a channel and provide a strong form of the data-processing inequality. While standard contraction coefficients ignore auxiliary quantum systems, existing extensions based on complete contraction coefficients require the compared states to have identical reference marginals. In this work, we introduce conditional contraction coefficients, a novel family that incorporates arbitrary quantum reference systems by subtracting the distinguishability already present in the reference system. We develop a general framework for contraction coefficients with such quantum side information, including the corresponding strong-data-processing-inequality (SDPI) constants, expansion coefficients, and relative contraction coefficients. For the trace distance, we show that the optimization can be restricted to orthogonal input states. For the quantum relative entropy, we prove that its conditional contraction coefficient is exactly equal to the contraction coefficient of the conditional mutual information, extending the classical correspondence between relative-entropy contraction and mutual-information contraction to the setting with quantum side information. More generally, we identify structural properties of divergences required for these results and discuss extensions beyond the relative entropy. These results establish a unified framework for analyzing information contraction in quantum network settings, where quantum side information and distributed correlations are intrinsic features of the information-processing task. Applications include an extension of the Polyanskiy--Wu bounds on mutual information contraction, new perspectives on mixing times, and fundamental limits on quantum memories.
\end{abstract}

\newpage
\tableofcontents

\section{Introduction}

\subsection{Background and motivation}

Contraction coefficients give a strengthening of the data-processing inequality for divergences and are fundamental to our understanding of how systems behave under noise. Given a quantum channel $\cN_{A\to B}$, they provide a universal bound on the loss of distinguishability of 
\begin{align}
    \cN(\rho_A) \text{ versus } \cN(\sigma_A) 
\end{align}
for all quantum states $\rho_A$ and $\sigma_A$. Due to their importance, they have been thoroughly investigated in the literature~\cite{ruskai1994beyond,petz1998contraction,LR1999,reeb2011hilbert,Hiai15,Hirche2022contraction,hirche2024quantumDivergences,george2025unifiedapproachquantumcontraction,delsol2025computationalaspectstracenorm}. Applications range from quantum computing~\cite{fawzi2022lower}, error mitigation~\cite{takagi2023universal,quek2024exponentially}, differential privacy~\cite{Christoph2024sample,nuradha2024contraction}, quantum capacities~\cite{belzig2024reversetypedataprocessinginequality}, and mixing times~\cite{george2025unifiedapproachquantumcontraction}. For a given divergence $\DD$, the contraction coefficient is formally defined as
    \begin{align}
        \eta_\DD(\cN) \coloneqq  \sup_{\substack{\rho_{A}\neq \sigma_{A}}} \frac{\DD(\cN(\rho_{A})\|\cN(\sigma_{A}))}{\DD(\rho_{A}\|\sigma_{A})}. \label{eq:def-gen-contra-coef-intro}
    \end{align}

What is not covered by this setting is the presence of quantum correlations beyond the input system of the channel. We are often interested in the loss of distinguishability of 
\begin{align}
    (\operatorname{id}\otimes \cN)(\rho_{RA}) \text{ versus } (\operatorname{id}\otimes \cN)(\sigma_{RA}) 
\end{align}
for all quantum states $\rho_{RA}$ and $\sigma_{RA}$. However, one can easily find two different states, for example $\rho_{RA}=\rho_R \otimes \tau_A$ and $\sigma_{RA}=\sigma_R \otimes \tau_A$, such that their distinguishability is not affected by the channel. The idea here is simply that all available distinguishability is due entirely to the additional reference system. As a result, a naive extension of contraction coefficients to include reference systems, simply by considering $\id\otimes\cN$, generally leads to a trivial bound~\cite{Hirche2022contraction}. 

A remedy is to restrict the states to be identical on the reference system. The result is the complete contraction coefficient~\cite[Sec.~6]{george2025quantumdoeblincoefficientsinterpretations}, motivated by related complete/tensor-stable relative-entropy inequalities developed in~\cite{gao2020fisher,gao2022complete}, defined as
\begin{align}
        \etaC_\DD(\cN) \coloneqq \sup_{\substack{|R|,\,\rho_{RA}\neq\sigma_{RA},\\ \rho_R=\sigma_R
        }} \frac{\DD((\id\otimes\cN)(\rho_{RA})\|(\id\otimes\cN)(\sigma_{RA}))}{\DD(\rho_{RA}\|\sigma_{RA})}.
        \label{eq:complete-contr-def}
\end{align}
Note that the optimization here is over a potentially unbounded reference system, making this quantity more challenging to work with.

Interestingly, the need for such quantities is a purely quantum phenomenon for certain divergences.
Indeed, if we assume that $R$ is a classical system and that $\DD$ has the extended direct-sum property, then the complete contraction coefficient and the usual one are equal (see \cref{app:classical-ref}). Nevertheless, in the quantum setting, such complete contraction coefficients provide a way of handling reference systems. However, the limitation to equal marginals is constraining and often not reflected in practical applications.

In what follows, we will introduce and discuss an alternative notion of contraction coefficients that takes into account reference systems without the restriction of equal marginals.  

\subsection{Summary of contributions}

In this paper, we propose an alternative notion of a non-trivial contraction coefficient that circumvents the restrictions discussed above. To this end, we define the \textit{conditional contraction coefficient} of a quantum channel $\mathcal{N}$ as follows: 
    \begin{align}
        \etaD_\DD(\cN) \coloneqq \sup_{\substack{|R|,\,\rho_{RA}\neq\sigma_{RA}}} \frac{\DD((\id\otimes\cN)(\rho_{RA})\|(\id\otimes\cN)(\sigma_{RA})) - \DD(\rho_R\|\sigma_R)}{\DD(\rho_{RA}\|\sigma_{RA})- \DD(\rho_R\|\sigma_R)}.
        \label{eq:conditional-contr-coeff-def}
    \end{align}
Here, the reference system is not artificially restricted, and instead we appropriately take into account the distinguishability that it provides. This concept is somewhat similar to the idea of amortized channel divergences~\cite{wilde2020amortized}, in which a reference system can aid the distinguishability. The difference of divergences in the numerator and denominator of~\eqref{eq:conditional-contr-coeff-def} also appeared under the name ``conditional relative entropy'' in~\cite{capel2018quantum}. If we again restrict the optimization to states with equal marginals, we recover the complete contraction coefficient, thus implying the bound $\etaC_\DD(\cN)\leq \etaD_\DD(\cN)$.    

Similar to the complete contraction coefficient, the need for reference systems in the conditional contraction coefficient is a purely quantum phenomena for certain divergences. Indeed, if the reference system is restricted to be classical and the divergence obeys the extended direct-sum property, then the usual contraction coefficient and the conditional one are equal (see \cref{app:classical-ref}).

A different notion of contraction coefficient involves correlation measures instead of divergences, in particular, the mutual information. Such contraction coefficients are closely related to the usual contraction coefficients for the relative entropy. In particular, it is known that~\cite{watanabe2012private,Hirche2022contraction}
\begin{align}
    \eta_D(\cN)=\sup_{\rho_{UA}} \frac{I(U:B)}{I(U:A)}, \label{eq:contr-MI-rel-ent}
\end{align}
where we chose $\mathbb{D}$ to be the quantum relative entropy~\cite{Umegaki1962}, the latter denoted by $D$.
Here $U$ is a classical system, the channel $\mathcal{N}$ has input system $A$ and output system~$B$, and the supremum is over every classical-quantum state $\rho_{UA}$. The notation $I(U:B)$ denotes the mutual information of the output state and $I(U:A)$ the mutual information of the input state. Hence, the contraction of the mutual information is identical to that of the relative entropy.

Since the contraction coefficient in~\eqref{eq:contr-MI-rel-ent} does not capture quantum correlations, one might rather consider fully quantum states or even the conditional mutual information (CMI), incorporating another quantum reference system. Such contractions have been considered, for example, in~\cite{chen2020matrix}. 
One of our main results states that the CMI contraction is indeed the information counterpart to the conditional contraction coefficient $\etaD_D(\cN)$; concretely, 
\begin{align}
    \etaD_D(\cN) = \sup_{\rho_{RCA}} \frac{I(C:B|R)}{I(C:A|R)}.  \label{eq:intro-cond-MI-cond-rel-ent}
\end{align}
Interestingly, for this result, it does not matter whether $C$ is a quantum system, a classical system, or even restricted to be a binary classical system.

In addition to the equality in~\eqref{eq:intro-cond-MI-cond-rel-ent}, we extensively discuss these and other relations between divergence and mutual information based contraction. For a summary of these results, refer to~\cref{fig:contraction-relations} and~\cref{Cor:Summary1}. We then also extend these results to divergences beyond the relative entropy.

\begin{figure}[t]
\centering
\begin{tikzpicture}


\node[
  draw,
  rounded corners=2pt,
  minimum width=1.9cm,
  minimum height=0.9cm
] (A) at (0,0) {$\eta_D$};

\node[
  draw,
  rounded corners=2pt,
  minimum width=1.9cm,
  minimum height=0.9cm
] (B) at (3.0,0) {$\eta_D^p$};

\node[
  draw,
  rounded corners=2pt,
  minimum width=1.9cm,
  minimum height=0.9cm
] (C) at (6.0,0) {$\etaC_D$};

\node[
  draw,
  rounded corners=2pt,
  minimum width=1.9cm,
  minimum height=0.9cm,
  fill=blue!8
] (D) at (9.0,0) {$\etaD_D$};


\node at (0,0.85)
  {\small no reference};

\node at (3.0,0.85)
  {\small product reference};

\node at (6.0,0.85)
  {\small $\rho_R=\sigma_R$};

\node at (9.0,0.85)
  {\small arbitrary $\rho_R,\sigma_R$};


\draw[thick,->] (A.east) -- (B.west);
\draw[thick,->] (B.east) -- (C.west);
\draw[thick,->] (C.east) -- (D.west);

\node at (1.5,0.25) {$\leq$};
\node at (4.5,0.25) {$\leq$};
\node at (7.5,0.25) {$\leq$};


\node[
  draw,
  rounded corners=2pt,
  minimum width=1.9cm,
  minimum height=0.75cm,
  fill=green!8
] (E) at (0,-1.7) {$\eta_{\mathrm{cqMI}}$};

\node[
  draw,
  rounded corners=2pt,
  minimum width=1.9cm,
  minimum height=0.75cm,
  fill=green!8
] (F) at (3.0,-1.7) {$\eta_{\mathrm{MI}}$};

\node[
  draw,
  rounded corners=2pt,
  minimum width=1.9cm,
  minimum height=0.75cm,
  fill=green!8
] (G) at (6.0,-1.7) {$\eta_{\mathrm{MIR}}$};

\node[
  draw,
  rounded corners=2pt,
  minimum width=1.9cm,
  minimum height=0.75cm,
  fill=green!8
] (H) at (9.0,-1.7) {$\eta_{\mathrm{CMI}}$};


\draw[double,double distance=1.5pt,very thick]
  (A.south) -- (E.north);

\draw[double,double distance=1.5pt,very thick]
  (B.south) -- (F.north);

\draw[double,double distance=1.5pt,very thick]
  (C.south) -- (G.north);

\draw[double,double distance=1.5pt,very thick]
  (D.south) -- (H.north);


\node[
  font=\scriptsize,
  anchor=west
] at (0,-1)
  {\cite[Prop.~6.2]{Hirche2022contraction}};

\node[
  font=\scriptsize,
  anchor=west
] at (3,-1)
  {Lem.~\ref{Lem:MI-p}};

\node[
  font=\scriptsize,
  anchor=west
] at (6,-1)
  {Cor.~\ref{Cor:MIR-equals-complete}};

\node[
  font=\scriptsize,
  anchor=west
] at (9,-1)
  {Cor.~\ref{Cor:CMI-amo-equal}};

\end{tikzpicture}

\caption{
Relations between the contraction coefficients considered in this paper
for the relative entropy. Arrows denote inequalities and double lines
denote equalities. References indicate the corresponding results in the
paper. The depicted inequalities are summarized in~\cref{Cor:Summary1}.
}
\label{fig:contraction-relations}
\end{figure}
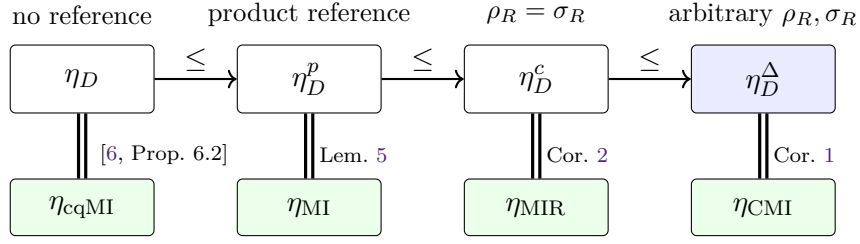

Many of the results described above also hold for the strong-data-processing-inequality (SDPI) constants, formally defined in~\cref{Def:Gen-SDPI}, where the second state is fixed and not subject to optimization. A general SDPI constant is denoted by $\eta_\DD(\cN,\sigma)$, where $\cN$ is a channel and $\sigma$ is a fixed state. One question that is of particular interest for SDPI constants is whether the tensorization property holds, namely whether
\begin{align}
    \eta_\DD(\cN_1\otimes\cN_2,\sigma_1\otimes\sigma_2) \overset{?}{=} \max\{ \eta_\DD(\cN_1,\sigma_1), \eta_\DD(\cN_2,\sigma_2)\}
\end{align}
holds for channels $\cN_1$ and $\cN_2$ and states $\sigma_1$ and $\sigma_2$. This is true for many classical divergences~\cite{Raginsky2016}, but is not in general in the quantum setting~\cite{Cao2026noTensorization}. In~\cref{Sec:Tensorization}, we consider the special case of relative entropy SDPI constants and give an approximate tensorization result that leads to exact tensorization for several special cases, including the generalized depolarizing channel with respect to its fixed point. 

Often it is sensible to consider tensorization in the presence of reference systems. In fact, Ref.~\cite{gao2020fisher} showed that tensorization indeed holds for a particular definition of relative entropy contraction; we refer to~\cref{Sec:Tensorization} for details. We show here that a similar result is equivalent to tensorization of the mutual information SDPI constant,
\begin{align}
    \eta_{\MI}(\cN_1\otimes\cN_2,\sigma_1\otimes\sigma_2) = \max\{ \eta_{\MI}(\cN_1,\sigma_1), \eta_\MI(\cN_2,\sigma_2)\}. 
\end{align}
Even more generally, using a conditional SDPI constant, we show that tensorization also holds for the conditional mutual information, 
\begin{align}
    \eta_\CMI(\cN_1\otimes\cN_2,\sigma_1\otimes\sigma_2) = \max\{ \eta_\CMI(\cN_1,\sigma_1), \eta_\CMI(\cN_2,\sigma_2)\}. 
\end{align}
Thus, our paper significantly extends the understanding of tensorization of SDPI constants beyond the quantities usually considered in the literature. 

Another relevant divergence in the context of contraction coefficients is the trace distance. In particular, its contraction has the favorable property that it is optimized by orthogonal states. We show that this indeed holds also in the presence of reference systems. That is, for the conditional contraction coefficient, the following equality holds:
\begin{align}
    \etaD_{\Tr}(\cN) &= \sup_{\substack{\tau_{RA} \perp \bar\tau_{RA}}} \frac{E_1(\cN(\tau_{RA})\|\cN(\bar\tau_{RA})) - E_1(\tau_R\|\bar\tau_R)}{1- E_1(\tau_R\|\bar\tau_R)}.
\end{align}
Considering instead the more general hockey-stick divergence we give an example demonstrating that the conditional contraction coefficient can  be strictly larger than the complete contraction coefficient, 
\begin{align}
    \etaD_{E_\gamma}(\cN) > \etaC_{E_\gamma}(\cN),
\end{align}
even in the classical setting. 

Next, we consider applications of our extended contraction framework.
First, we discuss contraction in network settings. As our main result in this context, we generalize a bound from~\cite{polyanskiy2015strong} to the quantum setting. For a network channel $\cN_{A\to BE} = \cR_{D\to E}\circ\cM_{A\to BD}$, we prove the following inequality:
    \begin{align}
        \eta_{\MI}(\cN_{A\to BE}) &\leq \eta_{\cR}\cdot\eta_{\MI}(\cM_{A\to BD}) + (1-\eta_{\cR})\cdot \eta_{\MI}(\cM_{A\to B}).
\end{align}
We follow up with implications of this inequality, including novel bounds on the contraction coefficients of tensor-product channels. 
The second main application is the introduction of \textit{replacing time}, as a generalization of mixing time. It is defined by
\begin{align}               t_{\operatorname{rep}}^{\mbb{D}}(\cN,\delta) \coloneqq \inf \left\{ n \in \mbb{N} : \begin{matrix}   \mbb{D}((\id_{R} \otimes \cN^{n})(\rho_{RA}) \Vert \sigma_{R} \otimes \omega_{A}) - \DD(\rho_{R} \Vert \sigma_{R}) \leq \delta \, , \\
                \forall R \, , \, \forall \rho_{RA},\sigma_{RA} \in  \Density(R \otimes A) \end{matrix} \right\} 
            \end{align}
and we give bounds in terms of the conditional contraction coefficient and show that replacing is asymptotically equivalent to mixing. 

Finally, we give some applications to bounding storage times in quantum memories and limitations in quantum machine learning.

\subsection{Paper organization}

The remainder of our paper is organized as follows. First, we establish various definitions and prove some general properties in~\cref{Sec:Definitions}. Here we also discuss expansion coefficients and relative contraction coefficients. In~\cref{Sec:Relative-entropy}, we focus on the relative entropy and, in particular, its relation to the mutual information and the CMI. We also extend the results to divergences that are structurally similar to the relative entropy, such as $f$-divergences. In~\cref{Sec:Tensorization}, we continue with the relative entropy and discuss tensorization of SDPI constants, providing our exact and approximate tensorization results. 
 In~\cref{Sec:Hockey-Stick}, we specialize the investigation to the trace distance and, more generally, the hockey-stick divergence. 
 We study the applications of our new contraction coefficients in quantum networks, discrete-time quantum Markov chains, and quantum machine learning, in~\cref{Sec:Applications1},~\cref{Sec:Applications2}, and~\cref{Sec:Applications3}, respectively.
Finally, we conclude in~\cref{Sec:Conclusions} and briefly discuss some open problems.


\section{Contraction with reference systems}

\label{Sec:Definitions}

In this section, we begin by recalling quantum divergences and their properties (\cref{sec:q-divs}). After that, we introduce generalized SDPI constants and contraction coefficients, including our new definition of conditional contraction coefficients and restricted coefficients with a product constraint on the reference system (\cref{sec:contrac-coeffs}). \cref{sec:props-contr-coeffs} then establishes various properties of these contraction coefficients. In~\cref{sec:expan-coeffs}, we introduce expansion coefficients, and finally, we introduce relative contraction and expansion coefficients, in which two channels are compared (\cref{sec:relative-contr-expans}). 

\subsection{Quantum divergences}

\label{sec:q-divs}

First, we recall some basic definitions for quantum divergences. 
We define a quantum divergence $\DD$ to be a function that maps two quantum states to a real number and satisfies the data-processing inequality: 
\begin{align}
    \DD(\rho\|\sigma) \geq \DD(\cN(\rho)\|\cN(\sigma)) .
\end{align}
Sometimes we need additional properties beyond this basic statement.

A divergence $\DD$ is \textit{non-negative} if 
\begin{align}
    \DD(\rho\|\sigma)\geq0 \label{Eq:nonneg}
\end{align}
for all states $\rho$ and $\sigma$, and we say that a non-negative divergence is \textit{faithful} if
\begin{equation}
\DD(\rho\|\sigma) = 0 \quad \Leftrightarrow \quad \rho = \sigma.    
\end{equation}

 A divergence $\DD$ is \textit{additive} if
\begin{align}
    \DD(\rho_A\otimes\rho_B\|\sigma_A\otimes\sigma_B) = \DD(\rho_A\|\sigma_A) + \DD(\rho_B\|\sigma_B),
\end{align}
for all states $\rho_A$, $\rho_B$, $\sigma_A$, and $\sigma_B$.

We say that $\DD$ has the \textit{direct-sum property} if the following equality holds for all classical-quantum states $\rho_{XA} = \sum_x p(x) |x\rangle\!\langle x|_X \otimes \rho^x_A$ and $\sigma_{XA}= \sum_x p(x) |x\rangle\!\langle x|_X \otimes \sigma^x_A$:
\begin{align}
    \DD(\rho_{XA}\|\sigma_{XA}) =\sum_x p(x) \DD(\rho_A^x\|\sigma_A^x). \label{eq:direct-sum-prop}
\end{align}

A divergence $\DD $ is \textit{jointly convex} if the following inequality holds for all $\lambda\in [0,1]$ and states $\rho_0$, $\rho_1$, $\sigma_0$, and $\sigma_1$:
\begin{equation}
    (1-\lambda)\DD(  \rho_0 \|  \sigma_0)  + \lambda\DD( \rho_1  \|  \sigma_1 )\geq \DD((1-\lambda) \rho_0 + \lambda \rho_1  \| (1-\lambda) \sigma_0 + \lambda \sigma_1  )  
     . \label{eq:jointly-convex-def}
\end{equation}
A divergence that obeys the direct-sum property is jointly convex, as a consequence of~\eqref{eq:direct-sum-prop} and data processing under partial trace over the classical register $X$.

We will introduce further more specific properties where needed. 

\subsection{Contraction coefficients}

\label{sec:contrac-coeffs}

We now introduce various notions of contraction coefficient. Here we distinguish  between different settings, depending on which states are known. If we want an improvement in the data-processing inequality for all states in the first argument, but fixed second argument, then we consider SDPI constants of the form 
    \begin{align}
        \eta_\DD(\cN,\sigma_A)    \coloneqq \sup_{\substack{\rho_{A}\neq \sigma_{A}}} \frac{\DD(\cN(\rho_{A})\|\cN(\sigma_{A}))}{\DD(\rho_{A}\|\sigma_{A})}. \label{eq:def-gen-SDPI-coef}
    \end{align}
If we want an improvement to hold when both inputs are arbitrary, we consider contraction coefficients, as defined in~\eqref{eq:def-gen-contra-coef-intro}. They can be phrased as optimized SDPI constants:
    \begin{align}
        \eta_\DD(\cN) = \sup_\sigma \eta_\DD(\cN,\sigma) = \sup_{\substack{\rho_{A}\neq \sigma_{A}}} \frac{\DD(\cN(\rho_{A})\|\cN(\sigma_{A}))}{\DD(\rho_{A}\|\sigma_{A})}. \label{eq:def-gen-contra-coef}
    \end{align}
    
We now define versions of these that explicitly include reference systems, as initiated in~\cite{gao2022complete} for the relative entropy. Let us begin with the SDPI constants. 

\begin{definition}[Generalized SDPI constants]\label{Def:Gen-SDPI}
    For a quantum channel $\cN$, a state $\sigma_{RA}$, and a generalized divergence $\DD$, we define the SDPI constant with product reference as
        \begin{align}
        \label{SDPI-product}\eta^p_{\DD}(\cN,\sigma_A)\coloneqq \sup_{\substack{|R|,\,\rho_{RA}}} \frac{\DD(\cN(\rho_{RA})\|\rho_R\otimes \cN(\sigma_A))}{\DD(\rho_{RA}\|\rho_R\otimes\sigma_A)}, 
    \end{align}
    the complete SDPI constant as~\cite{george2025quantumdoeblincoefficientsinterpretations}
    \begin{align}
        \etaC_\DD(\cN,\sigma_{RA}) \coloneqq \sup_{\substack{\rho_{RA}\neq\sigma_{RA},\\ \rho_R=\sigma_R }} \frac{\DD((\id\otimes\cN)(\rho_{RA})\|(\id\otimes\cN)(\sigma_{RA}))}{\DD(\rho_{RA}\|\sigma_{RA})}, 
    \end{align}
    and the conditional SDPI constant as
    \begin{align}
        \etaD_\DD(\cN,\sigma_{RA}) \coloneqq \sup_{\substack{\rho_{RA}\neq\sigma_{RA}}} \frac{\DD((\id\otimes\cN)(\rho_{RA})\|(\id\otimes\cN)(\sigma_{RA})) - \DD(\rho_R\|\sigma_R)}{\DD(\rho_{RA}\|\sigma_{RA})- \DD(\rho_R\|\sigma_R)}.
    \end{align}
\end{definition}

Similarly, we define contraction coefficients with reference systems.

\begin{definition}[Generalized contraction coefficient]
    For a quantum channel $\cN$ and a generalized divergence $\DD$, we define the contraction coefficient with product reference as
    \begin{align}
        \eta^p_\DD(\cN)\coloneqq \sup_{\substack{|R|,\,\rho_{RA},\,\sigma_A:\\\rho_{RA} \neq \rho_R \otimes \sigma_A}} \frac{\DD(\cN(\rho_{RA})\|\rho_R\otimes \cN(\sigma_A))}{\DD(\rho_{RA}\|\rho_R\otimes\sigma_A)}, 
    \end{align}
    the complete contraction coefficient as~\cite{george2025quantumdoeblincoefficientsinterpretations} 
    \begin{align}
        \etaC_\DD(\cN) \coloneqq \sup_{\substack{|R|,\,\rho_{RA}\neq\sigma_{RA},\\ \rho_R=\sigma_R 
        }} \frac{\DD((\id\otimes\cN)(\rho_{RA})\|(\id\otimes\cN)(\sigma_{RA}))}{\DD(\rho_{RA}\|\sigma_{RA})}, \label{eq:Comp_CC}
    \end{align}
    and the conditional contraction coefficient as
    \begin{align}
        \etaD_\DD(\cN) \coloneqq \sup_{\substack{|R|,\,\rho_{RA}\neq\sigma_{RA}}} \frac{\DD((\id\otimes\cN)(\rho_{RA})\|(\id\otimes\cN)(\sigma_{RA})) - \DD(\rho_R\|\sigma_R)}{\DD(\rho_{RA}\|\sigma_{RA})- \DD(\rho_R\|\sigma_R)}.\label{Eq:conditional-def}
    \end{align}
\end{definition}

\begin{remark}
    Extending the notion of conditional relative entropy from~\cite{capel2018quantum}, we can define the generalized conditional relative entropy of states $\rho_{RA}$ and $\sigma_{RA}$ as follows:
    \begin{align}
    \DD_A(\rho_{RA}\|\sigma_{RA}) \coloneqq  \DD(\rho_{RA}\|\sigma_{RA}) - \DD(\rho_{R}\|\sigma_{R}).         
    \end{align}
    Using this notation, the conditional contraction coefficient can alternatively be expressed as 
        \begin{align}
        \etaD_\DD(\cN) = \sup_{\substack{\rho_{RA}\neq\sigma_{RA}}} \frac{\DD_B((\id\otimes\cN)(\rho_{RA})\|(\id\otimes\cN)(\sigma_{RA}))}{\DD_A(\rho_{RA}\|\sigma_{RA})}.
    \end{align}
\end{remark}

\begin{remark}
    We note here that the notation used above for $\etaD_\DD(\cN,\sigma_{RA}) $ and $\etaD_\DD(\cN)$ is incomplete. The optimization is should be over states such that $\infty>\DD(\rho_{RA}\|\sigma_{RA})- \DD(\rho_R\|\sigma_R)>0$, which is a strictly stronger constraint than just $\rho_{RA}\neq\sigma_{RA}$. However, for notational brevity, we will keep using the latter notation. As a general rule, in this paper, all ratio optimizations are over strictly positive and finite denominators. If no admissible state with strictly positive finite denominator exists, the corresponding contraction coefficient is defined to be zero.
\end{remark}

\subsection{Properties of contraction coefficients}

\label{sec:props-contr-coeffs}

In this section, we will first discuss some general properties of the definitions from \cref{sec:contrac-coeffs} for arbitrary divergences, before specializing to particular divergences in the following sections. 

The following inequalities hold for every non-negative divergence, as defined in~\eqref{Eq:nonneg},
\begin{align}
    0\leq\eta_\DD(\cN) &\leq \eta^p_\DD(\cN) \leq \etaC_\DD(\cN)  \leq  \etaD_\DD(\cN)\leq 1. 
    \label{eq:simple-ineqs-relations}
\end{align}
Furthermore, the following bounds hold for a concatenation of channels. 
\begin{lemma}
    For two quantum channels $\cN$ and $\cM$ and a state $\sigma_{RA}$, the following hold for the SDPI constants, 
    \begin{align}
        \eta^p_\DD(\cM\circ\cN,\sigma) & \leq \eta^p_\DD(\cM,\cN(\sigma)) \;\eta^p_\DD(\cN,\sigma), \\
        \etaC_\DD(\cM\circ\cN,\sigma) & \leq \etaC_\DD(\cM,\cN(\sigma)) \;\etaC_\DD(\cN,\sigma), \\
        \etaD_\DD(\cM\circ\cN,\sigma) & \leq \etaD_\DD(\cM,\cN(\sigma)) \;\etaD_\DD(\cN,\sigma),
    \end{align}
    where $\mathcal{N}(\sigma) \equiv (\operatorname{id}\otimes \mathcal{N})(\sigma_{RA})$.
    Hence the following hold for the contraction coefficients:
        \begin{align} 
        \eta^p_\DD(\cM\circ\cN) & \leq \eta^p_\DD(\cM) \;\eta^p_\DD(\cN), \\
        \etaC_\DD(\cM\circ\cN) & \leq \etaC_\DD(\cM) \;\etaC_\DD(\cN), \\
        \etaD_\DD(\cM\circ\cN) & \leq \etaD_\DD(\cM) \;\etaD_\DD(\cN). \label{eq:concatenation_etaD}
    \end{align}
\end{lemma}

\begin{proof}
    We show the second inequality. Consider that 
    \begin{align}
        &\DD((\id\otimes(\cM\circ\cN))(\rho_{RA})\|(\id\otimes(\cM\circ\cN))(\sigma_{RA})) \notag \\
        &\leq \etaC_\DD(\cM,\cN(\sigma)) \; \DD((\id\otimes\cN)(\rho_{RA})\|(\id\otimes\cN)(\sigma_{RA})) \\
        &\leq \etaC_\DD(\cM,\cN(\sigma)) \; \etaC_\DD(\cN,\sigma)\; \DD(\rho_{RA}\|\sigma_{RA}), 
    \end{align}
    where both inequalities hold by the definition of the SDPI constant. The other two claims for SDPI constants follows from a similar proof. Those for contraction coefficients then follow by optimizing over~$\sigma_{RA}$.
\end{proof}

If the underlying divergence is also additive, then the behavior of the contraction coefficients under tensor products is similar to that found previously~\cite[Lemma 3.8]{Hirche2022contraction}.

\begin{lemma}
\label{lem:tensorization-ineqs-lem}
    For quantum channels $\cN$ and $\cM$, states $\sigma_{R_1A_1}$ and $\sigma_{R_2A_2}$, and a generalized divergence $\DD$ that is additive on tensor-product states, the following hold for the SDPI constants:
    \begin{align}
        \eta^p_\DD(\cM\otimes\cN,\sigma_{A_1}\otimes\sigma_{A_2}) &\geq \max\{ \eta^p_\DD(\cM,\sigma_{A_1}), \eta^p_\DD(\cN,\sigma_{A_2}) \}, \\
        \etaC_\DD(\cM\otimes\cN,\sigma_{R_1A_1}\otimes\sigma_{R_2A_2}) &\geq \max\{ \etaC_\DD(\cM,\sigma_{R_1A_1}), \etaC_\DD(\cN,\sigma_{R_2A_2}) \}, \\
        \etaD_\DD(\cM\otimes\cN,\sigma_{R_1A_1}\otimes\sigma_{R_2A_2}) &\geq \max\{ \etaD_\DD(\cM,\sigma_{R_1A_1}), \etaD_\DD(\cN,\sigma_{R_2A_2}) \},
    \end{align}
    and the following hold for the contraction coefficients:
    \begin{align}
        \eta^p_\DD(\cM\otimes\cN) &\geq \max\{ \eta^p_\DD(\cM), \eta^p_\DD(\cN) \}, \\
        \etaC_\DD(\cM\otimes\cN) &\geq \max\{ \etaC_\DD(\cM), \etaC_\DD(\cN) \}, \\
        \etaD_\DD(\cM\otimes\cN) &\geq \max\{ \etaD_\DD(\cM), \etaD_\DD(\cN) \}.
    \end{align}
\end{lemma}
\begin{proof}
    For the complete SDPI constant, the proof follows by restricting the supremum over $\rho_{R_1A_1R_2A_2}$ to states of the form either $\rho_{R_1A_1}\otimes\sigma_{R_2A_2}$ or $\sigma_{R_1A_1}\otimes\rho_{R_2A_2}$, and similar for the conditional case and that with product reference. 
    Here we use that $\DD(\rho\otimes\tau\|\sigma\otimes\tau)=\DD(\rho\|\sigma)$, which is also a consequence of data processing; see, for example~\cite[Eq.~(9)]{wilde2014strong}. 
    For the contraction coefficients, the claim follows by restricting both states in the optimization to product states where one channel input gets assigned a fixed input state in both cases. 
\end{proof}
While these tensorization inequalities for contraction coefficients are usually strict, in the case of SDPI constants achievability depends on the exact definition used. We discuss this further in~\cref{Sec:Tensorization}. 

The contraction coefficients are also convex if the divergence is jointly convex (see~\eqref{eq:jointly-convex-def}).

\begin{lemma}\label{Lem:Convexity}
    For a jointly convex divergence $\DD$, $\lambda\in [0,1] $, channels $\mathcal{N}_1$ and $\mathcal{N}_2$, and a channel $\cN_{\lambda}\equiv \lambda\cN_1 + (1-\lambda)\cN_0$, the following inequalities hold:
    \begin{align}
        \eta^p_\DD(\cN_{\lambda}) &\leq \lambda \eta^p_\DD(\cN_1) + (1-\lambda) \eta^p_\DD(\cN_0), \\
        \etaC_\DD(\cN_{\lambda}) &\leq \lambda \etaC_\DD(\cN_1) + (1-\lambda) \etaC_\DD(\cN_0), \\ 
        \etaD_\DD(\cN_{\lambda}) &\leq \lambda \etaD_\DD(\cN_1) + (1-\lambda) \etaD_\DD(\cN_0). 
    \end{align}
\end{lemma}
\begin{proof}
    The properties follow directly from the joint convexity of the divergence.
\end{proof}

\subsection{Quantum Doeblin coefficient upper bounds on contraction coefficients}

Contraction coefficients are generally difficult to compute, and so upper bounds are also of interest. For this purpose, Refs.~\cite{hirche2024quantum,george2025quantumdoeblincoefficientsinterpretations} employed quantum Doeblin coefficients, extending the original development in \cite[Theorem~8.17 and Eq.~(8.86)]{wolf2012quantum}. Here we use the following definition:
\begin{align}\label{Eq:Doeblin+}
    \alpha_+(\cN) \coloneqq  \sup\{\epsilon \in [0,1] : \cE_\epsilon \deg \cN \}, 
\end{align}
where $\cE_\epsilon(\rho) \coloneqq  (1-\epsilon)\rho + \epsilon |e\rangle\!\langle e|$ is the quantum erasure channel and $\cN$ can be degraded from $\cE_\epsilon$, denoted $\cE_\epsilon \deg \cN$, if there exists a quantum channel $\cD$ such that $\cN=\cD\circ\cE_\epsilon$. Note that the Doeblin coefficient is a semidefinite program (SDP)~\cite{hirche2024quantum}.
Additionally, we will later use the following expression.
\begin{proposition}[{\cite[Proposition III.3]{hirche2024quantum}}]
\label{prop:Hirche-Doeblin-Replacer-Form}
    For a quantum channel $\cN_{A \to B}$,
    \begin{align}
        \alpha_{+}(\cN) = \max\{\epsilon : \exists \cD \in \operatorname{CPTP}(A,B) \text{ s.t. } \cN = (1-\epsilon)\cD + \epsilon\cR_{\tau}\} \ , 
    \end{align}
    where $\cR_{\tau}(X) = \tau \Tr[X]$ and $\tau $ is a quantum state. 
\end{proposition}

We then have the following result. 

\begin{proposition}
\label{Lem:a-Doeblin}
    For a jointly convex divergence $\DD$, the following inequality holds: 
    \begin{align}
        \etaD_\DD(\cN) \leq 1-\alpha_+(\cN).
    \end{align}
\end{proposition}

\begin{proof}
    By~\cref{prop:Hirche-Doeblin-Replacer-Form}, we have that $\cN = (1-\epsilon)\cD + \epsilon\cR_{\tau}$ with $\epsilon=\alpha_{+}(\cN)$. Hence, we get the following,
    \begin{align}
        \etaD_\DD(\cN) \leq (1-\epsilon) \etaD_\DD(\cD) + \epsilon \etaD_\DD(\cR_\tau) \leq 1-\epsilon = 1-\alpha_+(\cN),
    \end{align}
    where the first inequality follows from~\cref{Lem:Convexity} and the second uses $\etaD_\DD(\cD)\leq1$ and $\etaD_\DD(\cR_\tau)=0$, which completes the proof.
\end{proof}

The bound in~\cref{Lem:a-Doeblin} above represents an improvement over the previous bound from \cite[Proposition~31 and Corollary~7]{george2025quantumdoeblincoefficientsinterpretations}. Indeed, there we proved, for some choices of $\DD$, that $\eta^c_\DD(\cN) \leq 1-\alpha_+(\cN)$. However, since $\eta^c_\DD(\cN)\leq \etaD_\DD(\cN)$, it then follows that $\etaD_\DD(\cN)$ is a tighter lower bound on $1-\alpha_+(\cN)$ than $\eta^c_\DD(\cN)$. Additionally, we have proven the claim rather generally for every conditional contraction coefficient based on a divergence that is jointly convex.

\subsection{Expansion coefficients}

\label{sec:expan-coeffs}

In contrast to contraction coefficients, one may also ask if a channel preserves at least a fixed amount of information. To that end, expansion coefficients were introduced~\cite{ramakrishnan2020computing,hirche2024quantum,laracuente2023information} and defined as
    \begin{align}
        \check\eta_\DD(\cN) \coloneqq  \inf_{\substack{\rho_{A}\neq \sigma_{A}}} \frac{\DD(\cN(\rho_{A})\|\cN(\sigma_{A}))}{\DD(\rho_{A}\|\sigma_{A})}. \label{eq:def-gen-exp-coef}
    \end{align}
A natural extension is to introduce a reference system also here. Hence, we define the following generalized expansion coefficients. 
\begin{definition}[Generalized expansion coefficient]
    For a quantum channel $\cN$ and a generalized divergence $\DD$, we define the expansion coefficient with product reference as 
        \begin{align}
        \check\eta^p_\DD(\cN)\coloneqq \inf_{\substack{|R|, \, \rho_{RA},\sigma_A:\\\rho_{RA} \neq \rho_R \otimes \sigma_A}} \frac{\DD(\cN(\rho_{RA})\|\rho_R\otimes \cN(\sigma_A))}{\DD(\rho_{RA}\|\rho_R\otimes\sigma_A)}, 
    \end{align}
    the complete expansion coefficient as~\cite{george2025quantumdoeblincoefficientsinterpretations}
    \begin{align}
        \cetaC_\DD(\cN) \coloneqq \inf_{\substack{ |R|,\, \rho_{RA}\neq\sigma_{RA},\\ \rho_R=\sigma_R 
       }} \frac{\DD((\id\otimes\cN)(\rho_{RA})\|(\id\otimes\cN)(\sigma_{RA}))}{\DD(\rho_{RA}\|\sigma_{RA})}, \label{eq:Exp_CC}
    \end{align}
    and the conditional expansion coefficient as
    \begin{align}
        \cetaD_\DD(\cN) \coloneqq \inf_{\substack{|R|,\, \rho_{RA}\neq\sigma_{RA}}} \frac{\DD((\id\otimes\cN)(\rho_{RA})\|(\id\otimes\cN)(\sigma_{RA})) - \DD(\rho_R\|\sigma_R)}{\DD(\rho_{RA}\|\sigma_{RA})- \DD(\rho_R\|\sigma_R)}.
    \end{align}
\end{definition}
Also here, one could additionally define SDPI-like versions by keeping the second state fixed. 

An important note to make here is that the usefulness of these coefficients depends strongly on the choice of divergence. It was found that even the standard expansion coefficient is often trivial (i.e., equal to zero), for many divergences~\cite{hirche2024quantum,laracuente2023information,belzig2024reversetypedataprocessinginequality,iyer2025quantum}. Opposite to~\eqref{eq:simple-ineqs-relations}, the following inequalities hold:
\begin{align}
    \check\eta_\DD(\cN) &\geq
    \check\eta^p_\DD(\cN) \geq \cetaC_\DD(\cN)  \geq  \cetaD_\DD(\cN)\geq0,
\end{align}
implying that our generalized expansion coefficients with reference systems are equal to zero for the scenarios considered in~\cite{hirche2024quantum,laracuente2023information,belzig2024reversetypedataprocessinginequality,iyer2025quantum}. One common case where such expansion coefficients remain generally interesting is when they are defined in terms of the trace distance~\cite{hirche2024quantum}.

\subsection{Relative contraction and expansion coefficients}

\label{sec:relative-contr-expans}

The goal of relative contraction coefficients is to compare the data-processing strength of different channels. Such quantities appeared classically as early as~\cite[Eq.~(43)]{Makur2018SymChannels}. The idea was then picked up in the quantum setting in~\cite{Hirche2022contraction}, but only investigated in much more detail in~\cite{belzig2024reversetypedataprocessinginequality}; see also~\cite{iyer2025quantum}. Related recent work in the classical setting can also be found in~\cite{steinberg2025degradedness}.

We generalize this concept to the setting with reference systems here. 
\begin{definition}[Generalized relative contraction coefficient]
    For quantum channels $\cN$ and $\cM$ and a generalized divergence $\DD$, we define the relative complete contraction coefficient as follows: 
    \begin{align}
        \etaC_\DD(\cN,\cM) \coloneqq \sup_{\substack{|R|, \rho_{RA}\neq\sigma_{RA},\\ \rho_R=\sigma_R}} \frac{\DD((\id\otimes\cN)(\rho_{RA})\|(\id\otimes\cN)(\sigma_{RA}))}{\DD((\id\otimes\cM)(\rho_{RA})\|(\id\otimes\cM)(\sigma_{RA}))}, \label{eq:Comp_Rel_CC}
    \end{align}
    and the relative conditional contraction coefficient as follows:
    \begin{align}\label{eq:relative_contraction_channels_cond}
        \etaD_\DD(\cN,\cM) \coloneqq \sup_{\substack{|R|, \rho_{RA}\neq\sigma_{RA}}} \frac{\DD((\id\otimes\cN)(\rho_{RA})\|(\id\otimes\cN)(\sigma_{RA})) - \DD(\rho_R\|\sigma_R)}{\DD((\id\otimes\cM)(\rho_{RA})\|(\id\otimes\cM)(\sigma_{RA})) - \DD(\rho_R\|\sigma_R)}.
    \end{align}
\end{definition}
\begin{definition}[Generalized relative expansion coefficient]
    For quantum channels $\cN$ and $\cM$ and a generalized divergence $\DD$, we define the relative complete expansion coefficient as follows:
    \begin{align}
        \cetaC_\DD(\cN,\cM) \coloneqq \inf_{\substack{|R|, \rho_{RA}\neq\sigma_{RA},\\ \rho_R=\sigma_R 
        }} \frac{\DD((\id\otimes\cN)(\rho_{RA})\|(\id\otimes\cN)(\sigma_{RA}))}{\DD((\id\otimes\cM)(\rho_{RA})\|(\id\otimes\cM)(\sigma_{RA}))}, \label{eq:Exp_Rel_CC}
    \end{align}
    and the relative conditional expansion coefficient as follows:
    \begin{align}\label{eq:relative_expansion_channel_cond}
        \cetaD_\DD(\cN,\cM) \coloneqq \inf_{\substack{|R|,\,\rho_{RA}\neq\sigma_{RA}}} \frac{\DD((\id\otimes\cN)(\rho_{RA})\|(\id\otimes\cN)(\sigma_{RA})) - \DD(\rho_R\|\sigma_R)}{\DD((\id\otimes\cM)(\rho_{RA})\|(\id\otimes\cM)(\sigma_{RA})) - \DD(\rho_R\|\sigma_R)}.
    \end{align}
\end{definition}
Similar to what was observed in \cite[after Eq.~(4)]{belzig2024reversetypedataprocessinginequality}, relative expansion and relative contraction coefficients are related by
\begin{align}
    \etaC_\DD(\cN,\cM) & = \left(\cetaC_\DD(\cM,\cN) \right)^{-1}, \label{Eq:rel-ep-con-c}\\
    \etaD_\DD(\cN,\cM) & = \left(\cetaD_\DD(\cM,\cN) \right)^{-1}.\label{Eq:rel-ep-con-a}
\end{align}
Furthermore, such coefficients have  proven useful in the study of quantum capacities~\cite{belzig2024reversetypedataprocessinginequality}. 

In the next sections, we will dive deeper into investigating these quantities for specific divergences. 


\section{Relative entropy-like divergences}\label{Sec:Relative-entropy}

In this section, we will discuss contraction of primarily the relative entropy and the connection to contraction of related quantities such as the mutual information and conditional mutual information. However, we will also see that we only need very few properties of the relative entropy for most of the results, and we will later discuss similar divergences that possess the same properties. 

\subsection{(Conditional) mutual information contraction}\label{subsec:CMI-contraction}

Mutual information and conditional mutual information are correlation measures that have numerous applications in information theory and can be written in terms of the relative entropy. It is hence natural that the data-processing properties of mutual information and relative entropy are tightly connected. Without reference systems, several results in this direction are known~\cite{Hirche2022contraction}. Here, we extend the results to the setting with reference systems, which in particular allows us to incorporate the conditional mutual information. Unless stated otherwise, we adopt the shorthand convention
\begin{equation}
\mathcal{N}(\rho_{RA}) \equiv (\operatorname{id}_R \otimes \mathcal{N}_{A\to B})(\rho_{RA})   
\label{eq:shorthand-notation}
\end{equation}
in all that follows.

The quantum relative entropy of states $\rho$ and $\sigma$ is defined as \cite{Umegaki1962}
\begin{equation}
    D(\rho\|\sigma)\coloneqq \Tr[\rho(\log \rho - \log \sigma)]
\end{equation}
when $\operatorname{supp}(\rho)\subseteq \operatorname{supp}(\sigma)$ and it is equal to $+\infty$ otherwise.
The mutual information and the conditional mutual information are respectively defined for a tripartite state $\rho_{ABC}$ as 
\begin{align}
    I(A:B) &\coloneqq  D(\rho_{AB}\|\rho_A\otimes\rho_{B}), \\
    I(A:B|C) &\coloneqq  I(A:BC) - I(A:C) \label{eq:CMI-MI-chain-rule}\\
    &= D(\rho_{ABC}\|\rho_A\otimes\rho_{BC}) - D(\rho_{AC}\|\rho_A\otimes\rho_{C}). 
\end{align}
There are many ways to define meaningful contraction coefficients for the mutual information, some of which explicitly choose certain subsystems as classical. We list several of these possibilities below.

\begin{definition}
    Given a channel $\cN_{A \to B}$, we define the following mutual-information and conditional-mutual-information contraction coefficients:
\begin{align}
    \eta_{\MI}(\cN) &\coloneqq \sup_{\rho_{CA}} \frac{I(C:B)}{I(C:A)}, \label{eq:eta_MI} \\
    \eta_{\cqMI}(\cN) &\coloneqq \sup_{\rho_{UA}} \frac{I(U:B)}{I(U:A)},\\
    \eta_{\tcqMI}(\cN) &\coloneqq \sup_{\substack{\rho_{UA},\\|U|=2}} \frac{I(U:B)}{I(U:A)},\\
    \eta_{\CMI}(\cN) &\coloneqq \sup_{\rho_{RCA}} \frac{I(C:B|R)}{I(C:A|R)}, 
    \\
    \eta_{\operatorname{cqCMI}}(\cN) &\coloneqq \sup_{\rho_{RUA}} \frac{I(U:B|R)}{I(U:A|R)}, 
    \\
    \eta_{\operatorname{2cqCMI}}(\cN) &\coloneqq \sup_{\substack{\rho_{RUA},\\|U|=2}} \frac{I(U:B|R)}{I(U:A|R)}, 
\end{align}
where $\rho_{CBR}\coloneqq \cN_{A \to B}(\rho_{RCA})$ and $U$ is always a classical system, meaning that $\rho_{UA}$ and $\rho_{RUA}$ are classical-quantum states.
\end{definition}

It was shown in~\cite[Proposition~6.2]{Hirche2022contraction} that 
    \begin{align}
        \eta_{\cqMI}(\cN) &= \eta_D(\cN).
    \end{align}
However, upon reexamination, the proof there actually shows something stronger, namely, that
    \begin{align}
        \label{Eq:cqMI=D}
        \eta_{\tcqMI}(\cN) = \eta_{\cqMI}(\cN) &= \eta_D(\cN).
    \end{align}
This is a, perhaps surprising, simplification of the contraction coefficient, which was previously remarked upon in the classical literature; see~\cite[Footnote 5]{anantharam2013maximal}. We refrain from giving a proof here because it is the same as that given for~\cite[Proposition 6.2]{Hirche2022contraction}, while noting that $|U|=2$ always suffices for one direction of the proof.

Instead, we consider a more novel setting with the goal of extending previous results to the conditional mutual information and connect that to the conditional contraction coefficient.

\begin{theorem}
\label{Thm:CMI-equivalence}
    Let $\cM\colon A\to B$ and $\cN\colon A\to B'$ be quantum channels, $\sigma_{RA}$ a quantum state, and $\eta\geq0$. Then the following are equivalent:
\begin{enumerate}[label=(\roman*)]
    \item For all c-q states $\rho_{RUA}$ with marginal $\sigma_{RA}$, where U is a classical system of dimension $|U|=2$, 
    \begin{align}
        \eta I(U:B|R) \geq I(U:B'|R).
        \label{eq:CMI-contract-1}
    \end{align}
    \item For all c-q states $\rho_{RUA}$ with marginal $\sigma_{RA}$, where U is an arbitrary classical system, 
    \begin{align}
        \eta I(U:B|R) \geq I(U:B'|R).\label{eq:CMI-contract-2}
    \end{align}
\item For every quantum state $\rho_{RA}$ satisfying $\operatorname{supp}(\rho_{RA})\subseteq \operatorname{supp}(\sigma_{RA})$, 
\begin{align}
    \eta \left[ D(\cM(\rho_{RA})\|\cM(\sigma_{RA})) - D(\rho_R\|\sigma_R) \right] \geq D(\cN(\rho_{RA})\|\cN(\sigma_{RA})) - D(\rho_R\|\sigma_R).
\end{align}
\end{enumerate}
In~\eqref{eq:CMI-contract-1}--\eqref{eq:CMI-contract-2}, the conditional mutual informations on the left-hand side are evaluated with respect to $\mathcal{M}(\rho)$ and those on the right-hand side with respect to $\mathcal{N}(\rho)$.
\end{theorem}

\begin{proof}
    The proof for $(ii)\Rightarrow(i)$ is obvious, given that the set of states in $(ii)$ contains those in $(i)$.
    The direction $(iii)\Rightarrow(ii)$ follows easily by the direct-sum property. 
    
    We now focus on the proof of $(i)\Rightarrow(iii)$. Without loss of generality we assume that $\sigma_{RA}$ is a full-rank state. Then, for every $\rho_{RA}$, and $0\leq\lambda\leq\epsilon$, where $\epsilon$ is small enough such that $\sigma_{RA}-\epsilon\rho_{RA}\geq 0$, we let $U\equiv U_\lambda$ be a binary random variable with distribution $P_{U_\lambda}(0)=\lambda$ and $P_{U_\lambda}(1)=1-\lambda$, and we define states $\rho_{RA}^0\coloneqq \rho_{RA}$ and $\rho_{RA}^1(\lambda)\coloneqq (1-\lambda)^{-1}(\sigma_{RA}-\lambda\rho_{RA})$. Defining
    \begin{equation}
        \rho_{URA} \coloneqq \lambda |0\rangle \!\langle 0|_U \otimes \rho^0_{RA} + (1-\lambda) |1\rangle\!\langle 1|_U \otimes \rho^1_{RA}(\lambda),
        \label{eq:rho-URA-state}
    \end{equation}
     we conclude that $\tr_U(\rho_{RUA})=\sigma_{RA}$. Then, under the assumption that $(i)$ holds and using the direct-sum property of quantum relative entropy, define
    \begin{align}
        \varphi(\lambda) &\coloneqq  \eta I(U:B|R) - I(U:B'|R) \\
        &= \eta \left[ I(U:BR) - I(U:R)\right] - \left[ I(U:B'R) - I(U:R)\right] \\
    &= \eta \left[ D(\rho_{UBR}\|\rho_U\otimes\rho_{BR}) - D(\rho_{UR}\|\rho_U\otimes\rho_{R}) \right]\nonumber\\ &\qquad - \left[ D(\rho_{UB'R}\|\rho_U\otimes\rho_{B'R}) - D(\rho_{UR}\|\rho_U\otimes\rho_{R}) \right] \\
    &= \eta \left[ \lambda D(\cM(\rho_{RA})\|\cM(\sigma_{RA})) 
    + (1-\lambda) D(\cM(\rho^1_{RA}(\lambda))\|\cM(\sigma_{RA}))\right.\nonumber\\ &\qquad 
    \left.-\lambda D(\rho_{R}\|\sigma_{R}) 
    - (1-\lambda) D(\rho^1_{R}(\lambda)\|\sigma_{R}) \right] \nonumber\\
    &\quad- \left[ \lambda D(\cN(\rho_{RA})\|\cN(\sigma_{RA})) 
    + (1-\lambda) D(\cN(\rho^1_{RA}(\lambda))\|\cN(\sigma_{RA}))\right.\nonumber\\ &\qquad 
    \left.
    - \lambda D(\rho_{R}\|\sigma_{R}) 
    - (1-\lambda) D(\rho^1_{R}(\lambda)\|\sigma_{R}) \right].
    \end{align}
    At $\lambda = 0$, the equality $\rho^1_{RA}(0)=\sigma_{RA}$ holds, and we conclude that $\varphi(0)=0$. Moreover, $(i)$ implies that $\varphi(\lambda)\geq0$ for all $\epsilon\geq\lambda\geq0$. Hence, it must be the case that $\varphi'(0)\geq 0$. The desired result follows from computing the derivative $\varphi'(0)$: 
    \begin{align}
    \varphi'(0) &= \eta \left[ D(\cM(\rho_{RA})\|\cM(\sigma_{RA})) - D(\rho_{R}\|\sigma_{R}) \right] - D(\cN(\rho_{RA})\|\cN(\sigma_{RA})) + D(\rho_{R}\|\sigma_{R}) \nonumber\\
    &\quad+\Bigg( (1-\lambda) \frac{\partial}{\partial\lambda}\left[ \eta D(\cM(\rho^1_{RA}(\lambda))\|\cM(\sigma_{RA}))-\eta D(\rho^1_{R}(\lambda)\|\sigma_{R}) \right.\nonumber\\ &\qquad\qquad\qquad\qquad 
    \left.- D(\cN(\rho^1_{RA}(\lambda))\|\cN(\sigma_{RA})) + D(\rho^1_{R}(\lambda)\|\sigma_{R}) \right] \Bigg)\Bigg|_{\lambda=0}\\
    &= \eta \left[ D(\cM(\rho_{RA})\|\cM(\sigma_{RA})) - D(\rho_{R}\|\sigma_{R}) \right] - D(\cN(\rho_{RA})\|\cN(\sigma_{RA})) + D(\rho_{R}\|\sigma_{R}), 
    \end{align}
   where the second equality is a well known property of the relative entropy by which each individual one of the remaining derivatives vanishes at $\lambda=0$; see~\cite[Section~4]{petz1998contraction} and later discussions in~\cite{Hirche2022contraction,hirche2024quantumDivergences,beigi2025some}. For convenience, we provide a brief proof of the needed equality in~\cref{app:prop-rel-ent}.
\end{proof}

A direct consequence of \cref{Thm:CMI-equivalence} is the following result regarding our previously defined contraction coefficients.
\begin{corollary}
\label{Cor:CMI-amo-equal}
    For a channel $\cN$, we have
    \begin{align}
               \etaD_D(\cN) = \eta_{\operatorname{2cqCMI}}(\cN) = \eta_{\operatorname{cqCMI}}(\cN) = \eta_{\CMI}(\cN).
    \end{align}
\end{corollary}

\begin{proof}
    The first two equalities follow immediately from \cref{Thm:CMI-equivalence} by setting $\cM=\id$ and fixing $\eta$ to be the relevant contraction coefficient. 
    We only need to show the final equality. Clearly $\eta_{\operatorname{cqCMI}}(\cN) \leq \eta_{\CMI}(\cN)$. We conclude the claim by showing $\eta_{\CMI}(\cN)\leq\etaD_D(\cN)$. 
    Observe that
        \begin{align}
        I(C:B'|R')
        &= D(\cN(\rho_{R'CA})\|\rho_C\otimes\cN(\rho_{R'A})) - D(\rho_{CR'}\|\rho_C\otimes\rho_{R'}) \\
        &\leq \eta \left[ D(\cM(\rho_{R'CA})\|\rho_C\otimes\cM(\rho_{R'A})) - D(\rho_{CR'}\|\rho_C\otimes\rho_{R'}) \right] \\
        &=\eta I(C:B|R'),
    \end{align}
    where we identify $R'C \leftrightarrow R$, $\rho_{R'CA} \leftrightarrow \rho_{RA}$, and $\rho_C \otimes \rho_{R'A} \leftrightarrow \sigma_{RA}$.
    This implies the desired inequality. 
\end{proof}

\begin{remark}
    We can combine~\cref{Lem:a-Doeblin} and~\cref{Cor:CMI-amo-equal} to conclude the following inequality:
\begin{align}
    \eta_{\CMI}(\cN) = \etaD_D(\cN) \leq 1-\alpha_+(\cN)\label{eq:CMI-contract-doeblin},
\end{align}
thus reproducing the finding of \cite[Proposition~III.2]{chen2020matrix} after taking \cref{prop:Hirche-Doeblin-Replacer-Form} into account.
\end{remark}

\subsection{Mutual information with quantum reference systems}

The next question that we address is as follows: What is the mutual information counterpart of the complete contraction coefficient? 
We first check its relationship with the fully quantum mutual information contraction coefficient.

\begin{lemma}
\label{Lem:MI-complete}
    For a channel $\cN$, the following inequality holds:
    \begin{align}
        \eta_{\MI}(\cN) &\leq \etaC_D(\cN),
    \end{align}
    where $\etaC_D(\cN)$ is defined in~\eqref{eq:Comp_CC}  and $\eta_{\MI}(\cN)$ in~\eqref{eq:eta_MI}.
\end{lemma}
\begin{proof}
    Consider that
    \begin{align}
        \eta_{\MI}(\cN) = \sup_{\rho_{CA}} \frac{I(C:B)}{I(C:A)} 
        &= \sup_{\rho_{CA}} \frac{D(\cN(\rho_{CA})\|\rho_C\otimes\cN(\rho_{A}))}{D(\rho_{CA}\|\rho_C\otimes\rho_{A})} \\
        &\leq \etaC_D(\cN) \sup_{\rho_{CA}} \frac{D(\rho_{CA}\|\rho_C\otimes\rho_{A})}{D(\rho_{CA}\|\rho_C\otimes\rho_{A})} \\
        &= \etaC_D(\cN),
    \end{align}
    which implies the claim. 
\end{proof}

It is not currently known whether the inequality in~\cref{Lem:MI-complete} is indeed an equality. 
We have not made use of the product nature of the state in the second argument in the above proof. In order to find mutual-information-based contraction coefficients that do equal the complete contraction coefficient, we define the following quantities. 

\begin{definition}
For a channel $\cN$, we define the following mutual-information-based contraction coefficients:
\begin{align}
    \eta_{\operatorname{MIR}}(\cN) &\coloneqq  \sup_{\substack{\rho_{RCA},\\ \rho_{CR}=\rho_C\otimes\rho_R}} \frac{I(C:BR)}{I(C:AR)}, \\
    \eta_{\operatorname{cqMIR}}(\cN) &\coloneqq  \sup_{\substack{\rho_{RUA},\\ \rho^u_R=\rho^v_R\,\forall u,v} } \frac{I(U:BR)}{I(U:AR)}, \\
    \eta_{\operatorname{2cqMIR}}(\cN) &\coloneqq  \sup_{\substack{\rho_{RUA},\\ \rho^0_R=\rho^1_R, \\ |U|=2}} \frac{I(U:BR)}{I(U:AR)},
\end{align}
where the optimization is now over classical-quantum states of the form
\begin{align}
    \rho_{URA} = \sum_u p(u) |u\rangle\!\langle u|_U\otimes \rho^u_{RA},
\end{align}
defined by an ensemble $\left((p(u),\rho^u_{RA})\right)_u$ in which $\rho^u_R=\rho^v_R$ for all $u,v$, or in the second case, where there are only two states, those have to be equal.    
\end{definition}
 
We refer to the above contraction coefficients as MIR (``MI with Reference'') contraction coefficients. The MIR contraction coefficients can be expressed in alternative ways, as captured in the following proposition:

\begin{proposition}
    For a channel $\cN$, the following equalities hold:
    \begin{align}
        \eta_{\operatorname{MIR}}(\cN) & = \sup_{\substack{\rho_{RCA},  \\ \rho_{CR} = \rho_{C} \otimes \rho_{R}}} \frac{I(C:B \vert R)}{I(C:A\vert R)} \label{eq:MIR-to-CMI}\\
        & = \sup_{\substack{\rho_{RCA} , \\ \rho_{CR} = \rho_{C} \otimes \rho_{R}}} \frac{I(C:B:R)-I(B:R)}{I(C:A:R) - I(A:R) \  },\label{eq:MIR-to-multipartite-MI}
    \end{align}
    where the multipartite mutual information of a tripartite state $\rho_{ABC}$ is defined as
\begin{equation}
I(A:B:C) \coloneqq  D(\rho_{ABC} \Vert \rho_{A} \otimes \rho_{B} \otimes \rho_{C})    .
\end{equation}
\end{proposition}

\begin{proof}
The equality in \eqref{eq:MIR-to-CMI} follows from~\eqref{eq:CMI-MI-chain-rule} and the fact that $I(C:R) = 0$ because $\rho_{CR} = \rho_{C} \otimes \rho_{R}$.

Using the relative entropy chain rule in~\eqref{eq:rel-ent-chain}, we find that
\begin{align}
    I(C:BR) &= D(\cN(\rho_{CAR}) \Vert \rho_{C} \otimes \cN(\rho_{AR})) \\
    &= D(\cN(\rho_{CAR}) \Vert \rho_{C} \otimes \cN(\rho_{A}) \otimes \rho_{R}) - D(\cN(\rho_{AR}) \Vert \cN(\rho_{A}) \otimes \rho_{R}) \\
    &= I(C:B:R) - I(B:R) \ , 
\end{align}
thus establishing \eqref{eq:MIR-to-multipartite-MI}.
\end{proof}

We now establish a  result similar to~\cref{Thm:CMI-equivalence}. 

\begin{theorem}
\label{Thm:MIR-equivalence}
    Let $\cN\colon A\to B'$ and $\cM\colon A\to B$ be quantum channels, $\sigma_{RA}$ a quantum state, and $\eta\geq0$. Then the following are equivalent:
\begin{enumerate}[label=(\roman*)]
    \item For all c-q states $\rho_{RUA}$ with marginal $\sigma_{RA}$, where U is a classical system of size $|U|=2$, and $\rho_R^0=\rho_R^1$, 
    \begin{align}
        \eta I(U:BR) \geq I(U:B'R).
    \end{align}
    \item For all c-q states $\rho_{RUA}$ with marginal $\sigma_{RA}$, where U is an arbitrary classical system, and $\rho_R^u=\rho^v_R$ for all $u,v$,  
    \begin{align}
        \eta I(U:BR) \geq I(U:B'R).
    \end{align}
    \item For all quantum states $\rho_{RA}$ with $\operatorname{supp}(\rho_{RA})\subseteq \operatorname{supp}(\sigma_{RA})$ and $\rho_R=\sigma_R$, 
\begin{align}
    \eta  D(\cM(\rho_{RA})\|\cM(\sigma_{RA})) \geq D(\cN(\rho_{RA})\|\cN(\sigma_{RA})).
\end{align}
\end{enumerate}
\end{theorem}

\begin{proof}
    The proof is structurally similar to that of~\cref{Thm:CMI-equivalence}, so we keep it brief.
    The proof for $(ii)\Rightarrow(i)$ is again obvious, given that the set of states in $(ii)$ contain those in $(i)$.
    Also $(iii)\Rightarrow(ii)$ follows immediately by the direct-sum property.
    
    It remains to show $(i)\Rightarrow(iii)$: Without loss of generality we assume that $\sigma_{RA}$ is a full-rank state. Then, for every $\rho_{RA}$ such that $\rho_R=\sigma_R$, and $0\leq\lambda\leq\epsilon$, where $\epsilon$ is small enough such that $\sigma_{RA}-\epsilon\rho_{RA}\geq 0$, we let $U\equiv U_\lambda$ be a binary random variable of distribution $P_{U_\lambda}(0)=\lambda$ and $P_{U_\lambda}(1)=1-\lambda$, and states $\rho_{RA}^0=\rho_{RA}$ and $\rho_{RA}^1(\lambda)=(1-\lambda)^{-1}(\sigma_{RA}-\lambda\rho_{RA})$. After defining $\rho_{URA}$ as in~\eqref{eq:rho-URA-state}, clearly we have that $\tr_U(\rho_{RUA})=\sigma_{RA}$. The main difference with the proof of~\cref{Thm:CMI-equivalence} is that we now have $\rho_{R}=\rho_{R}^0=\rho_{R}^1=\sigma_{R}$. Then, define
    \begin{align}
        \varphi(\lambda) &\coloneqq  \eta I(U:BR) - I(U:B'R) \\
        &= \eta  D(\rho_{UBR}\|\rho_U\otimes\rho_{BR}) - D(\rho_{UB'R}\|\rho_U\otimes\rho_{B'R})  \\
    &= \eta \left[ \lambda D(\cM(\rho_{RA})\|\cM(\sigma_{RA})) 
    + (1-\lambda) D(\cM(\rho^1_{RA}(\lambda))\|\cM(\sigma_{RA}))\right] \nonumber\\
    &\quad- \left[ \lambda D(\cN(\rho_{RA})\|\cN(\sigma_{RA})) 
    + (1-\lambda) D(\cN(\rho^1_{RA}(\lambda))\|\cN(\sigma_{RA}))\right].
    \end{align}
    Since $\rho^1_{RA}(0)=\sigma_{RA}$, we have that $\varphi(0)=0$. Moreover, $(i)$ implies that $\varphi(\lambda)\geq0$ for all $\lambda\in [0,\epsilon]$. Hence, it must be the case that $\varphi'(0)\geq 0$. The result follows from computing the derivative $\varphi'(0)$:
    \begin{align}
    \varphi'(0) &= \eta  D(\cM(\rho_{RA})\|\cM(\sigma_{RA}))  - D(\cN(\rho_{RA})\|\cN(\sigma_{RA}))  \nonumber\\
    &\quad+ \left((1-\lambda) \left.\frac{\partial}{\partial\lambda}\left[ \eta D(\cM(\rho^1_{RA}(\lambda))\|\cM(\sigma_{RA})) - D(\cN(\rho^1_{RA}(\lambda))\|\cN(\sigma_{RA}))  \right] \right)\right|_{\lambda=0}\\
    &= \eta  D(\cM(\rho_{RA})\|\cM(\sigma_{RA}))  - D(\cN(\rho_{RA})\|\cN(\sigma_{RA})), 
    \end{align}
   where the second equality follows again from~\cref{lem:prop-rel-ent}. 
\end{proof}

\cref{Thm:MIR-equivalence} implies the following result. 
\begin{corollary}\label{Cor:MIR-equals-complete}
    For a channel $\cN$, the following equalities hold:
    \begin{align}
               \etaC_D(\cN)= \eta_{\operatorname{2cqMIR}}(\cN) = \eta_{\operatorname{cqMIR}}(\cN) = \eta_{\operatorname{MIR}}(\cN).
    \end{align}
\end{corollary}
\begin{proof}
    Also here, the first two equalities follow immediately from~\cref{Thm:MIR-equivalence}.
    The third equality follows after noting that the $CR$ system in the mutual information acts as the reference system and is indeed given by the same state $\rho_{CR}=\rho_C\otimes\rho_R$ in both arguments.
\end{proof}

With this, we have found equivalences for almost all of our contraction coefficients. The last one that does not have a relation to the relative entropy contraction coefficient is the mutual information contraction coefficient $\eta_{\MI}(\cN)$. Given the results in~\cref{Thm:CMI-equivalence} and~\cref{Thm:MIR-equivalence}, it is somewhat remarkable that it does not seem to fit in with the usual relative entropy contraction coefficient. In \cref{sec:coeff-seps}, we will see that those coefficients are indeed different. To complete the map, we show that $\eta^p_D(\cN)$ and $\eta_{\MI}(\cN)$ are equal.

\begin{lemma}
\label{Lem:MI-p}
  For a channel $\cN$, the following equality holds:
      \begin{align}
        \eta_{\MI}(\cN) = \eta^p_D(\cN). 
    \end{align}
\end{lemma}
\begin{proof}
    By the definition of mutual information, we immediately have $\eta_{\MI}(\cN) \leq \eta^p_D(\cN)$. We will now prove the opposite direction. To that end, recall the chain rule for relative entropy that follows from direct calculation:
    \begin{align}
    \label{eq:rel-ent-chain}
        D(\rho_{AB} \Vert \rho_{A} \otimes \sigma_{B}) = D(\rho_{AB} \Vert \rho_{A} \otimes \rho_{B}) +D(\rho_{B} \Vert \sigma_{B}) \ . 
    \end{align}
    Note that, 
    \begin{align}
        D(\cN(\rho_{RA})\|\rho_R\otimes \cN(\sigma_A)) &= D(\cN(\rho_{RA})\|\rho_R\otimes\cN(\rho_A)) + D(\cN(\rho_{A})\|\cN(\sigma_A)) \\
        &\leq \eta_{\MI}(\cN) D(\rho_{RA}\|\rho_R\otimes\rho_A) + \eta_D(\cN) D(\rho_{A}\|\sigma_A) \\
        &\leq \max\{\eta_{\MI}(\cN),\eta_D(\cN)\} \left[D(\rho_{RA}\|\rho_R\otimes\rho_A) +  D(\rho_{A}\|\sigma_A)\right] \\
        &=\max\{\eta_{\MI}(\cN),\eta_D(\cN)\} D(\rho_{RA}\|\rho_R\otimes\sigma_A) \\
        &= \eta_{\MI}(\cN) D(\rho_{RA}\|\rho_R\otimes\sigma_A),  
    \end{align}
    where in the last equality we have made use of $\eta_{\MI}(\cN)\geq\eta_D(\cN)$, which follows directly from~\eqref{Eq:cqMI=D}.  
\end{proof}

Having found equivalent expressions for all quantities we defined initially, we summarize our results in the following corollary.
\begin{corollary}
\label{Cor:Summary1}
    For a channel $\cN$, we have
    \begin{align}
               \eta_{\tcqMI}(\cN) & = \eta_{\cqMI}(\cN) = \eta_D(\cN) \\
               &\leq \eta_{\MI}(\cN)= \eta^p_D(\cN) \\
               &\leq \eta_{\operatorname{2cqMIR}}(\cN) = \eta_{\operatorname{cqMIR}}(\cN) = \eta_{\operatorname{MIR}}(\cN) = \etaC_D(\cN) \\
               &\leq \eta_{\operatorname{2cqCMI}}(\cN) = \eta_{\operatorname{cqCMI}}(\cN) = \eta_{\CMI}(\cN) = \etaD_D(\cN).
    \end{align}
\end{corollary}
The first inequality can be strict; see the next section. 

\subsection{Separations between the coefficients}

\label{sec:coeff-seps}

We know that for qubit depolarizing channels (see, e.g.,~\cite[Eq.~(3.54)]{Hirche2022contraction}), 
\begin{align}
    \eta_{D}(\cD_p) = (1-p)^2,
\end{align}
where $\cD_p(\rho) \coloneqq (1-p)\rho + p I/d$.
This allows us to find examples for which the other contraction coefficients are strictly larger. Consider first the mutual information contraction coefficient for the depolarizing channel at $p=0.5$ and lower bound it by choosing $\rho_{CA}=0.5\,\Phi^++0.5\,\frac{I}{2}\otimes \frac{I}{2}$,
\begin{align}\label{eq:separation-between-eta-MI-and-eta-D}
    \eta_{\MI}(\cD_p) \geq \frac{I(C:B)}{I(C:A)} > 0.265 > 0.25 = \eta_{D}(\cD_p). 
\end{align}

Generally, separating the larger contraction coefficients is more challenging, because they include optimizations over unbounded reference systems. We provide one additional example. 
Choosing $\rho_{RA}=\Phi^+$ and $\sigma_{RA}=\frac34 \Phi^- + \frac14 \Phi^+$, then one can numerically calculate that for $p=0.5$,
\begin{align}
    \etaC_{D}(\cD_p) \geq \frac{D((\id\otimes\cD_p)(\rho_{RA})\|(\id\otimes\cD_p)(\sigma_{RA}))}{D(\rho_{RA}\|\sigma_{RA})} > 0.2881 > 0.25 = \eta_{D}(\cD_p). 
\end{align}

In a numerical search, we did not find any states achieving a comparable value with the mutual information contraction coefficient, leading us to conjecture that the complete contraction coefficient is indeed larger. Nevertheless, this does not imply that such a state does not exist.

Especially interesting would be if any of the following can be shown
\begin{align}
    \etaD_{D}(\cD_p) \overset{?}{>} \etaC_{D}(\cD_p) \overset{?}{>} \eta_{\MI}(\cD_p).
\end{align} 
All three quantities here suffer from a potentially unbounded reference system, making them challenging to compute.
For now, it remains an open problem to conclusively show whether these three coefficients are different or indeed the same.

\subsection{Partial orders}\label{Sec:Partial-Orders}

Partial orders quantify the data-processing decrease of information relative to a second channel. Such orders have been investigated in the quantum setting since the introduction of degradability in the context of quantum capacities~\cite{devetak2005capacity}, while they have a much longer history classically~\cite{korner1975comparison}. Later, quantum generalizations of the less noisy and more capable orders were also introduced~\cite{watanabe2012private}. Since then, they have noticeably shaped the discussion around capacities~\cite{tikku2020additivity,hirche2022bounding,leditzky2023platypus,smith2025additivity} and have also been studied in their own right~\cite{buscemi2016degradable,belzig2024reversetypedataprocessinginequality}.

While degradability inherently takes into account reference systems, entropy-based partial orders such as less noisy do not in general. 
One suggestion, made in~\cite{Hirche2022contraction},  was to define the \textit{completely less noisy order}: Denote by $\cN \geq_{\operatorname{cln}} \cM$ that 
\begin{align}
    I(U:BR) &\geq I(U:B'R) &\forall\rho_{RUA} \label{eq:cln-MI-def} \\
    \Leftrightarrow\quad  D(\cN(\rho_{RA})\|\cN(\sigma_{RA})) &\geq D(\cM(\rho_{RA})\|\cM(\sigma_{RA})) &\forall\rho_{RA},\sigma_{RA}, \label{eq:cln-rel-ent-def} 
\end{align}
where $U$ is a classical system and the equivalence between the two conditions was shown in~\cite[Proposition 2.3]{Hirche2022contraction}. Although usually the less noisy order is closely related to relative entropy contraction, it was understood that building a contraction coefficient naively from these quantities leads to a trivial result.
\begin{align}
    \sup_{\rho_{RUA}} \frac{I(U:BR)}{I(U:AR)} = \sup_{\rho_{RA},\sigma_{RA}} \frac{D(\cN(\rho_{RA})\|\cN(\sigma_{RA}))}{D(\rho_{RA}\|\sigma_{RA})} = 1. 
\end{align}

However, \cref{Thm:CMI-equivalence} suggests that the contraction coefficient corresponding to the completely less noisy ordering should instead be constructed from conditional relative entropies. To see this, note that it immediately follows from~\eqref{eq:cln-rel-ent-def} that an equivalent condition for the completely less noisy partial order is
\begin{align}
    D(\cN(\rho_{RA})\|\cN(\sigma_{RA}))-D(\rho_R\|\sigma_R) \geq D(\cM(\rho_{RA})\|\cM(\sigma_{RA}))-D(\rho_R\|\sigma_R) \; \forall  \rho_{RA},\sigma_{RA}  . 
\end{align}
Given item (iii) of \cref{Thm:CMI-equivalence}, it seems clear that this represents a better connection of the completely less noisy partial order to a contraction coefficient. This identification also motivates the following result, which is the analogue of the relation between the less noisy partial order and the relative entropy contraction coefficient \cite[Proposition 2.5]{Hirche2022contraction}. 
\begin{lemma}
\label{Lem:eta-cln}
    For a channel $\cN$, the following equality holds:
    \begin{align}
        \etaD_D(\cN) = 1-\sup
        \left\{\epsilon\in[0,1] \mid \cE_\epsilon \geq_{\operatorname{cln}} \cN \right\}, 
        \label{eq:etaD-to-erasure}
    \end{align}
    where $\cE_\epsilon$ is the erasure channel with erasure probability $\epsilon$. 
\end{lemma}

\begin{proof}
    To begin with, let us note that
    \begin{align}
        D\!\left(\cE_\epsilon(\rho_{RA})\| \cE_\epsilon(\sigma_{RA})\right) = (1-\epsilon) D(\rho_{RA}\|\sigma_{RA}) + \epsilon D(\rho_R\|\sigma_R), 
    \end{align}
    implying that
    \begin{equation}
        \frac{D\!\left(\cE_\epsilon(\rho_{RA})\| \cE_\epsilon(\sigma_{RA})\right) - D(\rho_R\|\sigma_R)}{D(\rho_{RA}\|\sigma_{RA}) - D(\rho_R\|\sigma_R)} = 1-\epsilon. 
    \end{equation}
    for every pair of states $\rho_{RA}$ and $\sigma_{RA}$ such that $D(\rho_{RA}\|\sigma_{RA}) - D(\rho_R\|\sigma_R)>0$.
    Hence, by definition,
    \begin{align}
        \etaD_D(\cE_\epsilon) = 1-\epsilon.
    \end{align}
    
    We first prove the $\leq$ direction in \eqref{eq:etaD-to-erasure}. For an $\epsilon^\star$ that optimizes the right-hand side, we have 
    \begin{align}
        \etaD_D(\cN) \leq \etaD_D(\cE_{\epsilon^\star}) = 1-\epsilon^\star, 
    \end{align}
    where the inequality follows from \eqref{eq:cln-rel-ent-def} and the assumption that $\cE_\epsilon \geq_{\operatorname{cln}} \cN$.
    
    We now show the other direction in \eqref{eq:etaD-to-erasure}. Assume $\etaD_D(\cN)=1-\epsilon$. Then for all $\rho_{RA}$, $\sigma_{RA}$ we have, 
    \begin{align}
        D((\id\otimes\cN)(\rho_{RA})\|(\id\otimes\cN)(\sigma_{RA})) - D(\rho_R\|\sigma_R) \leq (1-\epsilon) \left[ D(\rho_{RA}\|\sigma_{RA}) - D(\rho_{R}\|\sigma_{R})\right],
    \end{align}
    which is equivalent to
    \begin{align}
        D((\id\otimes\cN)(\rho_{RA})\|(\id\otimes\cN)(\sigma_{RA})) &\leq (1-\epsilon)  D(\rho_{RA}\|\sigma_{RA}) + \epsilon D(\rho_{R}\|\sigma_{R}) \\
        &= D(\cE_\epsilon(\rho_{RA})\| \cE_\epsilon(\sigma_{RA})).
    \end{align}
    This in turn implies that $\cE_\epsilon \geq_{\operatorname{cln}} \cN$ and hence $\etaD_D(\cN) \geq 1-\sup\{\epsilon \in [0,1] \mid \cE_\epsilon \geq_{\operatorname{cln}} \cN \}$. 
\end{proof}

\cref{Lem:eta-cln} might also provide some intuition regarding the conditional contraction coefficient in the sense that it answers the question of how much information is erased by the channel. Consider the erasure channel with erasure probability $\epsilon$ and any divergence $\DD$ satisfying the direct-sum property. Then,
\begin{align}
    \DD(\cE_\epsilon(\rho_{RA})\|\cE_\epsilon(\sigma_{RA})) =  (1-\epsilon) \DD(\rho_{RA}\|\sigma_{RA}) + \epsilon \DD(\rho_R\|\sigma_R). 
\end{align}
Setting the erasure probability  $\epsilon=1-\etaD_{\DD}(\cN)$, the conditional contraction coefficient quantifies the decrease of information at the output system, while maintaining that in the reference system; i.e., 
\begin{equation}
    \DD(\cN(\rho_{RA}) \| \cN(\sigma_{RA})) \leq \DD(\cE_\epsilon(\rho_{RA})\|\cE_\epsilon(\sigma_{RA})).
\end{equation}
This follows because
\begin{align}
    \DD(\cN(\rho_{RA}) \| \cN(\sigma_{RA})) &= \DD(\cN(\rho_{RA}) \| \cN(\sigma_{RA})) - \DD(\rho_R\|\sigma_R) + \DD(\rho_R\|\sigma_R) \\
    & \leq \etaD_{\DD}(\cN) (\DD(\rho_{RA} \| \sigma_{RA}) -\DD(\rho_R\|\sigma_R)) + \DD(\rho_R\|\sigma_R) \\
    &= \etaD_{\DD}(\cN) \DD(\rho_{RA} \| \sigma_{RA})  + (1-\etaD_{\DD}(\cN))\DD(\rho_R\|\sigma_R) \\
    &= \DD(\cE_\epsilon(\rho_{RA})\|\cE_\epsilon(\sigma_{RA})).
\end{align}

From~\cref{Lem:eta-cln} we also obtain a direct connection to the Doeblin coefficient $\alpha_+(\cdot)$ as defined in~\eqref{Eq:Doeblin+}, showing, as a special case of \cref{Lem:a-Doeblin}, that
\begin{equation}
    \etaD_D(\cN) \leq 1-\alpha_+(\cN),
\end{equation}
because degradability always implies completely less noisy (see \cite[Eq.~(I.4) and Prop.~III.3]{hirche2024quantum} in this context).  

\begin{remark}
    Note that one could also define a variant of the completely less noisy partial order where $\rho_R=\sigma_R$. This would then allow for connections similar to the above relating the partial order to the complete contraction coefficient. 
\end{remark}

\begin{remark}
    Another direct consequence of~\cref{Thm:CMI-equivalence} is that in the (conditional) mutual information definition of completely less noisy, it is sufficient to consider a classical system $U$ with $|U|=2$. 
\end{remark}

\subsection{SDPI constants}

So far, we have mostly discussed contraction coefficients. However, it is worth noting that much of the above also holds for the generalized SDPI constants as defined in~\cref{Def:Gen-SDPI}, with regards to the relative entropy. To make that explicit, define the conditional mutual information SDPI constants as follows: 

\begin{definition}
Given a channel $\cN\colon A \to B$ and a state $\sigma_{RA}$, we define
    \begin{align}
    \eta_{\CMI}(\cN,\sigma_{RA}) &\coloneqq \sup_{\substack{\rho_{RCA}, \\ \rho_{RA}=\sigma_{RA}}} \frac{I(C:B|R)}{I(C:A|R)}, 
    \\
    \eta_{\operatorname{cqCMI}}(\cN,\sigma_{RA}) &\coloneqq \sup_{\substack{\rho_{RUA},\\ \rho_{RA}=\sigma_{RA}}} \frac{I(U:B|R)}{I(U:A|R)}, 
    \\
    \eta_{\operatorname{2cqCMI}}(\cN,\sigma_{RA}) &\coloneqq \sup_{\substack{\rho_{RUA}, \\ \rho_{RA}=\sigma_{RA}, \\|U|=2}} \frac{I(U:B|R)}{I(U:A|R)}. 
\end{align}
\end{definition}

We then immediately have the following result connecting most of the established SDPI constants.

\begin{corollary}\label{Cor:CMI-amo-SDPI-equal}
    For a channel $\cN\colon A\to B$, the following equalities hold:
    \begin{align}
               \etaD_D(\cN,\sigma_{RA}) =  \eta_{\operatorname{2cqCMI}}(\cN,\sigma_{RA}) = \eta_{\operatorname{cqCMI}}(\cN,\sigma_{RA}).
    \end{align}
\end{corollary}

\begin{proof}
    The proof follows from~\cref{Thm:CMI-equivalence}.  
\end{proof}

However, in the fully quantum CMI case one actually needs to be a bit more careful in generalizing the equivalence. To that end we introduce stabilized versions of the complete and conditional SDPI constants.

\begin{definition}
    For a generalized divergence $\DD$, a quantum channel $\cN\colon A\to B$ and a quantum state $\sigma_{RA}$, define the stabilized complete SDPI constant by
\begin{align}
    \eta_{\DD}^{\mathrm{c,stab}}(\cN,\sigma_{RA})
    \coloneqq
    \sup_{\substack{S,\,\tau_S}}
    \etaC_\DD(\cN,\tau_S\otimes\sigma_{RA}),
    \label{eq:stabilized-complete-SDPI}
\end{align}
where $S$ ranges over all finite-dimensional quantum systems and
$RS$ is regarded as the reference system on the right-hand side. We also define the stabilized conditional SDPI constant,
\begin{align}
    \eta_\DD^{\Delta,\operatorname{stab}}(\cN,\sigma_{RA})
    \coloneqq
    \sup_{\substack{S,\,\tau_S}}
    \etaD_\DD(\cN,\tau_S\otimes\sigma_{RA}).
    \label{eq:stabilized-conditional-SDPI}
\end{align}
\end{definition}
Note that for the relative entropy \cref{eq:stabilized-conditional-SDPI} can be simplified by restricting the optimization to the maximally mixed state as a consequence of the chain rule for the relative entropy, 
\begin{align}
    \eta_\DD^{\Delta,\operatorname{stab}}(\cN,\sigma_{RA})
    =
    \sup_{S}
    \etaD_\DD(\cN,\pi_S\otimes\sigma_{RA}),
\end{align}
where $\pi_S\coloneqq I_S/|S|$ is the maximally mixed state of dimension $|S|$.
This makes it purely an optimization over the size of the system $S$. 
We are now in a position to prove the remaining missing equalities.

\begin{proposition}[CMI SDPI as a stabilized conditional SDPI]
\label{Prop:CMI-stabilized-SDPI}
Let $\cN\colon A\to B$ be a quantum channel, and let $\sigma_{RA}$ be a
quantum state. 
Then
\begin{align}
    \eta_{\CMI}(\cN,\sigma_{RA})
    =
    \eta_D^{\Delta,\operatorname{stab}}(\cN,\sigma_{RA}).
    \label{eq:CMI-stabilized-SDPI-equality}
\end{align}
\end{proposition}

\begin{proof}
First, suppose that no extension $\rho_{RCA}$ of $\sigma_{RA}$ has
$I(C:A|R)_\rho>0$. Then $\eta_{\CMI}(\cN,\sigma_{RA})=0$.
Moreover, $\eta_D^{\Delta,\operatorname{stab}}(\cN,\sigma_{RA})=0$:
otherwise, the fixed-state cqCMI characterization would give some
$\omega_{URSA}$ with $\omega_{RSA}=\pi_S\otimes\sigma_{RA}$ and
$I(U:A|RS)_\omega>0$, and hence
\begin{align}
I(US:A|R)_\omega
= I(S:A|R)_\omega+I(U:A|RS)_\omega
= I(U:A|RS)_\omega>0,
\end{align}
contradicting the assumption. We may therefore assume that an extension
with strictly positive conditional mutual information exists.

We continue by proving the $\leq$ direction. 
First, consider an
arbitrary state $\rho_{RCA}$ satisfying
\begin{align}
    \rho_{RA}=\sigma_{RA}
    \qquad\text{and}\qquad
    I(C:A|R)_\rho>0.
\end{align}
The relative-entropy chain rule \eqref{eq:rel-ent-chain} gives
\begin{align}
    &D\!\left(
        \rho_{RCA}
        \middle\|
        \pi_C\otimes\sigma_{RA}
    \right)
    -
    D\!\left(
        \rho_{RC}
        \middle\|
        \pi_C\otimes\sigma_R
    \right)
    \nonumber\\
    &\qquad=
    D\!\left(
        \rho_{RCA}
        \middle\|
        \rho_C\otimes\sigma_{RA}
    \right)
    -
    D\!\left(
        \rho_{RC}
        \middle\|
        \rho_C\otimes\sigma_R
    \right)
    \\
    &\qquad=
    I(C:A|R)_\rho.
    \label{eq:CMI-stabilized-input-identity}
\end{align}
Similarly, after applying $\cN$ to $A$,
\begin{align}
    &D\!\left(
        \cN(\rho_{RCA})
        \middle\|
        \pi_C\otimes\cN(\sigma_{RA})
    \right)
    -
    D\!\left(
        \rho_{RC}
        \middle\|
        \pi_C\otimes\sigma_R
    \right) =
    I(C:B|R)_{\cN(\rho)}.
    \label{eq:CMI-stabilized-output-identity}
\end{align}
Therefore, by the definition of the conditional SDPI constant with
fixed second state $\pi_C\otimes\sigma_{RA}$,
\begin{align}
    I(C:B|R)_{\cN(\rho)}
    \leq
    \etaD_D(\cN,\pi_C\otimes\sigma_{RA})
    I(C:A|R)_\rho.
\end{align}
Consequently,
\begin{align}
    \frac{
        I(C:B|R)_{\cN(\rho)}
    }{
        I(C:A|R)_\rho
    }
    &\leq
    \etaD_D(\cN,\pi_C\otimes\sigma_{RA}) \leq
    \eta_D^{\Delta,\operatorname{stab}}
        (\cN,\sigma_{RA}).
\end{align}
Taking the supremum over $C$ and over all admissible states
$\rho_{RCA}$ proves
\begin{align}
    \eta_{\CMI}(\cN,\sigma_{RA})
    \leq
    \eta_D^{\Delta,\operatorname{stab}}(\cN,\sigma_{RA}).
    \label{eq:CMI-stabilized-first-direction}
\end{align}

For the converse inequality, fix an arbitrary system $S$. By applying \Cref{Cor:CMI-amo-SDPI-equal} with the joint space $RS$ acting as the reference system and expanding the definition of $\eta_{\operatorname{cqCMI}}$, we have
\begin{align}
    \etaD_D(\cN,\pi_S\otimes\sigma_{RA})
    =
    \sup_{\substack{\omega_{URSA},\\
                    U\ \mathrm{classical},\\
                    \omega_{RSA}=\pi_S\otimes\sigma_{RA},\\
                    I(U:A|RS)_\omega>0}}
    \frac{
        I(U:B|RS)_{\cN(\omega)}
    }{
        I(U:A|RS)_\omega
    }.
    \label{eq:stabilized-SDPI-cqCMI-expression}
\end{align}
The classical system $U$ may, in fact, be restricted to be binary.

Consider an arbitrary feasible state $\omega_{URSA}$ in~\eqref{eq:stabilized-SDPI-cqCMI-expression}, for which $I\!\left(S:A|R\right)_\omega=0$ by the product restriction.
Using the chain rule for conditional mutual information \eqref{eq:CMI-to-CMI-chain-rule}, we obtain
\begin{align}
    I(US:A|R)_\omega
    &=
    I(S:A|R)_\omega
    +
    I(U:A|RS)_\omega\\
    &=
    I(U:A|RS)_\omega.
    \label{eq:CMI-chain-input-stabilized}
\end{align}
Moreover, we have $\omega_{RA}=\sigma_{RA}$.
Thus, the joint system $US$ is an admissible quantum auxiliary system
in the optimization defining
$\eta_{\CMI}(\cN,\sigma_{RA})$. It follows that
\begin{align}
    \frac{
        I(U:B|RS)_{\cN(\omega)}
    }{
        I(U:A|RS)_\omega
    }
    &=
    \frac{
        I(US:B|R)_{\cN(\omega)}
    }{
        I(US:A|R)_\omega
    } \leq
    \eta_{\CMI}(\cN,\sigma_{RA}).
\end{align}
Taking the supremum over the states in
\eqref{eq:stabilized-SDPI-cqCMI-expression} gives
\begin{align}
    \etaD_D(\cN,\pi_S\otimes\sigma_{RA})
    \leq
    \eta_{\CMI}(\cN,\sigma_{RA}).
\end{align}
Since this holds for every finite-dimensional system $S$, taking the
supremum over $S$ yields
\begin{align}
    \eta_D^{\Delta,\operatorname{stab}}(\cN,\sigma_{RA})
    \leq
    \eta_{\CMI}(\cN,\sigma_{RA}).
    \label{eq:CMI-stabilized-second-direction}
\end{align}
Combining~\eqref{eq:CMI-stabilized-first-direction}
and~\eqref{eq:CMI-stabilized-second-direction} proves
\eqref{eq:CMI-stabilized-SDPI-equality}.
\end{proof}

Similarly, we can define the MIR (mutual information with reference) SDPI constants, 
\begin{definition}
Given a channel $\cN\colon A \to B$ and a state $\sigma_{RA}$, we define
\begin{align}
    \eta_{\operatorname{MIR}}(\cN,\sigma_{RA}) &\coloneqq  \sup_{\substack{\rho_{RCA},\\ \rho_{RA}=\sigma_{RA},\\ \rho_{CR}=\rho_C\otimes\rho_R}} \frac{I(C:BR)}{I(C:AR)}, \\
    \eta_{\operatorname{cqMIR}}(\cN,\sigma_{RA}) &\coloneqq  \sup_{\substack{\rho_{RUA},\\ \rho_{RA}=\sigma_{RA},\\ \rho^x_R=\rho^y_R}} \frac{I(U:BR)}{I(U:AR)}, \\
    \eta_{\operatorname{2cqMIR}}(\cN,\sigma_{RA}) &\coloneqq  \sup_{\substack{\rho_{RUA},\\ \rho_{RA}=\sigma_{RA},\\ \rho^0_R=\rho^1_R, \\ |U|=2}} \frac{I(U:BR)}{I(U:AR)}.
\end{align}
\end{definition}
With this we can show the following.
\begin{corollary}\label{Cor:MIR-equals-complete-SDPI}
    For a channel $\cN$, the following equalities hold:
    \begin{align}
               \etaC_D(\cN,\sigma_{RA}) =  \eta_{\operatorname{2cqMIR}}(\cN,\sigma_{RA}) = \eta_{\operatorname{cqMIR}}(\cN,\sigma_{RA}).
    \end{align}
\end{corollary}
\begin{proof}
    The proof follows just like the one of~\cref{Cor:MIR-equals-complete}.
\end{proof}
Similar to the CMI case, considering the fully quantum setting requires a modification.
\begin{proposition}[MIR SDPI as a stabilized complete SDPI]
\label{Prop:MIR-stabilized-SDPI}
Let $\cN\colon A\to B$ be a quantum channel and let $\sigma_{RA}$ be a
quantum state. 
Then
\begin{align}
    \eta_{\operatorname{MIR}}(\cN,\sigma_{RA})
    =
    \eta_{D}^{\mathrm{c,stab}}(\cN,\sigma_{RA}).
    \label{eq:MIR-stabilized-SDPI-equality}
\end{align}
\end{proposition}
\begin{proof}
We follow the same strategy as in
\cref{Prop:CMI-stabilized-SDPI}, highlighting only the differences
caused by the equal-reference-marginal constraint. The case where no admissible state with strictly positive denominator exists is treated the same way as in \cref{Prop:CMI-stabilized-SDPI}. 

First, let $\rho_{RCA}$ be admissible in the definition of
$\eta_{\operatorname{MIR}}(\cN,\sigma_{RA})$, so that
\begin{align}
    \rho_{RA}=\sigma_{RA},
    \qquad
    \rho_{CR}=\rho_C\otimes\sigma_R.
\end{align}
The two states $\rho_{RCA}$ and $\rho_C\otimes\sigma_{RA}$ have the same marginal $\rho_C\otimes\sigma_R$ on the reference
system $CR$. Moreover, as $\rho_{RA} = \sigma_{RA}$,
\begin{align}
    D(\rho_{RCA}\|\rho_C\otimes\sigma_{RA})
    &=
    I(C:AR)_\rho,\\
    D\!\left(
        \cN(\rho_{RCA})
        \middle\|
        \rho_C\otimes\cN(\sigma_{RA})
    \right)
    &=
    I(C:BR)_{\cN(\rho)}.
\end{align}
It follows from the definition of the complete SDPI constant that
\begin{align}
    \frac{I(C:BR)_{\cN(\rho)}}{I(C:AR)_\rho}
    &\leq
    \etaC_D(\cN,\rho_C\otimes\sigma_{RA})\\
    &\leq
    \eta_{D}^{\mathrm{c,stab}}(\cN,\sigma_{RA}).
\end{align}
Taking the supremum over all admissible $\rho_{RCA}$ gives
\begin{align}
    \eta_{\operatorname{MIR}}(\cN,\sigma_{RA})
    \leq
    \eta_{D}^{\mathrm{c,stab}}(\cN,\sigma_{RA}).
    \label{eq:MIR-stabilized-first-direction}
\end{align}

For the converse, fix a system $S$ and a state
$\tau_S$. By the fixed-state cqMIR characterization of the complete
SDPI constant, applied with reference system $RS$ and fixed state
$\tau_S\otimes\sigma_{RA}$, we have
\begin{align}
    \etaC_D(\cN,\tau_S\otimes\sigma_{RA})
    =
    \sup_{\omega_{URSA}}
    \frac{
        I(U:BRS)_{\cN(\omega)}
    }{
        I(U:ARS)_\omega
    },
    \label{eq:stabilized-complete-cqMIR-expression}
\end{align}
where $U$ is classical and the supremum is over c-q states
\begin{align}
    \omega_{URSA}
    =
    \sum_u p(u)|u\rangle\!\langle u|_U\otimes\omega^u_{RSA}
\end{align}
such that
\begin{align}
    \omega_{RSA}
    &=
    \tau_S\otimes\sigma_{RA},\\
    \omega^u_{RS}
    &=
    \tau_S\otimes\sigma_R
    \qquad\forall u,
\end{align}
and with $I(U:ARS)_\omega>0$. The classical system $U$ may be
restricted to be binary.
These marginal conditions imply
\begin{align}
    \omega_{USR}
    &=
    \omega_U\otimes\tau_{S}\otimes\sigma_R.
\end{align}
In particular,
\begin{align}
    I(U:S)_\omega=0,
    \qquad
    I(S:AR)_\omega=0.
\end{align}
Using these equalities and the mutual-information chain rule therefore gives
\begin{align}
    I(U:ARS)_\omega
    &=
    I(U:S)_\omega+I(U:AR|S)_\omega\\
    &=
    I(U:AR|S)_\omega\\
    &=
    I(S:AR)_\omega+I(U:AR|S)_\omega\\
    &=
    I(US:AR)_\omega,
    \label{eq:MIR-stabilized-input-chain-rule} 
\end{align}
and similarly,
\begin{align}
    I(U:BRS)_{\cN(\omega)}
    =
    I(US:BR)_{\cN(\omega)}.
    \label{eq:MIR-stabilized-output-chain-rule}
\end{align}

Moreover, $\omega_{RA}=\sigma_{RA}$ and
$\omega_{USR}=\omega_{US}\otimes\sigma_R$, so the joint system $US$ 
is an admissible quantum auxiliary system in the definition of
$\eta_{\operatorname{MIR}}(\cN,\sigma_{RA})$. Consequently,
\begin{align}
    \frac{
        I(U:BRS)_{\cN(\omega)}
    }{
        I(U:ARS)_\omega
    }
    &=
    \frac{
        I(US:BR)_{\cN(\omega)}
    }{
        I(US:AR)_\omega
    }\\
    &\leq
    \eta_{\operatorname{MIR}}(\cN,\sigma_{RA}).
\end{align}
Taking the supremum in
\eqref{eq:stabilized-complete-cqMIR-expression}, followed by the
supremum over $S$ and $\tau_S$, yields
\begin{align}
    \eta_{D}^{\mathrm{c,stab}}(\cN,\sigma_{RA})
    \leq
    \eta_{\operatorname{MIR}}(\cN,\sigma_{RA}).
    \label{eq:MIR-stabilized-second-direction}
\end{align}
Combining
\eqref{eq:MIR-stabilized-first-direction}
and
\eqref{eq:MIR-stabilized-second-direction}
proves the claim.
\end{proof}
These findings illustrate that, while many results transfer directly to the case of SDPI constants, others need further care. Intuitively, the stabilization considered here is necessary because fixing $\sigma_{RA}$ also fixes the size of the reference system, removing much of the flexibility present in the contraction coefficient case. 

We are now left with proving that the MI SDPI constant and $\eta^p(\cN,\sigma_A)$, as defined in~\eqref{SDPI-product}, are equal. 

\begin{lemma}[MI SDPI and product-reference SDPI]
\label{Lem:MI-product-SDPI}
Let $\cN\colon A\to B$ be a quantum channel, and let $\sigma_A$ be a
quantum state. Define the mutual-information SDPI constant at
$\sigma_A$ as
\begin{align}
    \eta_{\MI}(\cN,\sigma_A)
    \coloneqq
    \sup_{\substack{C,\,\rho_{CA},\\
                    \rho_A=\sigma_A,\\
                    I(C:A)_\rho>0}}
    \frac{
        I(C:B)_{\cN(\rho)}
    }{
        I(C:A)_\rho
    },
    \label{eq:MI-SDPI-fixed-state}
\end{align}
where the supremum is over every finite-dimensional quantum system
$C$. Then
\begin{align}
    \eta_{\MI}(\cN,\sigma_A)
    =
    \eta_D^p(\cN,\sigma_A).
    \label{eq:MI-SDPI-equals-product-reference-SDPI}
\end{align}
\end{lemma}

\begin{proof}
This follows from the stabilized CMI characterization in
\cref{Prop:CMI-stabilized-SDPI} by taking the original reference
system to be trivial. Indeed,
\begin{align}
    \eta_{\MI}(\cN,\sigma_A)
    =
    \sup_S
    \etaD_D(\cN,\pi_S\otimes\sigma_A),
    \label{eq:MI-SDPI-as-stabilized-conditional}
\end{align}
where $\pi_S\coloneqq I_S/|S|$ and $S$ is regarded as the reference
system on the right-hand side.

For a fixed system $S$ and an arbitrary state $\rho_{SA}$, the
relative-entropy chain rule gives
\begin{align}
    D(\rho_{SA}\|\pi_S\otimes\sigma_A)
    =
    D(\rho_{SA}\|\rho_S\otimes\sigma_A)
    +
    D(\rho_S\|\pi_S).
\end{align}
Consequently,
\begin{align}
    D(\rho_{SA}\|\pi_S\otimes\sigma_A)
    -
    D(\rho_S\|\pi_S)
    =
    D(\rho_{SA}\|\rho_S\otimes\sigma_A).
    \label{eq:conditional-product-reference-input}
\end{align}
It follows that
\begin{align}
    \etaD_D(\cN,\pi_S\otimes\sigma_A)
    =
    \sup_{\substack{\rho_{SA},\\
                    0<
                    D(\rho_{SA}\|\rho_S\otimes\sigma_A)
                    <\infty}}
    \frac{
        D\!\left(
            (\id_S\otimes\cN)(\rho_{SA})
            \middle\|
            \rho_S\otimes\cN(\sigma_A)
        \right)
    }{
        D(\rho_{SA}\|\rho_S\otimes\sigma_A)
    }.
    \label{eq:conditional-SDPI-is-product-reference-fixed-S}
\end{align}
Taking the supremum over all finite-dimensional systems $S$, the
right-hand side becomes precisely the product-reference SDPI constant
$\eta_D^p(\cN,\sigma_A)$. Combining this observation with
\eqref{eq:MI-SDPI-as-stabilized-conditional} proves
\eqref{eq:MI-SDPI-equals-product-reference-SDPI}.
\end{proof}
This completes our exploration of SDPI constants.

\subsection{Expansion coefficients and relative contraction coefficients}

Essentially every contraction coefficient above has an associated expansion coefficient by simply replacing the supremum with an infimum. Also, most of the results above would then transfer to this setting, with all inequalities changing direction, for example in~\cref{Cor:Summary1}. Note however, that this would make $\check\eta_D(\cN)$ the largest such coefficient, and Ref.~\cite{belzig2024reversetypedataprocessinginequality} showed that, whenever the output dimension is not larger than the input dimension, either $\cN$ is a unitary channel or the contraction coefficient is equal to zero. Hence, for this very large class of channels, expansion coefficients defined using the relative entropy are generally not useful, with or without reference systems. The result from~\cite{belzig2024reversetypedataprocessinginequality} was later generalized to all monotone $f$-divergences in~\cite{iyer2025quantum}. We will discuss such divergences in the next section. 

More interesting in this context are the relative contraction coefficients.
\begin{definition}
For two quantum channels $\cN\colon A\to B'$ and $\cM\colon A\to B$, define
\begin{align}
    \eta_{\CMI}(\cN,\cM) &\coloneqq \sup_{\rho_{RCA}} \frac{I(C:B'|R)}{I(C:B|R)}, 
    \\
    \eta_{\operatorname{cqCMI}}(\cN,\cM) &\coloneqq \sup_{\rho_{RUA}} \frac{I(U:B'|R)}{I(U:B|R)}, 
    \\
    \eta_{\operatorname{2cqCMI}}(\cN,\cM) &\coloneqq \sup_{\substack{\rho_{RUA},\\|U|=2}} \frac{I(U:B'|R)}{I(U:B|R)}. 
\end{align}    
\end{definition}

Then we have the following.

\begin{corollary}\label{Cor:CMI-amo-relative-equal}
    For channels $\cN$ and $\cM$, we have
    \begin{align}
               \eta_{\operatorname{2cqCMI}}(\cN,\cM) = \eta_{\operatorname{cqCMI}}(\cN,\cM) = \eta_{\CMI}(\cN,\cM) = \etaD_D(\cN,\cM).
    \end{align}
\end{corollary}

\begin{proof}
    The proof is identical to that of~\cref{Cor:CMI-amo-equal}.  
\end{proof}

Similarly, consider the following.
\begin{definition}
For quantum channels $\cN\colon A\to B'$ and $\cM\colon A\to B$, define
    \begin{align}
    \eta_{\operatorname{MIR}}(\cN,\cM) & \coloneq\sup_{\substack{\rho_{RCA},\\ \rho_{CR}=\rho_C\otimes\rho_R}} \frac{I(C:B'R)}{I(C:BR)}, \\
    \eta_{\operatorname{cqMIR}}(\cN,\cM) & \coloneq \sup_{\substack{\rho_{RUA},\\ \rho^x_R=\rho^y_R}} \frac{I(U:B'R)}{I(U:BR)}, \\
    \eta_{\operatorname{2cqMIR}}(\cN,\cM) &\coloneqq  \sup_{\substack{\rho_{RUA},\\ \rho^0_R=\rho^1_R, \\ |U|=2}} \frac{I(U:B'R)}{I(U:BR)}.
\end{align}
\end{definition}

Then the following holds.
\begin{corollary}\label{Cor:MIR-equals-complete-rel-cc}
    For a channel $\cN$, we have
    \begin{align}
               \eta_{\operatorname{2cqMIR}}(\cN,\cM) = \eta_{\operatorname{cqMIR}}(\cN,\cM) = \eta_{\operatorname{MIR}}(\cN,\cM)= \etaC_D(\cN,\cM).
    \end{align}
\end{corollary}

\begin{proof}
    The proof follows just like the one of~\cref{Cor:MIR-equals-complete}.
\end{proof}

All these results then naturally transfer to relative expansion coefficients via~\eqref{Eq:rel-ep-con-c} and~\eqref{Eq:rel-ep-con-a}.

\subsection{Beyond the relative entropy}

As mentioned before, while we have only discussed the relative entropy in this section so far, many of the results above actually hold for much larger classes of divergences. Reexamining the proofs of~\cref{Thm:CMI-equivalence} and~\cref{Thm:MIR-equivalence}, we note that they really only require the following two properties, 
\begin{enumerate}
    \item (Direct-sum property) For a classical-quantum state $\rho_{UA}$, 
    \begin{align}
        \DD(\rho_{UA}\|\rho_U\otimes\rho_A) = \sum_x p(x) \DD(\rho_A^x \| \rho_A), \label{eq:prop_direct_sum}
    \end{align}
    \item (Vanishing first derivative) For two quantum states $\rho$ and $\sigma$, such that $\supp(\rho)\subseteq\supp(\sigma)$, 
    \begin{align}
        \left.\frac{\partial}{\partial\lambda} \DD(\lambda\rho+(1-\lambda)\sigma\|\sigma)  \right|_{\lambda=0^{+}} = 0. \label{eq:prop_vanishing_gradient}
    \end{align}
\end{enumerate}
Particular classes of divergences that follow these properties include several families of $f$-divergences, including the Petz $f$-divergence~\cite{petz1985quasi,petz1986quasi}: 
\begin{align}
    \bar D_f(\rho\|\sigma) \coloneqq  \tr\!\left[\sigma^{\frac12}f(\Delta_{\rho,\sigma})(\sigma^{\frac12})\right], 
\end{align}
where $\Delta_{\rho,\sigma}(X)\coloneqq \rho X\sigma^{-1}$ is the relative modular operator, the Hirche--Tomamichel $f$-divergence~\cite{hirche2024quantumDivergences}: 
\begin{align}
    D_f(\rho\|\sigma) \coloneqq  \int_1^\infty f''(\gamma) E_\gamma(\rho\|\sigma) + \frac1{\gamma^3} f''(\gamma^{-1}) E_\gamma(\sigma\|\rho) \operatorname{d}\gamma, 
\end{align}
and the maximal $f$-divergence~\cite{matsumoto2018new}: 
\begin{align}
    \hat D_f(\rho\|\sigma) \coloneqq  \tr\!\left[\sigma f(\sigma^{-\frac12}\rho\sigma^{-\frac12})\right]. 
\end{align}
In all cases the proof of the direct-sum property given in~\eqref{eq:prop_direct_sum} is rather direct, using that $\rho_{UA}$ is block-diagonal. The second property given in~\eqref{eq:prop_vanishing_gradient} was proven for the first and third divergences above in~\cite[Proof of Proposition 6.2]{Hirche2022contraction} and for the second divergence above in~\cite[Proof of Theorem 2.8]{hirche2024quantumDivergences}; for more details see also~\cite[Theorem 2.2]{beigi2025some}. It was also noted in~\cite[Remark 6.3]{Hirche2022contraction} that the properties also hold for the optimized $f$-divergence of~\cite{wilde2018optimized} by extending the result from the Petz $f$-divergence. \cref{Lem:partial-equiv} in~\cref{app:prop-rel-ent} clarifies that the second property above implies the property used in the proofs of \cref{Thm:CMI-equivalence} and \cref{Thm:MIR-equivalence}.

For the remainder of this section, we will only consider divergences that satisfy these two properties. 
Noteworthy divergences that fall into this class include essentially all $\alpha$-Hellinger divergences (see, e.g., \cite[Section~3]{hirche2024quantumDivergences}),
but the R\'enyi relative entropies do not. 

We can define mutual information-like quantities as follows: 
\begin{align}
    \II(A:B) &\coloneqq \DD(\rho_{AB}\|\rho_A\otimes\rho_{B}), \\
    \II(A:B|C) &\coloneqq \II(A:BC) - \II(A:C) \\
    &= \DD(\rho_{ABC}\|\rho_A\otimes\rho_{BC}) - \DD(\rho_{AC}\|\rho_A\otimes\rho_{C}). 
\end{align}
With these, we can define contraction coefficients similarly to all of what we discussed above and then many results discussed earlier in this section translate directly. As our main example of this, we provide a generalization of~\cref{Thm:CMI-equivalence}.

\begin{theorem}
\label{Thm:gen-CMI-equivalence}
    Let $\DD$ be a divergence that satisfies \eqref{eq:prop_direct_sum} and \eqref{eq:prop_vanishing_gradient}, let $\cN\colon A\to B'$ and $\cM\colon A\to B$ be quantum channels, $\sigma_{RA}$ a quantum state, and $\eta\geq0$. Then the following are equivalent:
\begin{enumerate}[label=(\roman*)]
    \item For all c-q states $\rho_{RUA}$ with marginal $\sigma_{RA}$, where U is a classical system of size $|U|=2$, 
    \begin{align}
        \eta \II(U:B|R) \geq \II(U:B'|R).
    \end{align}
    \item For all c-q states $\rho_{RUA}$ with marginal $\sigma_{RA}$, where U is an arbitrary classical system, 
    \begin{align}
        \eta \II(U:B|R) \geq \II(U:B'|R).
    \end{align}
\item For all quantum states $\rho_{RA}$ with $\operatorname{supp}(\rho_{RA})\subseteq \operatorname{supp}(\sigma_{RA})$, 
\begin{align}
    \eta \left[ \DD(\cM(\rho_{RA})\|\cM(\sigma_{RA})) - \DD(\rho_R\|\sigma_R) \right] \geq \DD(\cN(\rho_{RA})\|\cN(\sigma_{RA})) - \DD(\rho_R\|\sigma_R).
\end{align}
\end{enumerate}
\end{theorem}

\begin{proof}
    The proof is the same as that for~\cref{Thm:CMI-equivalence}, while taking into account~\cref{Lem:partial-equiv} in~\cref{app:prop-rel-ent}.
\end{proof}

As a result, whenever $\DD$ satisfies \eqref{eq:prop_direct_sum} and \eqref{eq:prop_vanishing_gradient}, we can give alternative expressions for the contraction coefficients akin to those for the relative entropy.
\begin{corollary}
\label{Cor:gen-CMI-contraction-equivalence}
Let $\DD$ be a divergence satisfying
\eqref{eq:prop_direct_sum} and
\eqref{eq:prop_vanishing_gradient}. Then, for every quantum channel
$\cN\colon A\to B$,
\begin{align}
    \etaD_{\DD}(\cN)
    &=
    \sup_{\substack{\rho_{RUA},\\
                    U\ \mathrm{classical},\ |U|=2}}
    \frac{
        \II(U:B|R)_{\cN(\rho)}
    }{
        \II(U:A|R)_{\rho}
    }
    \label{eq:gen-binary-cqCMI-contraction}\\
    &=
    \sup_{\substack{\rho_{RUA},\\
                    U\ \mathrm{classical}}}
    \frac{
        \II(U:B|R)_{\cN(\rho)}
    }{
        \II(U:A|R)_{\rho}
    }
    \label{eq:gen-cqCMI-contraction}\\
    &=
    \sup_{\rho_{RCA}}
    \frac{
        \II(C:B|R)_{\cN(\rho)}
    }{
        \II(C:A|R)_{\rho}
    },
    \label{eq:gen-quantum-CMI-contraction}
\end{align}
where all suprema are restricted to states for which the denominator
is strictly positive and finite.
\end{corollary}

\begin{proof}
The first two equalities follow from
\cref{Thm:gen-CMI-equivalence} after optimizing over the fixed state
$\sigma_{RA}$.

Since a classical system is a special case of a quantum system, the
last supremum is at least $\etaD_{\DD}(\cN)$. To prove the converse,
consider an arbitrary state $\rho_{RCA}$ and regard $RC$ as the
reference system in the definition of $\etaD_{\DD}(\cN)$. Choose the
two states
\begin{align}
    \rho_{RCA}
    \qquad\text{and}\qquad
    \rho_C\otimes\rho_{RA}.
\end{align}
Then
\begin{align}
    &\DD(\rho_{RCA}\|\rho_C\otimes\rho_{RA})
    -
    \DD(\rho_{RC}\|\rho_C\otimes\rho_R)
    =
    \II(C:A|R)_\rho,
\end{align}
while, after applying the channel,
\begin{align}
    &\DD\!\left(
        \cN(\rho_{RCA})
        \middle\|
        \rho_C\otimes\cN(\rho_{RA})
    \right)
    -
    \DD(\rho_{RC}\|\rho_C\otimes\rho_R)
    =
    \II(C:B|R)_{\cN(\rho)}.
\end{align}
Thus every ratio in
\eqref{eq:gen-quantum-CMI-contraction} is  bounded from above by
$\etaD_{\DD}(\cN)$, which proves the last equality.
\end{proof}
Analogous results hold for the expansion coefficient and the relative contraction coefficients. Similarly,  generalizations of~\cref{Thm:MIR-equivalence} and its consequences for the complete contraction coefficient hold as well.

\section{On tensorization of SDPI constants}

\label{Sec:Tensorization}

As is standard in the literature~\cite{Raginsky2016,cao2019tensorization}, tensorization of an SDPI constant for $\DD$ refers to the following property: 
\begin{align}
    \eta_\DD(\cN_1\otimes\cN_2,\sigma_1\otimes\sigma_2) \overset{?}{=} \max\{ \eta_\DD(\cN_1,\sigma_1), \eta_\DD(\cN_2,\sigma_2)\}.
\end{align}
To motivate this question, we can consider mixing times, as discussed further in \cref{sec:Hierarchies-operational-times}, and take $\sigma$ to be the fixed point of the channel. Usually, considering multiple systems requires an SDPI constant to be evaluated on a tensor product of quantum channels and states. However, tensorization implies that, if all systems are affected by the same noise, the exponential contraction rate does not deteriorate with the number of tensor factors.

Classically, this holds, e.g., for many $f$-divergences~\cite[Theorem 3.9]{Raginsky2016}, including the relative entropy; see also~\cite{anantharam2013maximal}. In the quantum setting without reference systems, it holds for one very particular definition of the $\chi^2$-divergence~\cite{cao2019tensorization}, but not for another definition of the $\chi^2$ nor the relative entropy~\cite{Cao2026noTensorization}.
In general, it is still an open problem to determine exactly when tensorization holds. Even for the relative entropy, it holds for two erasure channels~\cite[Lemma 3.11]{Hirche2022contraction}.  In the following, we will give conditions under which tensorization does hold for the relative entropy, and we also provide several concrete examples.  

\subsection{Without reference systems}

Let us first consider the case without reference systems. One direction is easy to see,
\begin{align}
    \eta_\DD(\cN_1\otimes\cN_2,\sigma_1\otimes\sigma_2) \geq \max\{ \eta_\DD(\cN_1,\sigma_1), \eta_\DD(\cN_2,\sigma_2)\},
    \label{eq:tensorization-ineq-easy}
\end{align}
following the same line of reasoning given in the proof of Lemma~\ref{lem:tensorization-ineqs-lem}.
To establish the other direction, we introduce the following mutual-information correlation coefficient of a bipartite channel $\cN_{A_1A_2\rightarrow B_1B_2}$:
\begin{align}
    \zeta_{\MI}(\cN_{A_1A_2\rightarrow B_1B_2}) \coloneqq  \sup_{\rho_{A_1A_2}\not\in \,\operatorname{PROD}
     } \frac{I(B_1:B_2)}{I(A_1:A_2)},
\end{align}
where PROD denotes the set of product states.
In what follows, we will only be interested in the setting of a product channel. For that, we collect a few properties here. 

\begin{lemma}
\label{Lem:zeta-MI-prop}
    For a product channel $\cN_{A_1\rightarrow B_1}\otimes\cN_{A_2\rightarrow B_2}$, the following properties hold for the mutual-information correlation coefficient: 
    \begin{enumerate}
        \item The coefficient is bounded as follows:
        \begin{align}
            0\leq\zeta_{\MI}(\cN_{A_1\rightarrow B_1}\otimes\cN_{A_2\rightarrow B_2})\leq1.
        \end{align}
        Note that the upper bound does not generally hold for non-product channels.
        \item Additionally,
        \begin{align}
    \zeta_{\MI}(\cN_1\otimes\cN_2) \leq \eta_{\MI}(\cN_1) \eta_{\MI}(\cN_2) .
    \label{eq:corr-coeff-to-MI-coeffs}
\end{align}
    \end{enumerate}
\end{lemma}

\begin{proof}
    For property 1, the lower bound follows from non-negativity of the mutual information and the upper bound from data processing. For non-product channels, the upper bound is violated for example by a replacer channel that prepares a maximally entangled state. Property 2 follows by noting that, for a product channel $\cN_1\otimes\cN_2$, 
\begin{align}
    I(B_1:B_2) \leq \eta_{\MI}(\cN_1) I(A_1:B_2) \leq \eta_{\MI}(\cN_1)\eta_{\MI}(\cN_2) I(A_1:A_2). 
\end{align}
Since the bound holds for every possible input state $\rho_{A_1 A_2}\not\in \operatorname{PROD}$, we conclude~\eqref{eq:corr-coeff-to-MI-coeffs}.
\end{proof}

For the relative entropy SDPI constant, the following approximate tensorization result holds. 

\begin{proposition}
\label{prop:almost-tensorization}
    For channels $\cN_1$ and $\cN_2$ and states $\sigma_1$ and $\sigma_2$, 
    \begin{align}
     \max\{ \eta_D(\cN_1,\sigma_1), \eta_D(\cN_2,\sigma_2)\}
     &\leq \eta_D(\cN_1\otimes\cN_2,\sigma_1\otimes\sigma_2) \\
     &\leq \max\{ \zeta_{\MI}(\cN_1\otimes\cN_2), \eta_D(\cN_1,\sigma_1), \eta_D(\cN_2,\sigma_2)\}.
\end{align}
\end{proposition}

\begin{proof}
    The first inequality is a special case of \eqref{eq:tensorization-ineq-easy}. For the second, by using the chain rule in \eqref{eq:rel-ent-chain}, consider that
    \begin{align}
        &D(\cN_1\otimes\cN_2(\rho_{A_1A_2})\|\cN_1(\sigma_{A_1})\otimes\cN_2(\sigma_{A_2})) \notag \\
        &= D(\cN_1\otimes\cN_2(\rho_{A_1A_2})\|\cN_1(\rho_{A_1})\otimes\cN_2(\rho_{A_2})) + D(\cN_1(\rho_{A_1})\|\cN_1(\sigma_{A_1}))+ D(\cN_2(\rho_{A_2})\|\cN_2(\sigma_{A_2})) \nonumber \\
        &\leq \zeta_{\MI}(\cN_1\otimes\cN_2) D(\rho_{A_1A_2}\|\rho_{A_1}\otimes\rho_{A_2}) + \eta_D(\cN_1,\sigma_1) D(\rho_{A_1}\|\sigma_{A_1})+ \eta_D(\cN_2,\sigma_2) D(\rho_{A_2}\|\sigma_{A_2}) \nonumber \\
        &\leq \max\{ \zeta_{\MI}(\cN_1\otimes\cN_2), \eta_D(\cN_1,\sigma_1), \eta_D(\cN_2,\sigma_2)\} \nonumber\\
        &\qquad\times\left[ D(\rho_{A_1A_2}\|\rho_{A_1}\otimes\rho_{A_2}) +  D(\rho_{A_1}\|\sigma_{A_1})+ D(\rho_{A_2}\|\sigma_{A_2}) \right] \\
        &= \max\{ \zeta_{\MI}(\cN_1\otimes\cN_2), \eta_D(\cN_1,\sigma_1), \eta_D(\cN_2,\sigma_2)\} D(\rho_{A_1A_2}\|\sigma_{A_1}\otimes\sigma_{A_2}),
    \end{align}
    from which the claim follows.
\end{proof}

Note that related approximate tensorization questions with respect to conditional expectations were also studied in~\cite{bardet2022approximate,gao2022complete}.
Now, \cref{prop:almost-tensorization} leads to the following corollary. 

\begin{corollary}\label{Cor:tensorize2}
    For channels $\cN_1$ and $\cN_2$ and states $\sigma_1$ and $\sigma_2$, if 
    \begin{align}
        \eta_{\MI}(\cN_1) \eta_{\MI}(\cN_2) \leq \max\{ \eta_D(\cN_1,\sigma_1), \eta_D(\cN_2,\sigma_2)\},
    \end{align}
    then
    \begin{align}
    \eta_D(\cN_1\otimes\cN_2,\sigma_1\otimes\sigma_2) 
     &= \max\{ \eta_D(\cN_1,\sigma_1), \eta_D(\cN_2,\sigma_2)\}.
\end{align}
\end{corollary}

\begin{proof}
    This follows by using property 2 of~\cref{Lem:zeta-MI-prop} in~\cref{prop:almost-tensorization}.
\end{proof}
This gives the following special cases. 

\begin{corollary}
\label{cor:tensorization-special-cases}
    The relative entropy SDPI constant tensorizes, i.e.
    \begin{align}
    \eta_D(\cN_1\otimes\cN_2,\sigma_1\otimes\sigma_2) 
     &= \max\{ \eta_D(\cN_1,\sigma_1), \eta_D(\cN_2,\sigma_2)\},
\end{align}
in the following special cases:
    \begin{enumerate}
        \item If $\cN_1=\cE_\epsilon$ is an erasure channel and $\cN_2$ is an arbitrary channel, or vice versa. 
        \item If $\cN_1=\cN_2=\cD_{p,\tau}$ are generalized depolarizing channels, defined as
        \begin{equation}
        \cD_{p,\tau}(\rho)\coloneqq (1-p)\rho + p \tau ,    
        \end{equation}
        and $\sigma_1=\sigma_2=\tau$ are their full-rank fixed points.
        \item If $\cN_1$ is a quantum-to-classical channel and if   $\sigma_2$ is such that $\eta_{\cqMI}(\cN_2,\sigma_2)=\eta_{\cqMI}(\cN_2)$.
    \end{enumerate}
\end{corollary}

\begin{proof}
    (1) For the erasure channel, all relevant contraction coefficients are equal, 
    \begin{align}
        \eta_{\MI}(\cE_\epsilon) \eta_{\MI}(\cN_2) = (1-\epsilon) \eta_{\MI}(\cN_2) \leq 1-\epsilon =\eta_D(\cN_1,\sigma_1), 
    \end{align}
    from which the claim follows from~\cref{Cor:tensorize2}. \\
    (2) By~\cref{Lem:a-Doeblin} and~\cref{Cor:Summary1}, we have, 
    \begin{align}
        \eta_{\MI}(\cD_{p,\tau}) \leq \etaD_{D}(\cD_{p,\tau}) \leq 1-\alpha_+(\cD_{p,\tau}) = 1-p. 
    \end{align}
    With that, 
    \begin{align}
        \eta_{\MI}(\cD_{p,\tau}) \eta_{\MI}(\cD_{p,\tau}) \leq (1-p)^2 \leq  \eta_D(\cD_{p,\tau},\tau),
    \end{align}
    where the second inequality is from~\cite[Theorem 3.1]{muller2016relative} and is tight for dimension equal to two. The claim then follows from~\cref{Cor:tensorize2}.\\ 
    (3) Let $\sigma_2$ be such that $\eta_{\cqMI}(\cN_2,\sigma_2)=\eta_{\cqMI}(\cN_2)$. Because we assume $B_1$ to be classical, we have, 
    \begin{align}
        I(B_1:B_2) \leq \eta_{\cqMI}(\cN_2) I(B_1:A_2), 
    \end{align}
    and hence 
\begin{align}
    \zeta_{\MI}(\cN_1\otimes\cN_2) \leq \eta_{\MI}(\cN_1) \eta_{\cqMI}(\cN_2) \leq \eta_{\cqMI}(\cN_2)= \eta_{\cqMI}(\cN_2,\sigma_2) = \eta_{D}(\cN_2,\sigma_2),
\end{align}
where the second-to-last inequality follows from our initial assumption. This concludes the proof by~\cref{prop:almost-tensorization}. 
\end{proof}
Ref.~\cite{GaoZhao2026} recently provided derivations of tight large-index estimates for the relative-entropy contraction of generalized depolarizing dynamics, together with tensor-stable estimates under arbitrary ancillary extensions. For replacer dynamics, the latter correspond to the product-reference SDPI constant $\eta_D^p$ considered here. Combined with the tensorization result in item~2 of~\cref{cor:tensorization-special-cases}, their single-system estimates in the replacer setting also yield system-size-independent bounds for arbitrary tensor powers.
In the next part, we study tensorization in settings with reference systems. 

\subsection{SDPIs with reference systems}

In what follows, we discuss variants of the above results that directly address complete and conditional contraction coefficients. First, we note the particular role that $\eta_D^p(\cN,\sigma_A)$, as defined in~\eqref{SDPI-product}, has in this context. 
It turns out that this quantity indeed tensorizes. Closely related complete SDPI tensorization results were established in~\cite[Proposition 2.9]{gao2020fisher}; see also~\cite[Proposition 4.2]{gao2022complete}. For completeness, we state and prove here the product-reference form needed in our framework. 

\begin{theorem}
    For channels $\cN_1$ and $\cN_2$ and states $\sigma_1$ and $\sigma_2$,
    \begin{align}
    \eta^p_D(\cN_1\otimes\cN_2,\sigma_1\otimes\sigma_2) = \max\{ \eta^p_D(\cN_1,\sigma_1), \eta^p_D(\cN_2,\sigma_2)\}.
\end{align}
\end{theorem}

\begin{proof}
    The $\geq$ direction is clear by choosing appropriate product states. For the $\leq$ direction, observe
    \begin{align}
        &D(\cN_1\otimes\cN_2(\rho_{RA_1A_2})\|\rho_R\otimes\cN_1(\sigma_{A_1})\otimes\cN_2(\sigma_{A_2})) \notag \\
        &=D(\cN_1\otimes\cN_2(\rho_{RA_1A_2})\|\cN_1(\rho_{RA_1})\otimes\cN_2(\sigma_{A_2})) + D(\cN_1(\rho_{RA_1})\|\rho_R\otimes\cN_1(\sigma_{A_1})) \\
        &\leq\eta^p_D(\cN_2,\sigma_2) D(\cN_1(\rho_{RA_1A_2})\|\cN_1(\rho_{RA_1})\otimes\sigma_{A_2}) + \eta^p_D(\cN_1,\sigma_1) D(\rho_{RA_1}\|\rho_R\otimes\sigma_{A_1}) \\
        &\leq\eta^p_D(\cN_2,\sigma_2) D(\rho_{RA_1A_2}\|\rho_{RA_1}\otimes\sigma_{A_2}) + \eta^p_D(\cN_1,\sigma_1) D(\rho_{RA_1}\|\rho_R\otimes\sigma_{A_1}) \\
        &\leq\max\{ \eta^p_D(\cN_1,\sigma_1), \eta^p_D(\cN_2,\sigma_2)\} \left[ D(\rho_{RA_1A_2}\|\rho_{RA_1}\otimes\sigma_{A_2}) + D(\rho_{RA_1}\|\rho_R\otimes\sigma_{A_1}) \right] \\
        &=\max\{ \eta^p_D(\cN_1,\sigma_1), \eta^p_D(\cN_2,\sigma_2)\}  D(\rho_{RA_1A_2}\|\rho_R\otimes\sigma_{A_1}\otimes\sigma_{A_2}),
    \end{align}
    where the equalities follow from the relative entropy chain rule in \eqref{eq:rel-ent-chain}, the first inequality follows by applying the SDPI constants, the second follows from data processing, and the third is a direct upper bound. 
\end{proof}
Note that by~\cref{Lem:MI-product-SDPI}, this immediately implies that the MI SDPI constant tensorizes, 
\begin{align}
    \eta_{\MI}(\cN_1\otimes\cN_2,\sigma_1\otimes\sigma_2) 
     &= \max\{ \eta_{\MI}(\cN_1,\sigma_1), \eta_{\MI}(\cN_2,\sigma_2)\},
\end{align}
while the counterexample of \cite{Cao2026noTensorization} implies that the same does not hold for $\eta_{\cqMI}$. 

We now turn our attention to the complete and conditional SDPI constants. 
Again, it is clear that
\begin{align}
    \etaC_\DD(\cN_1\otimes\cN_2,\sigma_1\otimes\sigma_2) & \geq \max\{ \etaC_\DD(\cN_1,\sigma_1), \etaC_\DD(\cN_2,\sigma_2)\}, \\
    \etaD_\DD(\cN_1\otimes\cN_2,\sigma_1\otimes\sigma_2) & \geq \max\{ \etaD_\DD(\cN_1,\sigma_1), \etaD_\DD(\cN_2,\sigma_2)\},
\end{align}
by following the same line of reasoning given in the proof of Lemma~\ref{lem:tensorization-ineqs-lem}.
We will examine the converse direction closer for the case of the relative entropy. For that we need, 
\begin{align}
    \zeta^c_{\MI}(\cN_{A_1A_2\rightarrow B_1B_2}) &\coloneqq  \sup_{\substack{\rho_{A_1A_2R_1R_2},\\\rho_{R_1R_2}=\rho_{R_1}\otimes\rho_{R_2}}} \frac{I(R_1B_1:B_2R_2)}{I(R_1A_1:A_2R_2)}, \\ 
    \zeta^\Delta_{\MI}(\cN_{A_1A_2\rightarrow B_1B_2}) &\coloneqq  \sup_{\rho_{A_1A_2R_1R_2}} \frac{I(R_1B_1:B_2R_2)-I(R_1:R_2)}{I(R_1A_1:A_2R_2)-I(R_1:R_2)}. 
\end{align}
Similar to the case without reference systems, we can conclude the following approximate tensorization result. 

\begin{proposition}
\label{prop:almost-tensorization-ca}
For channels $\cN_1$ and $\cN_2$ and states $\sigma_1$ and $\sigma_2$,
    \begin{align}
     \max\{ \etaC_D(\cN_1,\sigma_1), \etaC_D(\cN_2,\sigma_2)\}
     &\leq \etaC_D(\cN_1\otimes\cN_2,\sigma_1\otimes\sigma_2) \\
     &\leq \max\{ \zeta^c_{\MI}(\cN_1\otimes\cN_2), \etaC_D(\cN_1,\sigma_1), \etaC_D(\cN_2,\sigma_2)\}, \\
    \max\{ \etaD_D(\cN_1,\sigma_1), \etaD_D(\cN_2,\sigma_2)\}
     &\leq \etaD_D(\cN_1\otimes\cN_2,\sigma_1\otimes\sigma_2) \\
     &\leq \max\{ \zeta^\Delta_{\MI}(\cN_1\otimes\cN_2), \etaD_D(\cN_1,\sigma_1), \etaD_D(\cN_2,\sigma_2)\}.
\end{align}
\end{proposition}
\begin{proof}
    The proof is essentially identical to that for~\cref{prop:almost-tensorization}. 
\end{proof}
Note however that, by the counterexample of \cite{Cao2026noTensorization}, exact tensorization cannot hold in general also for these quantities because the choice of $\sigma_{RA}$ fixes the reference system and could potentially be chosen to be trivial. 
\begin{corollary}[Failure of tensorization for complete
and conditional SDPIs]
\label{Cor:no-tensorization-complete-conditional-SDPI}
There exist quantum channels
$\cN,\overline{\cN}\colon A\to B$ and a full-rank state
$\widehat{\sigma}_{A}$ such that
\begin{align}
    &\etaC_D\!\left(
        \cN\otimes\overline{\cN},
        \widehat{\sigma}_{A}\otimes\widehat{\sigma}_{A}
    \right) >
    \max\left\{
        \etaC_D(\cN,\widehat{\sigma}_{A}),
        \etaC_D(\overline{\cN},\widehat{\sigma}_{A})
    \right\},
    \label{eq:no-tensorization-complete-SDPI}
\end{align}
and
\begin{align}
    &\etaD_D\!\left(
        \cN\otimes\overline{\cN},
        \widehat{\sigma}_{A}\otimes\widehat{\sigma}_{A}
    \right) >
    \max\left\{
        \etaD_D(\cN,\widehat{\sigma}_{A}),
        \etaD_D(\overline{\cN},\widehat{\sigma}_{A})
    \right\}.
    \label{eq:no-tensorization-conditional-SDPI}
\end{align}
Here, the fixed states are understood
as states with one-dimensional reference systems. Consequently,
neither the complete nor the conditional
relative-entropy SDPI constant tensorizes in general.
\end{corollary}

\begin{proof}
Ref.~\cite[Proposition~4]{Cao2026noTensorization} showed that there
exist quantum channels $\cN,\overline{\cN}$ and a full-rank state
$\widehat{\sigma}_{A}$ such that
\begin{align}
    \eta_D\!\left(
        \cN\otimes\overline{\cN},
        \widehat{\sigma}_{A}\otimes\widehat{\sigma}_{A}
    \right)
    >
    \eta_D(\cN,\widehat{\sigma}_{A})
    =
    \eta_D(\overline{\cN},\widehat{\sigma}_{A}).
    \label{eq:Cao-relative-entropy-counterexample}
\end{align}

Now take the reference system $R$ to be one-dimensional. In this case,
the equal-reference-marginal condition in the definition of the
complete SDPI constant is vacuous, while
\begin{align}
    D(\rho_R\|\sigma_R)=0.
\end{align}
Therefore, for every channel $\cM$ and every state $\sigma_A$,
\begin{align}
    \etaC_D(\cM,\sigma_A)
    =
    \etaD_D(\cM,\sigma_A)
    =
    \eta_D(\cM,\sigma_A),
    \label{eq:trivial-reference-SDPI-identities}
\end{align}
where $\sigma_A$ is identified with the state on the joint system
$RA$ for a one-dimensional $R$.

Applying~\eqref{eq:trivial-reference-SDPI-identities} to
$\cM=\cN$, $\cM=\overline{\cN}$, and
$\cM=\cN\otimes\overline{\cN}$ in
\eqref{eq:Cao-relative-entropy-counterexample} proves both
\eqref{eq:no-tensorization-complete-SDPI} and
\eqref{eq:no-tensorization-conditional-SDPI}.
\end{proof}

Interestingly, the approximate tensorization results in \cref{prop:almost-tensorization-ca} are in contrast to the exact tensorization of $\eta_D^p$. We have seen earlier, in~\cref{Lem:MI-product-SDPI}, that adding an additional reference system akin to $\etaD_D$ it can require additional stabilization to recover favorable properties. In the following, we show that $\eta_D^{\Delta,\operatorname{stab}}$ indeed tensorizes. 

\begin{theorem}[Tensorization of stabilized conditional SDPI constants]
\label{Prop:tensorization-stabilized-conditional-SDPI}
Let $\cN_i\colon A_i\to B_i$ be quantum channels and let
$\sigma^{(i)}_{R_iA_i}$ be quantum states for $i\in\{1,2\}$. Then
\begin{equation}
    \eta_D^{\Delta,\operatorname{stab}}
    \!\left(
        \cN_1\otimes\cN_2,
        \sigma^{(1)}_{R_1A_1}\otimes
        \sigma^{(2)}_{R_2A_2}
    \right)
    =
    \max\left\{
        \eta_D^{\Delta,\operatorname{stab}}
            (\cN_1,\sigma^{(1)}_{R_1A_1}),
        \eta_D^{\Delta,\operatorname{stab}}
            (\cN_2,\sigma^{(2)}_{R_2A_2})
    \right\}.
    \label{eq:tensorization-stabilized-conditional-SDPI}
\end{equation}
Equivalently, by
\cref{Prop:CMI-stabilized-SDPI}, the CMI SDPI constant tensorizes:
\begin{equation}
    \eta_{\CMI}
    \!\left(
        \cN_1\otimes\cN_2,
        \sigma^{(1)}_{R_1A_1}\otimes
        \sigma^{(2)}_{R_2A_2}
    \right)
    =
    \max\left\{
        \eta_{\CMI}
            (\cN_1,\sigma^{(1)}_{R_1A_1}),
        \eta_{\CMI}
            (\cN_2,\sigma^{(2)}_{R_2A_2})
    \right\}.
    \label{eq:tensorization-fixed-state-CMI-SDPI}
\end{equation}
\end{theorem}

\begin{proof}
By~\cref{Prop:CMI-stabilized-SDPI}, it suffices to prove
\eqref{eq:tensorization-fixed-state-CMI-SDPI}. Set
\begin{align}
    \eta_i
    \coloneqq
    \eta_{\CMI}
        (\cN_i,\sigma^{(i)}_{R_iA_i}),
    \qquad
    \eta\coloneqq\max\{\eta_1,\eta_2\}.
\end{align}

Let $\rho_{CR_1R_2A_1A_2}$ be an arbitrary state satisfying
\begin{align}
    \rho_{R_1R_2A_1A_2}
    =
    \sigma^{(1)}_{R_1A_1}
    \otimes
    \sigma^{(2)}_{R_2A_2}.
    \label{eq:product-fixed-marginal-CMI-tensorization}
\end{align}
Denote by
\begin{align}
    \omega_{CR_1R_2B_1A_2}
    &\coloneqq
    \cN_1(\rho_{CR_1R_2A_1A_2}),\\
    \nu_{CR_1R_2B_1B_2}
    &\coloneqq
    (\cN_1\otimes\cN_2)
        (\rho_{CR_1R_2A_1A_2}),
\end{align}
i.e., 
the states after the first channel and after both channels,
respectively.
By the chain rule for conditional mutual information,
\begin{align}
    I(C:B_1B_2|R_1R_2)_\nu
    &=
    I(C:B_1|R_1R_2)_\omega
    +
    I(C:B_2|R_1R_2B_1)_\nu.
    \label{eq:CMI-output-chain-tensorization}
\end{align}
The product condition
\eqref{eq:product-fixed-marginal-CMI-tensorization} implies
\begin{align}
    I(R_2:A_1|R_1)_\rho=0,
    \qquad
    I(R_2:B_1|R_1)_\omega=0.
\end{align}
Consequently, using the conditional mutual information chain rule \eqref{eq:CMI-to-CMI-chain-rule} twice,
\begin{align}
    I(C:B_1|R_1R_2)_\omega
    &=
    I(CR_2:B_1|R_1)_\omega\\
    &\leq
    \eta_1 I(CR_2:A_1|R_1)_\rho\\
    &=
    \eta_1 I(C:A_1|R_1R_2)_\rho.
    \label{eq:first-channel-CMI-tensorization}
\end{align}
For the second term, regard $CR_1B_1$ as the auxiliary system and
$R_2$ as the reference system for $\cN_2$. Again,
\eqref{eq:product-fixed-marginal-CMI-tensorization} implies
\begin{align}
    I(R_1B_1:A_2|R_2)_\omega=0,
    \qquad
    I(R_1B_1:B_2|R_2)_\nu=0.
\end{align}
It follows that
\begin{align}
    I(C:B_2|R_1R_2B_1)_\nu
    &=
    I(CR_1B_1:B_2|R_2)_\nu\\
    &\leq
    \eta_2 I(CR_1B_1:A_2|R_2)_\omega\\
    &=
    \eta_2 I(C:A_2|R_1R_2B_1)_\omega.
    \label{eq:second-channel-CMI-tensorization-pre-DPI}
\end{align}
The remaining CMI can be related to the original input state. Indeed,
the product marginal in
\eqref{eq:product-fixed-marginal-CMI-tensorization} gives
\begin{align}
    I(A_1:A_2|R_1R_2)_\rho=0,
    \qquad
    I(B_1:A_2|R_1R_2)_\omega=0.
\end{align}
Using the chain rule and data processing under
$\cN_1\colon A_1\to B_1$, we obtain
\begin{align}
    I(C:A_2|R_1R_2B_1)_\omega
    &=
    I(CB_1:A_2|R_1R_2)_\omega\\
    &\leq
    I(CA_1:A_2|R_1R_2)_\rho\\
    &=
    I(C:A_2|R_1R_2A_1)_\rho.
    \label{eq:second-channel-CMI-tensorization-DPI}
\end{align}
In summary, we find
\begin{align}
    I(C:B_1B_2|R_1R_2)_\nu &= I(C:B_1|R_1R_2)_\omega
    +
    I(C:B_2|R_1R_2B_1)_\nu \\
    &\leq
    \eta_1 I(C:A_1|R_1R_2)_\rho +
    \eta_2 I(C:A_2|R_1R_2A_1)_\rho\\
    &\leq
    \eta\left[
        I(C:A_1|R_1R_2)_\rho
        +
        I(C:A_2|R_1R_2A_1)_\rho
    \right]\\
    &=
    \eta I(C:A_1A_2|R_1R_2)_\rho,
    \label{eq:CMI-tensorization-upper-bound}
\end{align}
where the first equality follows from the chain rule in~\eqref{eq:CMI-output-chain-tensorization}, the first inequality follows from combining
\eqref{eq:first-channel-CMI-tensorization},
\eqref{eq:second-channel-CMI-tensorization-pre-DPI}, and
\eqref{eq:second-channel-CMI-tensorization-DPI}, and the final equality again follows from the chain rule.
Since $\rho_{CR_1R_2A_1A_2}$ is arbitrary, this proves
\begin{align}
    \eta_{\CMI}
    \!\left(
        \cN_1\otimes\cN_2,
        \sigma^{(1)}_{R_1A_1}\otimes
        \sigma^{(2)}_{R_2A_2}
    \right)
    \leq
    \max\{\eta_1,\eta_2\}.
    \label{eq:CMI-tensorization-upper}
\end{align}

The converse inequality follows by the same arguments given in the proof of \cref{lem:tensorization-ineqs-lem}, which involve choosing appropriate product states to satisfy the lower bound. 
Together with~\eqref{eq:CMI-tensorization-upper}, this proves
\eqref{eq:tensorization-fixed-state-CMI-SDPI}. The stabilized
conditional-SDPI identity
\eqref{eq:tensorization-stabilized-conditional-SDPI} then follows from
\cref{Prop:CMI-stabilized-SDPI}. 
\end{proof}

\begin{figure}[t]
\centering
\begin{tikzpicture}[
    base/.style={
        draw,
        rounded corners=2pt,
        text width=4.7cm,
        minimum height=1.35cm,
        inner sep=5pt,
        align=center,
        font=\small,
        line width=0.6pt
    },
    yes/.style={
        base,
        draw=green!55!black,
        fill=green!5
    },
    no/.style={
        base,
        draw=red!55!black,
        fill=red!5
    },
    heading/.style={
        font=\small\bfseries,
        align=center
    },
    relation/.style={
        fill=white,
        inner sep=1.5pt,
        font=\small
    }
]


\node[heading] at (-2.9,3.05)
    {Non-stabilized SDPI};

\node[heading] at (2.9,3.05)
    {Stabilized SDPI};


\node[no] (standard) at (-2.9,1.8)
{
    $\eta_D(\cN,\sigma_A)$\\
    $=\eta_{\cqMI}(\cN,\sigma_A)$\\[1mm]
    {\scriptsize\bfseries Does not tensorize in general}
};

\node[yes] (product) at (2.9,1.8)
{
    $\eta_D^p(\cN,\sigma_A)$\\
    $=\eta_{\MI}(\cN,\sigma_A)$\\[1mm]
    {\scriptsize\bfseries Exact tensorization}
};


\node[no] (complete) at (-2.9,0)
{
    $\etaC_D(\cN,\sigma_{RA})$\\
    $=\eta_{\operatorname{cqMIR}}(\cN,\sigma_{RA})$\\[1mm]
    {\scriptsize\bfseries Does not tensorize in general}
};

\node[no] (complete-stab) at (2.9,0)
{
    $\eta_D^{\mathrm{c,stab}}(\cN,\sigma_{RA})$\\
    $=\eta_{\operatorname{MIR}}(\cN,\sigma_{RA})$\\[1mm]
    {\scriptsize\bfseries Does not tensorize in general}
};


\node[no] (conditional) at (-2.9,-1.8)
{
    $\etaD_D(\cN,\sigma_{RA})$\\
    $=\eta_{\operatorname{cqCMI}}(\cN,\sigma_{RA})$\\[1mm]
    {\scriptsize\bfseries Does not tensorize in general}
};

\node[yes] (conditional-stab) at (2.9,-1.8)
{
    $\eta_D^{\Delta,\operatorname{stab}}
        (\cN,\sigma_{RA})$\\
    $=\eta_{\CMI}(\cN,\sigma_{RA})$\\[1mm]
    {\scriptsize\bfseries Exact tensorization}
};


\node[relation] at (0,1.8) {$\leq$};

\node[relation] at (0,0) {$\leq$};

\node[relation] at (0,-1.8) {$\leq$};

\end{tikzpicture}

\caption{
Status of tensorization of the relative-entropy SDPI constants considered
in this work. Stabilization cannot decrease the corresponding
coefficient. Green entries tensorize exactly for fixed product states and
red entries do not tensorize in general. The ordinary SDPI nevertheless
tensorizes under the sufficient conditions established in
\cref{Cor:tensorize2}. \cref{Sec:no-tensor-for-stabC} features a proof that the stabilized complete SDPI does not tensorize in general. For all quantities on the left side, we have an approximate  tensorization result. 
}
\label{fig:SDPI-tensorization-summary}
\end{figure}

Finally, the stabilized complete SDPI constant also does not tensorize
in general. We provide an explicit counterexample in
~\cref{Sec:no-tensor-for-stabC}. We summarize the resulting
tensorization picture in~\cref{fig:SDPI-tensorization-summary}. 


\section{Trace distance and hockey-stick divergences}\label{Sec:Hockey-Stick}

In this section, we focus on the quantum hockey-stick divergence and, in particular, the trace distance as a special case. The hockey-stick divergence is defined as~\cite{sharma2012strongconversesquantumchannel} 
\begin{align}
    E_\gamma(\rho\|\sigma) &\coloneqq  \tr(\rho-\gamma\sigma)_+ \\
    &= \sup_{0 \leq M \leq I} \tr[M(\rho-\gamma\sigma)], 
\end{align}
where $(A)_+$ denotes the positive part of a Hermitian matrix $A$.
The trace distance then corresponds to the above definition for $\gamma=1$.

One feature that makes the hockey-stick divergence particularly interesting here is that, for $\gamma\geq1$, its contraction coefficient simplifies to an optimization over orthogonal pure states~\cite{ruskai1994beyond,hirche2022quantum},
\begin{align}
\label{Eq:Ey-orthogonal-pure}
    \eta_\gamma(\cN) = \sup_{\Phi_1\perp\Phi_2} E_\gamma(\cN(\Phi_1)\|\cN(\Phi_2)),
\end{align}
where $\Phi_1$ and $\Phi_2$ are pure states.
It will be instructive for what follows to provide a brief proof of this result. In fact, the following proof is significantly simpler than the previous proof in~\cite[Theorem II.2]{hirche2022quantum} and is a direct consequence of the following auxiliary lemma.
\begin{lemma}
\label{lem:hockey-stick-orthogonal}
    For all quantum states $\rho$ and $\sigma$ and $\gamma\geq1$, there exist orthogonal states $\tau_+$ and $\tau_-$, such that
    \begin{align}
        \rho-\gamma\sigma\leq\lambda_+ (\tau_+-\gamma\tau_-),
    \end{align}
    where $\lambda_+=E_\gamma(\rho\|\sigma)$.
\end{lemma}

\begin{proof}
The case $\lambda_+=0$ is immediate, so suppose $\lambda_+>0$. We can use the Jordan--Hahn decomposition
\begin{align}
    \rho-\gamma\sigma &= X_+ - X_- \\
    &= \lambda_+ \tau_+ - \lambda_- \tau_-, \label{eq:ortho-state-decomp}
\end{align}
where $\lambda_+=\Tr X_+ = E_\gamma(\rho\|\sigma)$, $\lambda_-=\Tr X_-$, $\tau_\pm=\frac{X_\pm}{\lambda_\pm}$, and $X_+ X_{-}=0$.  
The crucial step here is the following observation: 
\begin{align}
    1-\gamma & = \lambda_+ - \lambda_- \\
    \Leftrightarrow \qquad \frac{\lambda_-}{\lambda_+} & = 1 - \frac{1-\gamma}{\lambda_+} \geq \gamma \\
    \Rightarrow \qquad  \lambda_- & \geq \gamma\lambda_+, 
\end{align}
where the inequality in the second line follows because  $\gamma\geq1$ and $\lambda_+\leq1$. 
This directly implies that
\begin{align}
    \rho-\gamma\sigma &= \lambda_+ \left(\tau_+ - \frac{\lambda_-}{\lambda_+} \tau_-\right) \leq \lambda_+ \left( \tau_+ - \gamma\tau_-\right),
\end{align}
thus concluding the proof.
\end{proof}

\cref{lem:hockey-stick-orthogonal} has a number of implications. First, let $\rho$ and $\sigma$ be pure states; then, as a special case, we recover~\cite[Lemma 2]{nuradha2024contraction}. 
Next, it implies the following inequality for all $\gamma \geq 1$: 
\begin{align}
    E_\gamma(\cN(\rho)\|\cN(\sigma)) \leq   E_\gamma(\rho\|\sigma)\sup_{\tau_1\perp\tau_2} E_\gamma(\cN(\tau_1)\|\cN(\tau_2)),  
    \label{eq:hockey-stick-ortho-channels}
\end{align}
which follows because
\begin{align}
\rho-\gamma\sigma& \leq E_\gamma(\rho\|\sigma) (\tau_+-\gamma\tau_-) \\
\implies \quad \cN (\rho)-\gamma \cN(\sigma)& \leq E_\gamma(\rho\|\sigma) (\cN(\tau_+)-\gamma \cN(\tau_-))\\
\implies \quad \Tr[M (\cN (\rho)-\gamma \cN(\sigma)]& \leq E_\gamma(\rho\|\sigma) \Tr[M(\cN(\tau_+)-\gamma \cN(\tau_-))] \quad \forall \ 0 \leq M \leq I\\
\implies E_\gamma(\cN(\rho)\|\cN(\sigma)) & \leq E_\gamma(\rho\|\sigma)\sup_{\tau_1\perp\tau_2} E_\gamma(\cN(\tau_1)\|\cN(\tau_2)).
\end{align}
Together with convexity, \cref{eq:hockey-stick-ortho-channels}
 recovers the known result in~\eqref{Eq:Ey-orthogonal-pure}. 

It turns out that, in the setting with reference systems, this approach does not quite work for general $\gamma$. However, for $\gamma=1$, the approach remains useful. Before presenting the results, we prove the following meta-theorem. 

\begin{theorem}\label{Thm:Meta-1}
    For every pair of states $\rho_{RA}$ and $\sigma_{RA}$ satisfying $\rho_R=\sigma_R$, and for every pair of channels $\mathcal{N}\colon A\to B$ and $\mathcal{M}\colon A\to B'$, there exists a pair of orthogonal states $\tau_{RA}$ and $\bar\tau_{RA}$ with $\tau_{R}=\bar\tau_{R}$, such that 
    \begin{align}
        \frac{E_1(\cN(\rho_{RA})\|\cN(\sigma_{RA}))}{E_1(\cM(\rho_{RA})\|\cM(\sigma_{RA}))} = \frac{E_1(\cN(\tau_{RA})\|\cN(\bar\tau_{RA}))}{E_1(\cM(\tau_{RA})\|\cM(\bar\tau_{RA}))}, \label{Eq:Meta-1}
    \end{align}
    and, similarly, for every pair of states $\rho_{RA}$ and $\sigma_{RA}$, there exists a pair of orthogonal states $\tau_{RA}$ and $\bar\tau_{RA}$, such that
    \begin{align}
        &\frac{E_1(\cN(\rho_{RA})\|\cN(\sigma_{RA})) - E_1(\rho_R\|\sigma_R)}{E_1(\cM(\rho_{RA})\|\cM(\sigma_{RA})) - E_1(\rho_R\|\sigma_R)} = \frac{E_1(\cN(\tau_{RA})\|\cN(\bar\tau_{RA})) - E_1(\tau_R\|\bar\tau_R)}{E_1(\cM(\tau_{RA})\|\cM(\bar\tau_{RA})) - E_1(\tau_R\|\bar\tau_R)} . \label{Eq:Meta-2}
    \end{align}
\end{theorem}

\begin{proof}
    The definitions of $\lambda_{\pm}$ for~\eqref{eq:ortho-state-decomp} imply for $\gamma = 1$ that $\lambda_+ = \lambda_-$. In this case, \cref{eq:ortho-state-decomp} simplifies to 
    \begin{align}
        \rho_{RA} - \sigma_{RA} = \lambda_{+}(\tau_{+}-\tau_{-}) \ ,\label{Eq:help1}
    \end{align}
    which we will use to prove both identities.
    We start by proving~\eqref{Eq:Meta-2}.
Because of the above,
\begin{align}
    \rho_R - \sigma_R &= \Tr_A (\rho_{RA} - \sigma_{RA}) = \lambda_+ \Tr_A[\tau_+ -  \tau_-] = \lambda_+\left(\tau_{+,R} - \tau_{-,R}\right),\label{Eq:help2}
\end{align}
where $\tau_{\pm,R} \coloneqq \Tr_A[\tau_{\pm}]$.
With this, we have, 
\begin{align}
\frac{E_1(\cN(\rho_{RA})\|\cN(\sigma_{RA})) - E_1(\rho_R\|\sigma_R)}{E_1(\cM(\rho_{RA})\|\cM(\sigma_{RA})) - E_1(\rho_R\|\sigma_R)} 
    & = \frac{\lambda_+ E_1(\cN(\tau_{+})\|\cN(\tau_{-})) - \lambda_+ E_1(\tau_{+,R}\|\tau_{-,R})}{\lambda_+ E_1(\cM(\tau_{+})\|\cM(\tau_{-})) - \lambda_+ E_1(\tau_{+,R}\|\tau_{-,R})} \\
    & = \frac{ E_1(\cN(\tau_{+})\|\cN(\tau_{-})) -  E_1(\tau_{+,R}\|\tau_{-,R})}{ E_1(\cM(\tau_{+})\|\cM(\tau_{-})) -  E_1(\tau_{+,R}\|\tau_{-,R})},
\end{align}
where, in the first equality, we have used~\eqref{Eq:help1} and~\eqref{Eq:help2} and the second equality follows by reducing the fraction to remove $\lambda_+$. This proves~\eqref{Eq:Meta-2}. 

Next, we show~\eqref{Eq:Meta-1}. Note that, if $\rho_{RA}$ and $\sigma_{RA}$ have the same marginals, then
\begin{align}
    0 = \Tr_A (\rho_{RA}-\sigma_{RA}) = \lambda_+ \Tr_A(\tau_+ - \tau_-),
\end{align}
which implies that $\Tr_A\tau_+=\Tr_A\tau_-$ and hence $E_1(\tau_{+,R}\|\tau_{-,R})=0$. Plugging this into~\eqref{Eq:Meta-2} completes the proof.  
\end{proof}

With this result, it is now straightforward to prove simplified expressions for the trace distance contraction coefficients. 
\begin{proposition}\label{Prop:sup-ort-enough}
For a channel $\mathcal{N}\colon A\to B$, the coefficients $\etaC_{\Tr}(\cN)$ and $\etaD_{\Tr}(\cN)$ are optimized by orthogonal states, so that
    \begin{align}
        \etaC_{\Tr}(\cN) &= \sup_{\substack{\tau_{RA} \perp \bar\tau_{RA},\\ \tau_R=\bar\tau_R}} E_1(\cN(\tau_{RA})\|\cN(\bar \tau_{RA})), \label{eq:orthogonal-complete-cc}\\
        \etaD_{\Tr}(\cN) &= \sup_{\substack{\tau_{RA} \perp \bar\tau_{RA}}} \frac{E_1(\cN(\tau_{RA})\|\cN(\bar\tau_{RA})) - E_1(\tau_R\|\bar\tau_R)}{1- E_1(\tau_R\|\bar\tau_R)}.
        \label{eq:orthogonal-conditional-cc}
\end{align}
\end{proposition} 

\begin{proof}
    Both claims directly follow by one of the results in~\cref{Thm:Meta-1} each, by choosing $\cM=\id$, taking the supremum over all states on the right side, and then on the left. 
\end{proof}

Notably, the optimization here is over mixed orthogonal states, in contrast to pure orthogonal states as in \eqref{Eq:Ey-orthogonal-pure}. It is not clear if the optimizations can be further simplified to be over pure orthogonal states. 

We next prove that similar results hold for expansion coefficients.

\begin{proposition} For a channel $\cN\colon A\to B$, 
    \begin{align}
         \cetaC_{\Tr}(\cN) &= \inf_{\substack{\tau_{RA} \perp \bar\tau_{RA},\\ \tau_R=\bar\tau_R}} E_1(\cN(\tau_{RA})\|\cN(\bar \tau_{RA})), \label{eq:orthogonal-complete-ec}\\
        \cetaD_{\Tr}(\cN) &= \inf_{\substack{\tau_{RA} \perp \bar\tau_{RA}}} \frac{E_1(\cN(\tau_{RA})\|\cN(\bar\tau_{RA})) - E_1(\tau_R\|\bar\tau_R)}{1- E_1(\tau_R\|\bar\tau_R)}.
        \label{eq:orthogonal-conditional-ec}
    \end{align}
\end{proposition}
\begin{proof}
    Both claims directly follow by one of the results in~\cref{Thm:Meta-1} each, by choosing $\cM=\id$, taking the infimum over all states on the right side and then on the left. 
\end{proof}
Note that also the complete trace distance expansion coefficient is often non-trivial, because it can be lower bounded by a so-called reverse Doeblin coefficient as shown in~\cite[Proposition 33]{george2025quantumdoeblincoefficientsinterpretations}. This quantity is an SDP and can easily be checked for numerous examples. However, whether the same bound holds for the conditional expansion coefficient is currently unknown. 

\begin{remark}
    In the setting without reference systems, the restriction to orthogonal states in the optimization for all $\gamma\geq1$ has the useful consequence that one can easily prove that $\eta_\gamma(\cN)$ is non-increasing in $\gamma$. The absence of a similar result in the settings with reference systems raises the question whether the monotonicity does still hold here. We leave this question open for now, but note that some preliminary numerical evidence suggests that monotonicity in $\gamma$ might not hold in the presence of reference systems. 
\end{remark}

Finally, we remark that, by~\cref{Thm:Meta-1}, it also suffices to optimize over orthogonal states for the relative contraction coefficient and the relative expansion coefficient. However,  this observation does not simplify the expression by much because the channel $\cM$ is generally arbitrary. 

\subsection{Separating complete and  conditional contraction coefficients}

In this subsection, we explore the general question of whether $\etaC_{\DD}(\cN)$ is strictly smaller than $\etaD_{\DD}(\cN)$ by presenting an affirmative answer when the divergence is the hockey-stick divergence.

\begin{proposition}
For the hockey-stick divergence with $\gamma >1$, there exists a classical channel $\cN \equiv W_{Y|X}$ such that $\etaD_{E_\gamma}(\cN) > \etaC_{E_\gamma}(\cN)$.     
\end{proposition}

\begin{proof}
Let $\cN$ be a binary symmetric channel with flipping probability $p \in [0,1/2)$ such that $(1-p)/p=\gamma$. We deduce the above claim by showing that $0=\etaC_{E_\gamma}(\cN)<1-2p \leq \etaD_{E_\gamma}(\cN)$.

We first show that $\etaC_{E_\gamma}(\cN)=0$ even when the reference
system is quantum for the chosen classical channel $\cN$. Let
\begin{align}
    \rho_{RX}
    &= q\,\rho_R^0\otimes |0\rangle\!\langle 0|
       +(1-q)\,\rho_R^1\otimes |1\rangle\!\langle 1|,\\
    \sigma_{RX}
    &= r\,\sigma_R^0\otimes |0\rangle\!\langle 0|
       +(1-r)\,\sigma_R^1\otimes |1\rangle\!\langle 1|
\end{align}
be arbitrary classical-quantum states, where $\rho_R^0$, $\rho_R^1$, $\sigma_R^0$, and $\sigma_R^1$
are normalized quantum states, satisfying 
\begin{align}
    \rho_R=\sigma_R\eqqcolon \omega_R.
\end{align}
Thus,
\begin{align}
    q\rho_R^0+(1-q)\rho_R^1
    =
    r\sigma_R^0+(1-r)\sigma_R^1
    =\omega_R.
\end{align}
Denote the two output blocks of $(\id_R\otimes\cN)(\rho_{RX})$ by
$\widetilde\rho_R^0$ and $\widetilde\rho_R^1$. For the binary symmetric
channel with crossover probability $p$, they satisfy
\begin{align}
    \widetilde\rho_R^0
    &= (1-p)q\rho_R^0+p(1-q)\rho_R^1
     = p\omega_R+(1-2p)q\rho_R^0
     \leq (1-p)\omega_R,\\
    \widetilde\rho_R^1
    &= pq\rho_R^0+(1-p)(1-q)\rho_R^1
     = p\omega_R+(1-2p)(1-q)\rho_R^1
     \leq (1-p)\omega_R,
\end{align}
where we used $q\rho_R^0\leq \omega_R$ and $(1-q)\rho_R^1\leq \omega_R$.
Analogously, for the output blocks
$\widetilde\sigma_R^0,\widetilde\sigma_R^1$ of
$(\id_R\otimes\cN)(\sigma_{RX})$, we have
\begin{align}
    \widetilde\sigma_R^0\geq p\omega_R,
    \qquad
    \widetilde\sigma_R^1\geq p\omega_R.
\end{align}
Since $\gamma=(1-p)/p$, it follows for $x\in\{0,1\}$ that
\begin{align}
    \widetilde\rho_R^x
    \leq (1-p)\omega_R
    = \gamma p\omega_R
    \leq \gamma\widetilde\sigma_R^x.
\end{align}
Consequently,
\begin{align}
    (\id_R\otimes\cN)(\rho_{RX})
    \leq
    \gamma(\id_R\otimes\cN)(\sigma_{RX}),
\end{align}
and hence
\begin{align}
    E_\gamma\!\left(
    (\id_R\otimes\cN)(\rho_{RX})
    \middle\|
    (\id_R\otimes\cN)(\sigma_{RX})
    \right)=0.
\end{align}
As this holds for every quantum reference system $R$ and every pair
with $\rho_R=\sigma_R$, we conclude that
\begin{align}
    \etaC_{E_\gamma}(\cN)=0.
\end{align}

To obtain a lower bound on the conditional contraction coefficient, we consider the following setting by choosing the reference system $R$ to be classical: for the following input joint distributions for $r,x \in \{0,1\}$, with $r$ a random variable for the reference system and $x$ for the channel input,
\begin{align}
    P_{RX}(r,x)& \coloneqq  \mathbf{1}\{(r,x)=(1,1) \}, \\
     Q_{RX}(r,x) & \coloneqq  \left(1- \frac{1}{\gamma} \right) \mathbf{1} \{(r,x)=(0,1) \} + \frac{1}{\gamma } \mathbf{1}\{(r,x)=(1,0) \},
\end{align}
where $\mathbf{1}$ is the indicator function.
With that, along with the notation $P_{RY} \equiv W_{Y|X} (P_{RX})$ and $Q_{RY} \equiv W_{Y|X} (Q_{RX})$, we have 
\begin{align}
     E_{\gamma}(P_R \Vert Q_R) &=( 0- \gamma (1-1/\gamma))_+ + (1-\gamma / \gamma)_+ =0, \\
     E_{\gamma}(P_{RX} \Vert Q_{RX}) &=1, \\ 
    E_{\gamma}(W_{Y|X} (P_{RX}) \Vert W_{Y|X}( Q_{RX})) &= \sum_{y \in \{0,1\}} (P_{RY}(1,y) -\gamma Q_{RY}(1,y))_+ + 0\\
    &= \sum_{y \in \{0,1\}} \left(W_{Y|X}(y|1) -\gamma \frac{1}{\gamma} W_{Y|X}(y|0) \right)_+ \\
    &=(1-2p).
\end{align}
Since $p >0$ for $\gamma >1$, it follows that
\begin{align}
     \sup_{\substack{P_{RX},Q_{RX}}} 
   \frac{E_{\gamma}(W_{Y|X}(P_{RX}) \| W_{Y|X}(Q_{RX}) - E_{\gamma}(P_R \| Q_R) }{ E_\gamma (P_{RX} \Vert Q_{RX}) -E_\gamma(P_R \Vert Q_R)} &\geq \frac{(1-2p)-0 }{1-0}\\
    &  =1-2p \\
    & >0 \\
    &= \etaC_{E_\gamma}(W_{Y|X}),
\end{align}
thus proving the claim.
\end{proof}

A crucial ingredient in this example is that the hockey-stick divergence is not faithful. It remains an interesting open question to determine whether similar examples exist for faithful divergences.

\section{Applications to quantum networks}

\label{Sec:Applications1}

\subsection{Contraction of mutual information in networks}

One of the most natural motivations for considering contraction with reference systems is to improve our understanding of data processing in network settings. This can ultimately lead to a better understanding of what information is preserved in a network and how information should flow through the network. 
Ref.~\cite{polyanskiy2015strong} considered the contraction of mutual information in classical networks. 
Our goal in what follows is to find a suitable  generalization of their result that applies to quantum networks. 

Consider the setting of a quantum channel $\cN_{A\to BE} = \cR_{D\to E}\circ\cM_{A\to BD}$ that is applied to an input state $\rho_{CA}$, depicted  diagrammatically as follows:
\vspace{1em}
\begin{center}
\begin{tikzpicture}[
    node distance=2cm,
    every node/.style={},
    >={Stealth[length=2mm,width=2mm]}
]
\node (A) {$A$};
\node[left=of A] (C) {$C$};
\node[right=of A] (B) {$B$};
\node[below right=of A] (D) {$D$};
\node[right=of D] (E) {$E$};

\draw[-] (C) -- node[midway, above] {$\rho_{CA}$} (A);
\draw[->] (A) -- (B);
\draw[->] (A) --node[pos=0.2, right] {$\;\cM_{A\to BD}$}  (D);
\draw[->] (D) -- node[midway, above] {$\cR_{D\to E}$} (E);
\end{tikzpicture}
\end{center}
\vspace{1em}
In what follows, we use the notation $\cM_{A\to B}\equiv \tr_D \circ \cM_{A\to BD}$.

We are interested in bounding the contraction of the channel $\cN$, based on the contraction of $\cM$ and $\cR$. The system $C$ here plays the role of a reference system for the input, and we are interested in the mutual information between system $C$ and different outputs of the quantum network. We first establish the main result of this section. 
\begin{proposition}\label{Prop:Polyanski-quantum}
    With the above definitions, the following inequalities hold:
    \begin{align}
        \eta_{\MI}(\cN_{A\to BE}) &\leq \eta_{\cR}\cdot\eta_{\MI}(\cM_{A\to BD}) + (1-\eta_{\cR})\cdot \eta_{\MI}(\cM_{A\to B}), \label{Eq:Polyanski-quantum-1} \\
        \eta_{D}(\cN_{A\to BE}) &\leq \eta_{\cR}\cdot\eta_{D}(\cM_{A\to BD}) + (1-\eta_{\cR}) \cdot\eta_{D}(\cM_{A\to B}), \label{Eq:Polyanski-quantum-2}
    \end{align}
    where
    \begin{align}
        \eta_{\cR} \equiv \etaD_D(\cR_{D\to E}).
        \label{eq:eta-R-notation}
    \end{align}
\end{proposition}
\begin{proof}
    We start by proving the first inequality in~\eqref{Eq:Polyanski-quantum-1}. First note the following:
    \begin{align}
        I(C:BE) &= I(C:BE) - I(C:B) + I(C:B) \\ 
        &\leq \etaD_D(\cR_{D\to E})\left[ I(C:BD) - I(C:B) \right] + I(C:B) \label{Eq:use-a-here}\\
        &= \etaD_D(\cR_{D\to E}) I(C:BD) +(1-\etaD_D(\cR_{D\to E})) I(C:B) , 
    \end{align}
    where the inequality follows from \eqref{eq:CMI-MI-chain-rule} and \cref{Cor:CMI-amo-equal}.
    This implies, using the definition of~$\eta_{\cR}$, that
    \begin{align}
        \frac{I(C:BE)}{I(C:A)} &\leq \eta_{\cR} \frac{I(C:BD)}{I(C:A)} +(1-\eta_{\cR}) \frac{I(C:B)}{I(C:A)} \\
        &\leq \eta_{\cR}\,\eta_{\MI}(\cM_{A\to BD}) + (1-\eta_{\cR}) \,\eta_{\MI}(\cM_{A\to B}). 
    \end{align}
    Now taking the supremum over all $\rho_{CA}\not\in \operatorname{PROD}$ on the left proves the claim.
    
    The second inequality in~\eqref{Eq:Polyanski-quantum-2} for the relative entropy contraction then follows in the same way, assuming that $C$ is a classical system and recalling that $\eta_D(\cN)=\eta_{\cqMI}(\cN)$. 
\end{proof}

Note that, in the classical setting, Eqs.~\eqref{Eq:Polyanski-quantum-1} and~\eqref{Eq:Polyanski-quantum-2} are equivalent; however, in the quantum case, they are different, as follows from \eqref{eq:separation-between-eta-MI-and-eta-D}.
While~\cref{Eq:Polyanski-quantum-2} appears as the more direct generalization of the classical result, \cref{Eq:Polyanski-quantum-1} is a fully quantum version that explicitly considers a quantum system $C$. Recall that we can state both results either as mutual information contraction or as relative entropy contraction, because we have seen previously that $\eta_D(\cN)=\eta_{\cqMI}(\cN)$ and $\eta^p_D(\cN)=\eta_{\MI}(\cN)$ (recall \cref{Cor:Summary1}). Interestingly, the proof is the same for classical and quantum $C$, because, from~\cref{Cor:CMI-amo-equal}, we know that the conditional mutual information contraction coefficient is the same in both cases. 

If we are interested in a computable bound for $\eta_{\cR}$ defined in \eqref{eq:eta-R-notation}, we can use the following variant. 
\begin{corollary}
        With the above definitions, 
    \begin{align}
        \eta_{\MI}(\cN_{A\to BE}) &\leq (1-\alpha_+(\cR))\cdot\eta_{\MI}(\cM_{A\to BD}) + \alpha_+(\cR)\cdot \eta_{\MI}(\cM_{A\to B}), \label{Eq:Polyanski-quantum-1a} \\
        \eta_{D}(\cN_{A\to BE}) &\leq (1-\alpha_+(\cR))\cdot\eta_{D}(\cM_{A\to BD}) + \alpha_+(\cR) \cdot\eta_{D}(\cM_{A\to B}). \label{Eq:Polyanski-quantum-2a}
    \end{align}
\end{corollary}
\begin{proof}
    This follows as before, by using the bound $\eta_{\cR}\leq 1-\alpha_+(\cR)$ after~\eqref{Eq:use-a-here}.
\end{proof}

At this point, it is worth remarking that the results in~\cite{polyanskiy2015strong} also cover scenarios that include classical feedback. In the fully quantum setting, the limitations in the topology of our networks compared to those in~\cite{polyanskiy2015strong} are akin to the no-cloning theorem, forbidding us to copy and reuse quantum information. However, our result is actually more general, in the sense that if we restrict some communication paths to be at least partially classical, we can recover the topology discussed in~\cite{polyanskiy2015strong}. Here we consider a network of the following form:
\vspace{1em}
\begin{center}
\begin{tikzpicture}[
    node distance=2cm,
    every node/.style={},
    >={Stealth[length=2mm,width=2mm]}
]
\node (A) {$X$};
\node[left=of A] (C) {$C$};
\node[right=of A] (B) {$V$};
\node[below=of B] (D) {$D$};
\node[right=of D] (E) {$E$};

\draw[-] (C) -- node[midway, above] {$\rho_{CX}$} (A);
\draw[double,->] (A) -- node[midway, above] {$\;\cC_{X\to V}$} (B);
\draw[double,->] (B) -- (D);
\draw[double,->] (B) -- (E);
\draw[double,->] (A) --node[pos=0.2, right] {$\;\cM_{XV\to D}$}  (D);
\draw[->] (D) -- node[midway, above] {$\cR_{VD\to E}$} (E);
\end{tikzpicture}
\end{center}
\vspace{1em}
In the above diagram, $X$ and $V$ are now classical systems, and maps denoted by double arrows are classical input channels, with $\cC_{X\to V}$ being a fully classical channel. Now, we see that this is indeed a special case of~\cref{Prop:Polyanski-quantum}, by noting that we can make three copies of the output of the fully classical channel $\cC_{X\to V}$, which results in the following diagram:
\vspace{1em}
\begin{center}
\begin{tikzpicture}[
    node distance=2cm,
    every node/.style={},
    >={Stealth[length=2mm,width=2mm]},
    copy wire/.style={
        ->,
        densely dashed,
        line width=0.45pt
    }
]

\node (A) {$X$};
\node[left=of A] (C) {$C$};
\node[right=of A] (B) {$V$};
\node[below=of B] (D) {$D$};
\node[right=of D] (E) {$E$};

\draw[-]
    (C) --
    node[midway,above] {$\rho_{CX}$}
    (A);

\draw[double,->]
    (A) --
    node[midway,below] {$\;\cC_{X\to V}$}
    (B);

\draw[double,->]
    (A) --
    node[pos=0.5,left] {$\;\cM_{XV\to D}$}
    (D);

\draw[->]
    (D) --
    node[midway,above] {$\cR_{VD\to E}$}
    (E);

\draw[copy wire]
    (B) to[bend right=42]
    (A);

\draw[copy wire]
    (B) --node[midway,right,font=\scriptsize]
    {copies of $V$} (D);

\end{tikzpicture}
\end{center}
\vspace{1em}
and is the desired special case. 

The result in~\cref{Prop:Polyanski-quantum} also allows us to observe another interesting property of the mutual-information contraction coefficients. 

\begin{corollary}
\label{Cor:MI-tensor-n}
    Let $n\in \mathbb{N}$, and let $\cN$ be a quantum channel. Then
    \begin{align}
        \eta_{\MI}(\cN^{\otimes n}) \leq 1 - (1-\etaD_D(\cN))^n \leq n\cdot\etaD_D(\cN). 
    \end{align}
\end{corollary}

\begin{proof}
    Let us denote the input and output systems of $\cN^{\otimes n}$ by $\hat{A}^n \equiv \hat{A}_1 \cdots \hat{A}_n$ and $\hat{B}^n \equiv \hat{B}_1 \cdots \hat{B}_n$, respectively. By identifying
         $\cM = \cN^{\otimes (n-1)} \otimes \id_{\hat{A}_{n}}$ and $\cR = \cN_{\hat{A}_{n} \to \hat{B}_{n}}$  and applying~\cref{Prop:Polyanski-quantum}, we have that
    \begin{align}
        \eta_{\MI}(\cN^{\otimes n}) &\leq \etaD_D(\cN)\cdot \eta_{\MI}\!\left(\cN^{\otimes(n-1)}\otimes\id_{A_n}\right) + (1-\etaD_D(\cN)) \eta_{\MI}(\cN^{\otimes(n-1)}) \\
        &=  \etaD_D(\cN) + (1-\etaD_D(\cN)) \eta_{\MI}(\cN^{\otimes(n-1)}). \label{Eq:Aux-tensor-bound} 
    \end{align}
    where the equality follows because  $\eta_{\MI}(\id)=1$. 
    Subtracting $1$ on both sides, this becomes,
    \begin{align}
        \eta_{\MI}(\cN^{\otimes n}) -1 &\leq  (1-\etaD_D(\cN)) \left(\eta_{\MI}(\cN^{\otimes(n-1)})-1\right) \\
        &\leq  (1-\etaD_D(\cN))^{n-1} \left(\eta_{\MI}(\cN)-1\right)\\
        &\leq  (1-\etaD_D(\cN))^{n-1} \left(\etaD_D(\cN)-1\right) \\
        &=  - (1-\etaD_D(\cN))^{n},    \end{align}
    where the second inequality iteratively applies the first and the third uses $\eta_{\MI}(\cN)\leq \etaD_D(\cN)$. Adding back $1$ on both sides proves the first claim. For the second claim we used that $1-x^n \leq n(1-x)$ for all $n\in\NN$ and $x\in[0,1]$.  
\end{proof}
This can be compared to the Doeblin upper bound derived in~\cite[Corollary 9]{george2025quantumdoeblincoefficientsinterpretations}, which is the middle inequality in the following, 
\begin{align}
    \eta_D(\cN^{\otimes n}) \leq \eta_{\tr}(\cN^{\otimes n}) \leq 1 - \alpha_\wang(\cN)^n \leq 1- \alpha_+(\cN)^n. 
\end{align}
(See \cite[Definition~6]{george2025quantumdoeblincoefficientsinterpretations} for a definition of $\alpha_\wang(\cN)$.)
Although the two might not be generally comparable,~\cref{Cor:MI-tensor-n} gives a tighter bound whenever we are interested in the contraction of mutual information. A classical analog of this result was previously reported in~\cite[Corollary~6]{polyanskiy2015strong}.

Alternatively, if we are interested in the tensor product of different channels, we can conclude results along the following lines.

\begin{corollary}
\label{Cor:NtensorM-upper}
    For two quantum channels $\cN$ and $\cM$ and $\star\in\{D,\MI\}$,
    \begin{align}
        \eta_{\star}(\cN\otimes\cM) \leq \min\begin{cases}
            \eta_{\star}(\cM) + \etaD_{D}(\cN) -\eta_{\star}(\cM)\cdot\etaD_{D}(\cN) \\
            \etaD_{D}(\cM) + \eta_{\star}(\cN) - \etaD_{D}(\cM)\cdot \eta_{\star}(\cN)
        \end{cases}.
    \end{align}
\end{corollary}

\begin{proof}
    The proof follows directly from~\eqref{Eq:Aux-tensor-bound}, applied once to remove $\cN$ and once to remove~$\cM$. 
\end{proof}
This gives some general upper bounds on the contraction coefficient for tensor-product channels, which were previously only known for special cases, such as when one channel is an erasure channel (see~\cite[Lemma 3.9]{Hirche2022contraction}). Concretely, considering the qubit depolarizing channel leads to the following bounds:
\begin{align}
    (1-p)^2\leq \frac{2(1-p)^2}{1+(1-p)^2}\leq\eta_D(\cD_p^{\otimes 2}) \leq 1-2p^2+p^3 \leq 1-p^2, \label{Eq:Contraction-depol-2}
\end{align}
 The inequality $\eta_D(\cD_p^{\otimes 2}) \leq 1-2p^2+p^3$ follows from~\cref{Cor:NtensorM-upper}.
 The final inequality $1-2p^2+p^3 \leq 1-p^2$ gives, to our knowledge, the best previously known upper bound.
This follows from
$\etaD_{D}(\cD_p)\leq1-p$, based on the Doeblin upper bound. An improvement over this would also improve the upper bound in the above example. 
The loose lower bound above is given by $\eta_D(\cD_p)$ itself, while the tighter lower bound $\frac{2(1-p)^2}{1+(1-p)^2}\leq\eta_D(\cD_p^{\otimes 2})$ can be derived by fixing input states $\rho_0=a|00\rangle\!\langle00|+(1-a)|11\rangle\!\langle11|$ and $\rho_1=(1-a)|00\rangle\!\langle00|+a|11\rangle\!\langle11|$ and taking the limit $a\rightarrow\frac12$. 
We illustrate the resulting bounds in~\cref{fig:Dp-bounds}. 

\begin{figure}
    \centering
    \includegraphics[width=0.6\linewidth]{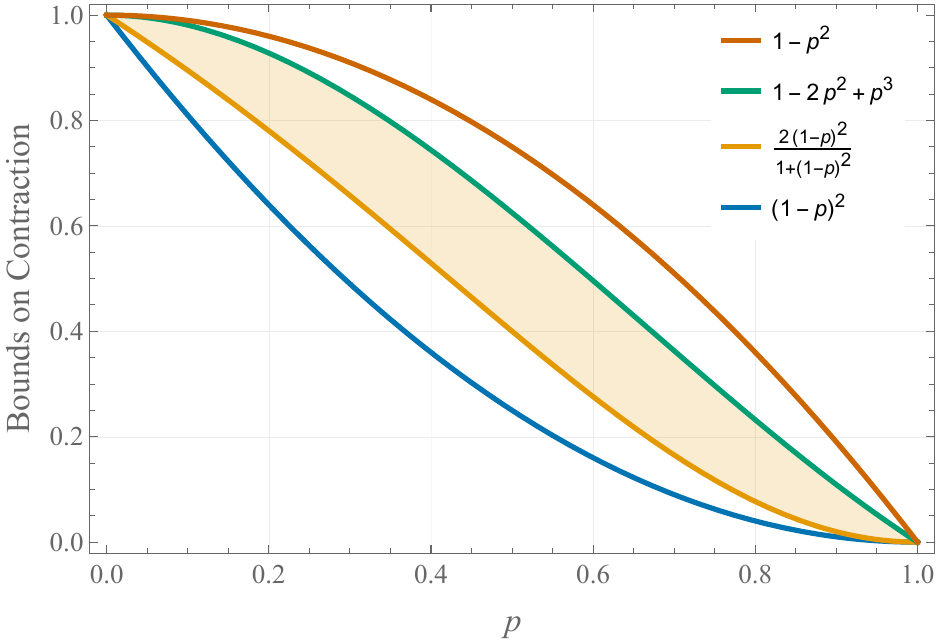}
    \caption{Illustration of the bounds on $\eta_D(\cD_p^{\otimes 2})$ as given in~\eqref{Eq:Contraction-depol-2}. The shaded area gives the remaining range of possible values for the contraction coefficient.}
    \label{fig:Dp-bounds}
\end{figure}

\subsection{Limits on quantum memories for preserving resources even with post-processing}

\label{subsec:limits-on-q-mem-with-post-processing}

A quantum memory is a device that preserves a quantum system and its correlations with other systems over time. In other words, it preserves entanglement or a useful state such that it may be used for subsequent quantum information processing tasks. For this reason, it is thought that quantum memories will be a significant component of the future quantum internet \cite{Wehner18a}.

As near-term quantum memories will be noisy, it is important to determine the limits of quantum memories with noisy components. One method for determining these limits is to model the quantum memory as a discrete, time-homogeneous Markov chain $\left(\cN^{n}_{A \to A}\right)_{n \in \mbb{N}}$ and study to what depth $n$ the chain still serves as a `good' quantum memory. One way of addressing this question involves studying the zero-error quantum capacity of $\left(\cN^{n}\right)_{n}$ over time \cite{singh2024zeroerrorcommunicationdiscretetimemarkovian,fawzi2025capacitiesquantummarkoviannoise,bhattacharyya2025sequentialtransmissionshorttimes}. However, as identified in previous works, e.g.,~\cite{singh2024zeroerrorcommunicationdiscretetimemarkovian,fawzi2025capacitiesquantummarkoviannoise}, a problem with considering the capacity is that there is only one quantum memory in practice, rather than many copies of the memory, as would be needed for the capacity to be relevant. Thus, one should consider the one-shot setting. A second problem that, to the best of our knowledge has neither been raised nor addressed, is that the users of the quantum memory can in principle post-process the state after it has undergone the noisy memory for the relevant task. For example, if the memory is to preserve entanglement between two parties, the parties could perform entanglement distillation on the state that comes out of their memories in order to obtain some entanglement, so long as the memories did not destroy all the distillable entanglement (see Fig.~\ref{fig:quantum-memory-post-proc} for a visualization). 

\begin{figure}
    \centering
    \includegraphics[width=0.7\linewidth]{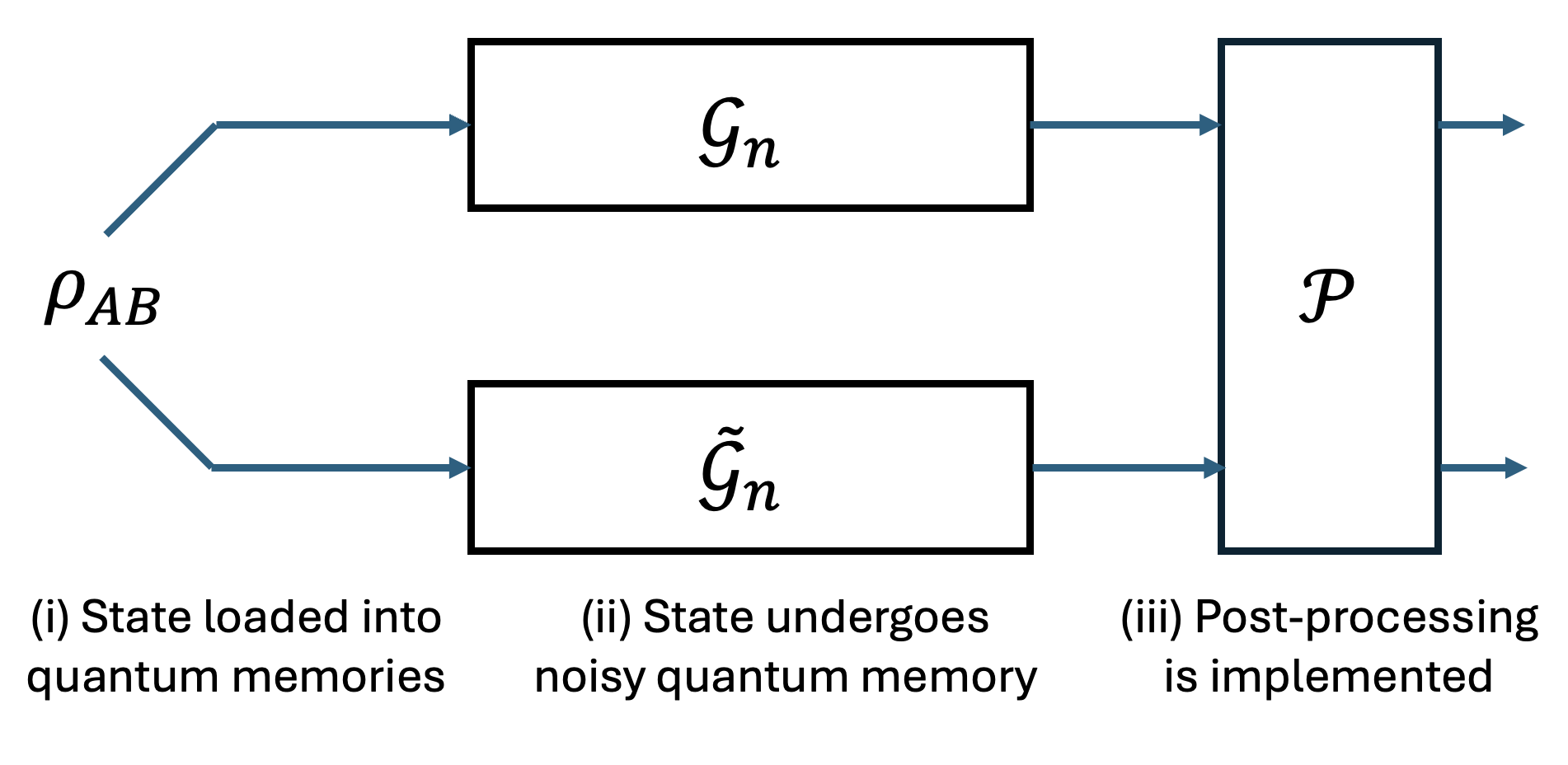}
    \caption{Depiction of the setting we consider in Section \ref{subsec:limits-on-q-mem-with-post-processing}.}
    \label{fig:quantum-memory-post-proc}
\end{figure}

In this subsection, we establish limits on the time $n$ such that the quantum memories cannot be used for storing entanglement or quantum states for secret key, even under post-processing. All the proofs follow the same idea: initially Alice and Bob have some shared, possibly resourceful state $\rho_{A_{1}B_{1}}$ loaded into their local quantum memories. These quantum memories are described by discrete, time-inhomogeneous Markov chains
\begin{equation}
\left(\cG_{n} \coloneq \bigcirc_{i \in [n]}  \cN^{i}_{A_{i} \to A_{i+1}} \right)_{n}, \qquad  \left(\wt{\cG}_{n} \coloneq \bigcirc_{i \in [n]}  \wt{\cN}^{i}_{A_{i} \to A_{i+1}} \right)_{n}   , 
\end{equation}
defined by channels $\cN^{i}_{A_{i} \to A_{i+1}}$ and $\wt{\cN}^{i}_{B_{i} \to B_{i+1}}$, respectively. At some time $n' \in \mbb{N}$, the users decide that they want the relevant resource, so they run resource distillation on the output state $(\cG_{n'} \otimes \wt{\cG}_{n'})(\rho)$. By bounding the one-shot rate of the relevant resource distillation task by mutual information or conditional mutual information and applying the appropriate contraction coefficient, we determine a time $n \in \mbb{N}$ after which no resource can be distilled from the output of the quantum memories, independent of what the initial state $\rho_{A_{1}B_{1}}$ is. The correspondence between task, allowed post-processing operations, and used contraction coefficient is summarized in Table~\ref{table:limits-on-memories}.

An interesting point to note is that entanglement distillation with LOCC post-processing is controlled by $\eta_{\CMI}$ but entanglement distillation with non-entangling operations is controlled by $\eta_{\MI}$. This is interesting as $\eta_{\MI} \leq \eta_{\CMI}$ while LOCC is a strict subset of non-entangling operations. Thus, it should be \textit{easier} to distill entanglement with non-entangling operations, but our bounds contract at least as fast in this setting. Of course, the relevant constants in the theorems differ, so this could be what makes up for the difference in principle. In contrast to entanglement distillation, for key distillation, we see that one-way LOCC is limited by $\eta_{\MI}$, whereas general LOCC is limited by $\eta_{\CMI}$. In this case, the governing contraction coefficient and the power of the post-processing resource aligns appropriately.
\begin{table}
\centering
\begin{tabular}
[c]{c|c|c}\hline\hline
Task & Allowed Post-Processing & Contraction Coefficient  \\
\hline\hline
\multirow{2}{*}{Entanglement Distillation} & Non-Entangling & $\eta_{\MI}$ (\Cref{thm:q-mem-lims-from-ent-dist-with-NE-maps}) \\ \cline{2-3}
& LOCC & $\eta_{\CMI}$ (\Cref{thm:q-mem-lims-from-ent-dist-with-LOCC-maps}) \\ \hline
\multirow{2}{*}{Key Distillation} & One-Way LOCC & $\eta_{\MI}$ (\Cref{thm:limits-for-one-way-key-distil}) \\ \cline{2-3}
& LOCC & $\eta_{\CMI}$ (\Cref{thm:limits-on-key-distil-with-two-way}) \\
\hline \hline
\end{tabular}
\caption{Summary of limits on quantum memories for achieving certain tasks with post-processing.
}
\label{table:limits-on-memories}
\end{table}

In the remaining part of this subsection, we prove  our bounds.

\subsubsection{Limits on entanglement preservation with non-entangling map post-processing}

We begin with entanglement distillation as measured in terms of singlet fraction, which for a $d$-dimensional state is the quantity 
\begin{align}
    F(\rho,\dyad{\Phi_d}) = \bra{\Phi_{d}}\rho\ket{\Phi_{d}} \ ,
\end{align} 
where $\ket{\Phi_{d}} \coloneqq  \frac{1}{\sqrt{d}} \sum_i |i\rangle |i\rangle $ denotes the standard maximally entangled state vector of Schmidt rank $d$ and $F(P,Q) \coloneqq \left\Vert \sqrt{P}\sqrt{Q} \right\Vert_{1}^{2}$ is the Uhlmann fidelity of positive semidefinite matrices $P$ and $Q$ \cite{uhlmann1976}. 
For entanglement distillation under an allowed set of channels, $\cS$, we denote the one-shot distillable entanglement by
\begin{align}
    E_{D,\cS}^{\varepsilon}(\rho) \coloneq \max_{d \in \mbb{N},\cN \in \cS} \{ \log d : \bra{\Phi_{d}}\cN(\rho)\ket{\Phi_{d}} \geq 1 - \ve \} \ .
\end{align}
The first set of channels we will consider are non-entangling channels, which we denote by $\operatorname{NE}$. To establish bounds on the depth of quantum memories at which one is able to distill entanglement under non-entangling channels, we need to bound $E^{\ve}_{D,\operatorname{NE}}$ by a correlation measure for which we have a contraction coefficient. The following establishes such a bound in terms of mutual information.
\begin{lemma}
\label{lem:one-shot-dist-ent-ub-by-MI}
    For $\varepsilon \in [0,1]$ and a bipartite state $\rho$,
    \begin{align}
        E_{D,\operatorname{NE}}^{(1),\varepsilon}(\rho) \leq \frac{1}{1-\ve}[I(A:B)_\rho +h(\ve)],
    \end{align}
    where $h(\ve)\coloneqq -\ve \log \ve - (1-\ve)\log (1-\ve)$.
\end{lemma}

\begin{proof}
Consider that
    \begin{align}
        E_{D,\operatorname{NE}}^{(1),\varepsilon}(\rho) &\leq \max_{\substack{0 \leq \Lambda \leq I \\ \Tr[\rho \Lambda] \geq 1 - \epsilon}} \min_{\sigma \in \operatorname{SEP}} - \log\Tr[\Lambda \sigma] \\
        &\leq \min_{\sigma \in \operatorname{SEP}} \max_{\substack{0 \leq \Lambda \leq I \\ \Tr[\rho \Lambda] \geq 1 - \epsilon}} - \log\Tr[\Lambda \sigma] \\
        &= \min_{\sigma \in \operatorname{SEP}} D^{\ve}_{h}(\rho \Vert \sigma) \\
        &\leq D_{h}^{\ve}(\rho_{AB} \Vert \rho_{A} \otimes \rho_{B}) \\
        &\leq \frac{1}{1-\ve}[I(A:B)+h(\ve)] \ , 
    \end{align}
    where the first inequality follows from \cite[Theorem 1]{Brandao-Datta-one-shot-ent}, the second from the max-min inequality, the equality from the definition of hypothesis testing relative entropy (see, e.g., \cite{{Wang_2012}}), the third inequality from choosing $\sigma = \rho_{A} \otimes \rho_{B}$, which is manifestly separable, and the final inequality from \cite[Eq.~(2)]{Wang_2012} and the definition of mutual information.
\end{proof}

We now establish our limits on quantum memories.

\begin{theorem}
\label{thm:q-mem-lims-from-ent-dist-with-NE-maps}
    Let $\ve \in [0,1/2)$, and let $\left(\cG_{n}\right)_{n}$ and $ \left(\wt{\cG}_{n}\right)_{n}$ be inhomogeneous Markov chains defined by channels $\cN^{i}_{A_{i} \to A_{i+1}}$ and $\wt{\cN}^{i}_{B_{i} \to B_{i+1}}$, respectively. For $n \geq \hat{n}$ where 
    \begin{align}
        \hat{n} \coloneq \min\!\left\{n \in \mbb{N}: \left[\prod_{i \in [n]} \eta_{\MI}(\cN^{i})\eta_{\MI}(\wt{\cN}^{i}) \right] < \frac{1-\ve - h(\ve)}{2\log\min\{\vert A_{1} \vert, \vert B_{1} \vert\}} \right\} \ ,
    \end{align}
    it is the case $E^{\ve}_{D,\operatorname{NE}}\left(\cG_{n} \otimes \wt{\cG}_{n}(\rho_{A_{1}B_{1}})\right) < 1$ for all quantum states $\rho_{A_{1}B_{1}}$. That is, for every input quantum state to the quantum memories, after time $\hat{n}$ the users cannot distill a state with singlet fraction greater than $1-\ve$ using non-entangling post-processing. \

    In particular, for homogeneous Markov chains $\cN_{A \to A}$ and $\cM_{B \to B}$, 
    \begin{align}
        \hat{n} = \left \lceil \frac{\log\!\left(\frac{2\log(\min\{\vert A \vert, \vert B \vert\})}{1-\ve-h(\ve)}\right)}{\log\!\left(\frac{1}{\eta_{\MI}(\cN)\eta_{\MI}(\cM)}\right)} \right \rceil \ . 
    \end{align}
\end{theorem}

\begin{proof}
    It suffices to show that the one-shot distillable entanglement is bounded from above by one, as this would show that one cannot distill any state with sufficient overlap with the maximally entangled state, i.e., not even a single ebit. By arithmetical manipulation,
    \begin{align}
        \frac{1}{1-\ve}[I(A:B) + h(\ve)] < 1 \iff I(A:B) < 1 - (\ve + h(\ve)) \ .
    \end{align}
    By using the definition of the mutual information contraction coefficient, for every initial state $\rho_{A_{1}B_{1}}$,
    \begin{align}
        I(A_{n}:B_{n})_{(\cG_{n} \otimes \wt{\cG}_{n})(\rho)} &\leq \prod_{i \in [n]} \eta_{\MI}(\cN^{i}) I(A_{1}:B_{n})_{(\id \otimes \wt{\cG}_{n})(\rho)} \\
        &\leq \left[\prod_{i \in [n]} \eta_{\MI}(\cN^{i})\eta_{\MI}(\wt{\cN}^{i}) \right]I(A_{1}:B_{1})\\
        &\leq \left[\prod_{i \in [n]} \eta_{\MI}(\cN^{i})\eta_{\MI}(\wt{\cN}^{i}) \right]2\log\min\{\vert A_{1} \vert, \vert B_{1} \vert\} \ . \label{eq:MI-upper-bound}
    \end{align}
    Thus, by \cref{lem:one-shot-dist-ent-ub-by-MI}, if~\eqref{eq:MI-upper-bound} is strictly bounded from above by $1 - (\ve + h(\ve))$,  the claim holds. 

    The special case of homogeneous Markov chains follows from arithmetic and the contraction coefficient on each system being the same at every time step.
\end{proof}

\begin{remark}
    The fact that the interval does not include $1/2$ is a consequence of there always existing a product state with singlet fraction $1/2$. One may verify that \cref{lem:one-shot-dist-ent-ub-by-MI} would require $I(A:B) < 0$ in the case $\ve = 1/2$, which is consistent with this fact.
\end{remark}

\subsubsection{Limits on entanglement preservation with LOCC post-processing}

While \cref{thm:q-mem-lims-from-ent-dist-with-NE-maps} is appealing in the sense that it is a very strong no-go theorem, as non-entangling channels represent a powerful class of free channels, there are at least two natural reservations against it. First, non-entangling channels are perhaps \textit{too} powerful a class of channels to be relevant. Second, in \eqref{eq:MI-upper-bound} we relaxed the mutual information to twice the local dimension, which ought to be too loose of a starting quantity as one can only distill at most half of that--- the minimum of the local dimensions. 

The following development resolves both of these problems. First, we place limits on using LOCC post-processing, the set of channels implementable by the two users being able to implement any quantum channel to their local systems and being able to exchange classical communication between each other\footnote{Of course, if the users can perform any local channel, one might expect they also could have a fault-tolerant quantum memory. Nonetheless, this only makes our no-go theorems more powerful.}, which are considered more reasonable than non-entangling channels (see \cite{Chitambar_2014} and references therein). Second, to do this, we use bounds in terms of the squashed entanglement of a bipartite state $\rho_{AB}$ \cite{christandl2004squashed},
\begin{align}
\label{eq:squashed-entanglement}
        E_{\sq}(A:B)_{\rho} \coloneq \frac{1}{2} \inf_{\hat{\rho}_{ABE}}  \left\{I(A:B \vert E)_{\hat{\rho}} : \hat{\rho}_{AB}= \rho_{AB} \right\}\ .
\end{align} 
Due to the factor of $1/2$, the squashed entanglement is always upper bounded by the logarithm of the minimum of the local dimensions, rather than twice that.

With the physical and technical motivation and methods addressed, we state and prove our theorem.
\begin{theorem}\label{thm:q-mem-lims-from-ent-dist-with-LOCC-maps}
    Let $\ve \in (0,0.07)$, and let $\left(\cG_{n}\right)_{n}$ and $ \left(\wt{\cG}_{n}\right)_{n}$ be inhomogeneous Markov chains defined by channels $\cN^{i}_{A_{i} \to A_{i+1}}$ and $\wt{\cN}^{i}_{B_{i} \to B_{i+1}}$, respectively. For $n \geq \hat{n}$ where 
    \begin{align}
        \hat{n} \coloneq \min\!\left\{n \in \mbb{N}: \left[\prod_{i \in [n]} \eta_{\CMI}(\cN^{i})\eta_{\CMI}(\wt{\cN}^{i}) \right] < \frac{1-\sqrt{\ve} - g(\sqrt{\ve})}{\log\min\{\vert A_{1} \vert, \vert B_{1} \vert\}} \right\} \ ,
    \end{align}
    where $g(x)\coloneqq (x+1)\log (x+1) - x \log x$ is the bosonic entropy function, it is the case $E^{\ve}_{D,\operatorname{LOCC}}\left(\cG_{n} \otimes \wt{\cG}_{n}(\rho_{A_{1}B_{1}})\right) < 1$ for all quantum states $\rho_{A_{1}B_{1}}$. That is, for every input quantum state to the quantum memories, after time $\hat{n}$ the users cannot distill a state with singlet fraction greater than $1-\ve$ using LOCC post-processing.

    In particular, for homogeneous Markov chains $\cN_{A \to A}$ and $\cM_{B \to B}$, 
    \begin{align}
        \hat{n} &= \left \lceil 
        \frac{\log\left(\frac{\log\min\{\vert A_{1} \vert, \vert B_{1} \vert\}}{1-\sqrt{\ve} - g(\sqrt{\ve})}\right)}
        {\log\!\left(\frac{1}{\eta_{\CMI}(\cN)\eta_{\CMI}(\cM)}\right)}\right\rceil \ . 
    \end{align}
\end{theorem}

\begin{proof}
    For a state $\rho_{AB}$ that is stored in the quantum memories, we have the sequence of inequalities
    \begin{align}
        E_{D,\LOCC}^{\ve}(\rho) &\leq \frac{1}{1-\sqrt{\ve}}\left( E_{\sq}(A_{n}:B_{n})_{(\cG_{n} \otimes \wt{\cG}_{n})(\rho)} + g(\sqrt{\ve}) \right) \\
        &\leq \frac{1}{1-\sqrt{\ve}}\left( \frac{1}{2}I(A_{n}:B_{n} \vert E)_{(\cG_{n} \otimes \wt{\cG}_{n} \otimes \id_{E})(\rho)} + g(\sqrt{\ve}) \right) \\
        &\leq \frac{1}{1-\sqrt{\ve}}\left( \frac{1}{2}\prod_{i \in \{1,\ldots ,n\}} \left[\eta_{\CMI}(\cN^{i})\eta_{\CMI}(\wt{\cN}^{i}) \right]I(A_{1}:B_{1} \vert E)_{\rho} + g(\sqrt{\ve}) \right) \\
        &\leq \frac{1}{1-\sqrt{\ve}}\left( \prod_{i \in \{1,\ldots ,n\}} \left[\eta_{\CMI}(\cN^{i})\eta_{\CMI}(\wt{\cN}^{i}) \right] \log\left(\min\{\vert A_{1} \vert, \vert B_{1} \vert \right) + g(\sqrt{\ve}) \right) \ , \label{eq:ent-distil-with-LOCC}
    \end{align}
    where the first inequality follows from \cite[Theorem 13.9]{khatri2024principlesquantumcommunicationtheory},
    the second inequality from choosing an arbitrary extension of the initial $\rho_{AB}$, the fourth from the definition of the CMI contraction coeffficient, and the final inequality from the standard dimension bound on the conditional mutual information. Since the above inequality holds for an arbitrary initial state, it follows that if this final upper bound is strictly bounded above by $1$, then no entanglement can be distilled using LOCC from the state shared between the quantum memories. As $\frac{g(\sqrt{\ve})}{1-\sqrt{\ve}} \geq 1$ for $\ve \in (0.075,1]$, we restrict $\ve \in (0,0.07)$ in the theorem statement. Thus, bounding \eqref{eq:ent-distil-with-LOCC} from above by 1 and rearranging terms completes the proof.
\end{proof}

\subsubsection{Limits on preserving resources for secret key with LOCC post-processing}

It is known that distilling key from quantum states can be easier than distilling entanglement. Thus, the previous theorems do not imply that the memory is useless for preserving states enough to perform key distillation. While obviously if the users knew at the start they wanted to distill secret key, they would do so immediately rather than storing the quantum resource in a memory, if the users stored the quantum state and \textit{later} decided they wanted to distill key from it, then bounds on distillable key are relevant. This is our motivation for establishing such limits on key distillation.

We begin with our limits on key distillation when using local operations and public communication (LOPC). This will again make use of the squashed entanglement \eqref{eq:squashed-entanglement}. This is because it is known to serve as an upper bound on the asymptotics of key distillation and an approximate upper bound in the one-shot setting \cite{Wilde_squashed-ent}.

\begin{theorem}
\label{thm:limits-on-key-distil-with-two-way}
    Let $\ve \in (0,0.01)$, and let $\left(\cG_{n}\right)_{n}$ and $ \left(\wt{\cG}_{n}\right)_{n}$ be inhomogeneous Markov chains defined by channels $\cN^{i}_{A_{i} \to A_{i+1}}$ and $\wt{\cN}^{i}_{B_{i} \to B_{i+1}}$, respectively. One cannot distill a state of fidelity $> 1 - \ve$ to a secret key using LOPC from any state stored in the memories for any time larger than or equal to
    \begin{align}
        \min\!\left\{n: \left[\prod_{i \in [n]} \eta_{\CMI}(\cN^{i})\eta_{\CMI}(\wt{\cN}^{i}) \right] < \frac{1-2\sqrt{\ve} -2g(\sqrt{\ve})}{\log\min\{\vert A_{1} \vert, \vert B_{1} \vert\}} \right\} \ .
    \end{align}

    In particular, for homogeneous Markov chains $\cN_{A \to A}$ and $\cM_{B \to B}$, one cannot distill a state of fidelity $> 1 - \ve$ to a secret key using LOCC from any state stored in the memories past the critical time
    \begin{align}
        n_{c} &= \left \lceil 
        \frac{\log\left(\frac{\log\min\{\vert A_{1} \vert, \vert B_{1} \vert\}}{1-2\sqrt{\ve} -2g(\sqrt{\ve})}\right)}
        {\log\!\left(\frac{1}{\eta_{\CMI}(\cN)\eta_{\CMI}(\cM)}\right)}\right\rceil .
    \end{align}
\end{theorem}

\begin{proof}
    Following the notation of \cite{Wilde_squashed-ent}, in the framework of key distillation, there exists a one-shot key distillation protocol for state $\rho_{AB}$ that distills $\lfloor \log_{2}\vert \widehat{A}\vert \rfloor$ bits of key to tolerated error $\ve \in (0,1)$ if there exists a quantum channel $\cL_{AB \to \widehat{A}A'\widehat{B}B'}$ implementable by two-way LOCC such that
    \begin{align}
        F(\cL(\rho), \gamma_{\widehat{A}A'\widehat{B}B'}) \geq 1 - \ve \ , 
    \end{align}
    where $\gamma_{\widehat{A}A'\widehat{B}B'}$ is a private state.\footnote{A private state is one such that there exist local operations that allow Alice and Bob to distill a key containing $\log_{2}|\widehat{A}|$ bits. See \cite{Wilde_squashed-ent} and references therein for more details.} Thus our goal is to show that, for the tolerated error $\ve$, it must be the case $\log\vert \widehat{A} \vert < 1$ regardless of the state initially stored in the quantum memories. To do this, let $\cL$ be an LOCC channel that achieves the tolerated error $\ve$ for dimension of key, $\vert \widehat{A} \vert$. Then, the following sequence of inequalities holds:
    \begin{align}
        \log \vert \widehat{A} \vert &\leq E_{\text{sq}}(\widehat{A}A':\widehat{B}B')_{\cL(\cG_{n} \otimes \wt{\cG}_{n}(\rho))} + f_{1}(\sqrt{\ve},\vert \widehat{A} \vert) \\ 
        &\leq E_{\text{sq}}(A_{n}:B_{n})_{\cG_{n} \otimes \wt{\cG}_{n}(\rho)} + f_{1}(\sqrt{\ve},\vert \widehat{A} \vert) \\
        &\leq \frac{1}{2}I(A_{n}:B_{n} \vert E)_{\cG_{n} \otimes \wt{\cG}_{n} \otimes \id_{E}(\rho)} + f_{1}(\sqrt{\ve},\vert \widehat{A} \vert) \\
        &\leq \prod_{i \in \{1,\ldots ,n\}} \left[\eta_{\CMI}(\cN^{i})\eta_{\CMI}(\wt{\cN}^{i}) \right] \frac{1}{2}I(A:B \vert E) + f_{1}(\sqrt{\ve},\vert \widehat{A} \vert) \\
        &\leq \prod_{i \in \{1,\ldots ,n\}} \left[\eta_{\CMI}(\cN^{i})\eta_{\CMI}(\wt{\cN}^{i}) \right]\log\min\{\vert A_{1} \vert , \vert B_{1} \vert \} + f_{1}(\sqrt{\ve},\vert \widehat{A} \vert) \ , 
    \end{align}
    where the first inequality follows from \cite[Theorem 2]{Wilde_squashed-ent} with
    \begin{equation} \label{eq:f_1_def}
     f_1(\delta, K) \coloneq 2 \delta \log(K) + 2 g(\delta)   ,
    \end{equation}
     the second because squashed entanglement monotonically decreases under LOCC
    \cite{christandl2004squashed}, the third from choosing an arbitrary extension of the initial state $\rho_{AB}$, the fourth from the definition of the CMI contraction coefficient, and the final inequality from the standard dimension bound on conditional mutual information. Next, using the definition of $f_{1}$ in~\eqref{eq:f_1_def} 
    and re-arranging terms, we rewrite the above inequality as
    \begin{align}
        \log \vert \widehat{A} \vert \leq \frac{ \prod_{i \in \{1,\ldots ,n\}} \left[\eta_{\CMI}(\cN^{i})\eta_{\CMI}(\wt{\cN}^{i}) \right]\log\min\{\vert A_{1} \vert , \vert B_{1} \vert \} + 2g(\sqrt{\ve})}{1-2\sqrt{\ve}} \ .
    \end{align}
    It thus suffices to upper bound the right-hand side by $1$. As $\frac{2g(\sqrt{\ve})}{1-2\sqrt{\ve}} \geq 1$ for $\ve \geq 0.010989$, we restrict to $\ve \in (0,0.01)$ in the theorem statement. The rest of the proof is arithmetic.
\end{proof}

While the operational relevance is less clear, we also prove bounds when the users are only willing to use one-way LOCC during post-processing.
\begin{theorem}
\label{thm:limits-for-one-way-key-distil}
    Let $\ve \in (0,0.14)$, and let $\left(\cG_{n}\right)_{n}$ and $ \left(\wt{\cG}_{n}\right)_{n}$ be inhomogeneous Markov chains defined by channels $\cN^{i}_{A_{i} \to A_{i+1}}$ and $\wt{\cN}^{i}_{B_{i} \to B_{i+1}}$, respectively. One cannot distill a state of fidelity $> 1 - \ve$ to a secret key using one-way LOCC from any state stored in the memories once the time is
    \begin{align}
        \min\!\left\{n: \left[\prod_{i \in [n]} \eta_{\MI}(\cN^{i})\eta_{\MI}(\wt{\cN}^{i}) \right] < \frac{\ell(\ve)}{2\log\min\{\vert A_{1} \vert, \vert B_{1} \vert\}} \right\} \ ,
    \end{align}
    where
    \begin{align}
        \ell(\ve) \coloneq (1-\ve)\log(1/\kappa(\ve)) - h(\ve) \text{ and } \kappa(\ve) \coloneq  \frac{1}{4} + \frac{\ve}{2} + \frac{\sqrt{3\ve(1-\ve)}}{2} \ .
    \end{align}

    In particular, for homogeneous Markov chains $\cN_{A \to A}$, $\cM_{B \to B}$, one cannot distill a state of fidelity $> 1 - \ve$ to a secret key using one-way LOCC from any state stored in the memories past the time
    \begin{align}
        \hat{n} = \left \lceil \frac{\log\!\left(\frac{2\log[\min\{\vert A \vert, \vert B \vert\}]}{\ell(\ve)}\right)}{ \log\!\left(\frac{1}{\eta_{\MI}(\cN)\eta_{\MI}(\cM)}\right)} \right \rceil \ . 
    \end{align}
\end{theorem}

\begin{proof}
    By \cite[Theorem 2]{Singh_2025}, for a state $\rho_{AB}$ and Hilbert space $B' \cong B$, if
    \begin{align}
        2^{-\inf_{\substack{\sigma_{ABB'}: \\\sigma_{AB} = \rho_{AB}}} D_{h}^{\ve}(\rho \Vert \Tr_{B}[\sigma])} > \kappa(\ve) \ , 
    \end{align}
    then there does not exist a one-way LOCC protocol that can distill any key. By basic arithmetic, we can re-express this condition as
    \begin{align}
        \inf_{\substack{\sigma_{ABB'}: \\\sigma_{AB} = \rho_{AB}}} D_{h}^{\ve}(\rho \Vert \Tr_{B}[\sigma]) < \log(1/\kappa(\ve)) \ .
    \end{align}
    Moreover, note
    \begin{align}
        \inf_{\substack{\sigma_{ABB'}: \\ \sigma_{AB} = \rho_{AB}}} D_{h}^{\ve}(\rho \Vert \Tr_{B}[\sigma]) \leq D^{\ve}_{h}(\rho_{AB} \Vert \rho_{A} \otimes \rho_{B}) \leq \frac{1}{1-\ve}\left[I(A:B) + h(\ve) \right]
    \end{align}
    where the first inequality follows from considering $\sigma_{ABB'} \coloneq \rho_{AB} \otimes \rho_{B'}$ and the second from \cite[Eq.~(2)]{Wang_2012}. It thus suffices that 
    \begin{align}
        \frac{1}{1-\ve}\left[I(A:B) + h(\ve) \right] < \log(1/\kappa(\ve)) \iff I(A:B) < \ell(\ve) \ 
    \end{align}
    for there to be no one-way distillable key to the appropriate tolerated error $\ve$. We remark that $\ell(\ve) > 0$ for $\ve \in (0,0.14)$, which selects for our parameter range. The rest of the proof is identical to that of \Cref{thm:q-mem-lims-from-ent-dist-with-NE-maps}.
\end{proof}

\subsection{Limits on quantum memories for quantum secret sharing schemes}\label{subsec:limits-on-QSS}
\begin{figure}
    \centering
    \includegraphics[width=0.8\linewidth]{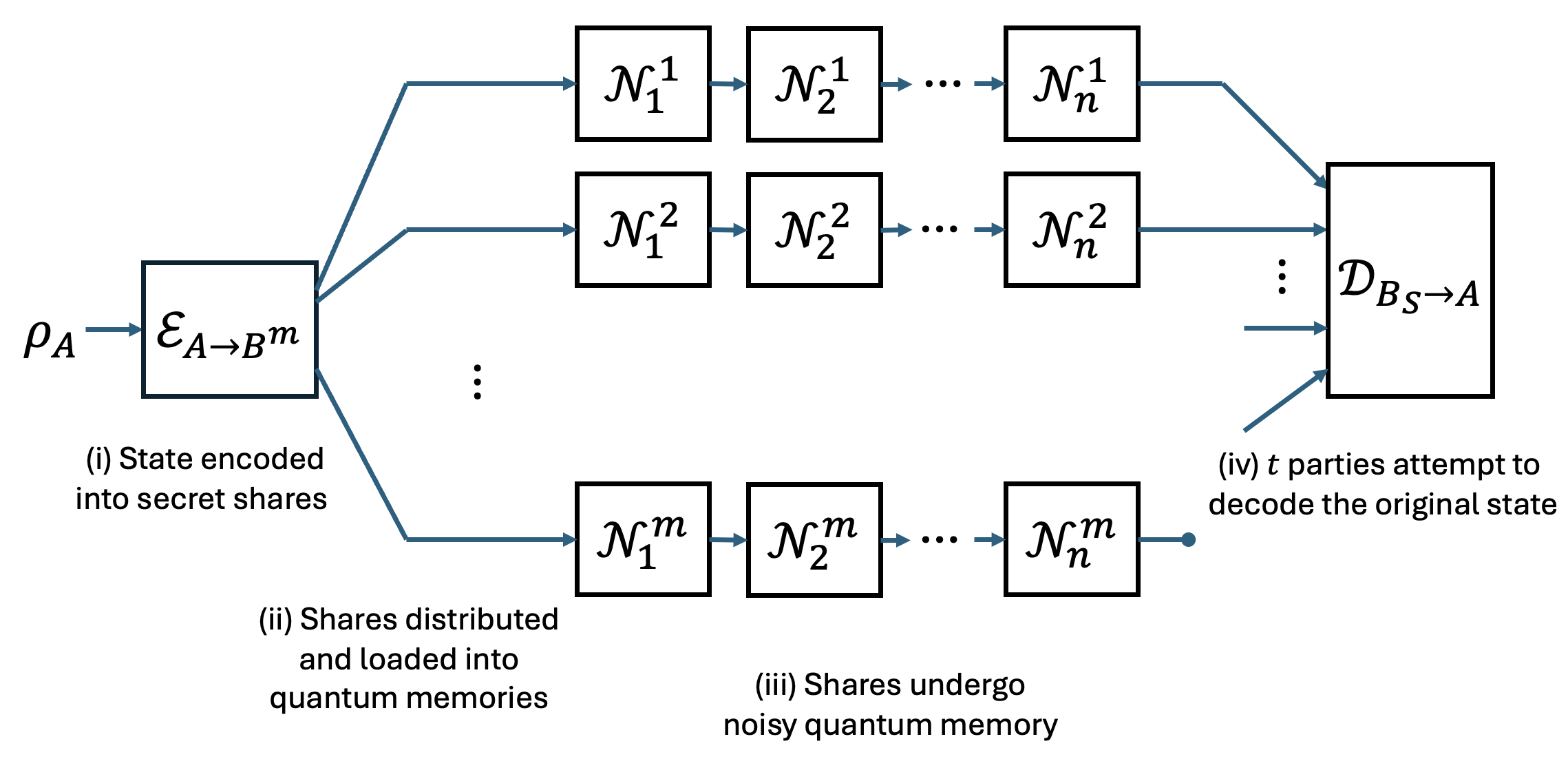}
    \caption{Depiction of model considered in Subsection \ref{subsec:limits-on-QSS}.}
    \label{fig:placeholder}
\end{figure}

In the previous section, we established limits on quantum memories by their ability to preserve the resources we want quantum states to contain even with post-processing. A similar use of a quantum memory is to store a quantum state that acts as part of a cryptographic scheme. An example of such a state consists of the shares in a quantum secret sharing scheme. To see this, recall that a $(t,m)$-quantum secret sharing (QSS) scheme is a method for mapping a single quantum system into a multipartite state made up of $m$ subsystems such that one can perfectly recover the original system by acting on only $t$ of the subsystems, but cannot recover the initial state ``much at all" with strictly less than $t$ subsystems \cite{Hillery99,gottesman2000theory}. 

Formally, consider Hilbert spaces $A$ and $\{B_{i}\}_{i \in \{1,\ldots,m\}}$, and let $B^{m} \coloneq \bigotimes_{i \in \{1,\ldots,m\}} B_{i}$. Note that $\vert B_{i} \vert$ need not be the same for each $i$. For every $S \subset \{1,\ldots,m\}$, define $B_{S} \coloneq \bigotimes_{i \in S} B_{i}$, let $S^{c}$ denote the complement of $S$ with respect to $\{1,\ldots,m\}$, and let $\cS_{t} \coloneq \{S \subset \{1,\ldots,m\}: \vert S \vert = t\}$. A $(t,m)$-quantum secret sharing scheme is then specified by an encoder $\cE_{A \to B^{m}}$ and decoders $\left\{\cD_{B_{S} \to A}\right\}_{S \in \cS_{t}}$. One may then require that the secret be recoverable with $t$ shares by demanding the channel fidelity of the encoding and decoding channels with the identity channel be large, i.e.,
\begin{align}
    1 - \ve \leq F(\cD_{B_{S} \to A} \circ \Tr_{B_{S^{c}}} \circ \cE, \id_{A}) \quad \forall S \in \cS_{t} \ ,
\end{align}
where $F(\cN,\cM) \coloneq \inf_{\rho_{RA}} F\!\left((\id_{R} \otimes \cN)(\rho), (\id_{R} \otimes \cM)(\rho)\right)$.

The concern we consider is then the following. To implement a $(t,m)$-quantum secret sharing scheme, the encoding $\cE$ will be performed, then the $m$ subsystems will be sent to the respective $m$ parties, who then will need to store their local system in their local quantum memory until a later time during which at least $t$ parties agree to gather their systems to decode the secret. One would expect that if the local quantum memories inject too much noise into the global state, the initial secret cannot be recovered. To capture this idea, one can ignore the intended decoders of the scheme and simply determine if decoders even \textit{exist}. Aligning with~\cite{Ouyang_2023}, we identify this as $\ve$-recoverability, which by itself does not actually require the channel to be a $(t,m)$ QSS scheme.
\begin{definition}\label{def:epsilon-recoverable-encoder}
    Let $\cE_{A \to B^{m}}$ and $\cN_{B^{m} \to \wt{B}^{m}}$ be quantum channels. We say $\cE$ is $\ve$-recoverable for $(t,m)$-sharing under noise model $\cN$ if
    \begin{align}
        1 - \ve \leq \min_{S \in \cS_{t}} \max_{\cD_{\wt{B}_{S} \to A}} F(\cD \circ \Tr_{\wt{B}_{S^{c}}} \circ \cN \circ \cE, \id_{A}) \ ,
    \end{align}
    where the maximization is over all decoding channels from $\wt{B}_{S} \coloneq \bigotimes_{i \in S} \wt{B}_{i}$ to $A$.
\end{definition}

The following shows that the MI contraction coefficient of the noise on $t$ memories controls the depth at which no secret sharing scheme will be $\ve$-recoverable for $(t,m)$-sharing given noisy quantum memories.

\begin{theorem}
\label{thm:QSS-no-go-theorem}
    Consider $m$ users, where user $j$ has a quantum memory described by a sequence of channels $\left(\cN_{j}^{i}\right)_{j \in \mbb{N}}$, and define $\cN^{i}_{S} \coloneqq \bigotimes_{j \in S} \cN_{j}^{i}$. Let $d \in \mbb{N}$, $\delta \in (0, \frac{d-1}{2d})$, and $\ve \in (0, \frac{d-1}{2d}-\delta)$. Then, for all $n\in\mathbb{N}$ larger than 
    \begin{align}
        \min\!\left\{n \in \mbb{N} : \exists S \in \cS_{t} : 
        \prod_{i \in \{1,\ldots,n\}} \eta_{\MI}(\cN^{i}_{S}) \leq \frac{\delta^{2}}{2\ln(2)\log(\min\{\vert B_{S} \vert, \vert A \otimes B_{S^{c}} \vert \})} \right\}
    \end{align}
    there does not exist an encoding channel that is $\ve$-recoverable for $(t,m)$ sharing under noise model $\cN^{n} \coloneq \bigcirc_{i \in \{1,\ldots,n}\} \bigotimes_{j \in \{1,\ldots,m\}} \cN^{i}_{j}$. 
\end{theorem}

A limitation of the practicality of computing the bound in \cref{thm:QSS-no-go-theorem} is that it is expressed in terms of the contraction coefficient of product channels, i.e.,~$\eta_{\MI}(\cN^{i}_{S}) = \eta_{\MI}(\bigotimes_{j \in S} N_{j}^{i})$. This however makes sense from a coding theory perspective. First, a (quantum) secret sharing scheme can be built from an error-correcting code. Second, one way of seeing that contraction coefficients should not tensorize is that one can design error correcting codes to preserve the distinguishability of inputs over product channels. As such, the encoding of the secret sharing scheme may act simultaneously as a secret sharing scheme and an error correcting code for the noise in the network of memories, and this is captured in \cref{thm:QSS-no-go-theorem}. Nonetheless, one can obtain computable bounds that only look at the channels individually by using the $\alpha_{+}$ Doeblin coefficient.
\begin{corollary}
    Consider $m$ users where user $j$ has a quantum memory described by a sequence of channels $\left(\cN_{j}^{i}\right)_{j \in \mbb{N}}$ and define $\cN^{i}_{S} \coloneqq \bigotimes_{j \in S} \cN_{j}^{i}$. Let $d \in \mbb{N}$, $\delta \in \left(0, \frac{d-1}{2d} \right)$, and $\ve \in \left(0, \frac{d-1}{2d}-\delta\right)$. Then, for all $n\in\mathbb{N}$ larger than 
    \begin{align}
        \min\!\left\{n \in \mbb{N} : \exists S \in \cS_{t} : 
        \prod_{i \in \{1,\ldots,n\}} \left(1 - \prod_{j \in S}\alpha_{+}(\cN^{i}_{j}) \right) \leq \frac{\delta^{2}}{2\ln(2)\log(\min\{\vert B_{S} \vert, \vert A \otimes B_{S^{c}} \vert \})} \right\}
    \end{align}
    there does not exist an encoding channel that is $\ve$-recoverable for $(t,m)$ sharing under noise model $\cN^{n} \coloneq \bigcirc_{i \in \{1,\ldots,n\}} \bigotimes_{j \in \{1,\ldots,m\}} \cN^{i}_{j}$.
\end{corollary}

\begin{proof}
    This follows from \cref{thm:QSS-no-go-theorem} and
    \begin{equation}
    \eta_{\MI}(\cN \otimes \cM) \leq \eta_{\CMI}(\cN \otimes \cM) \leq 1-\alpha_{+}(\cN \otimes \cM) \leq 1-\alpha_{+}(\cN)\alpha_{+}(\cM) \ , 
    \end{equation}
    where the last inequality follows from \cite[Item 2 of Lemma III.8]{hirche2024quantum} and so if we have the stated bound, we have bounded the contraction coefficients in \cref{thm:QSS-no-go-theorem}.
\end{proof}

We now turn to proving \Cref{thm:QSS-no-go-theorem}. To do so, we begin by recalling the main definitions and lemmas of which we will make use. A state $\rho_{AB}$ is $k$-extendible on $B$ if there exists $\rho_{AB_{1}B_{2}\cdots B_{k}}$ where $B_{i} \cong B$ for all $i \in [k]$, and $\rho_{AB_{i}} = \rho_{AB}$ for all $i \in [k]$. A $k$-symmetrically extendible state $\rho_{A'A}$ on $A$ where $A' \cong A$  has bounded overlap with the maximally entangled state (also known as singlet fraction).
\begin{proposition}
\label{prop:singlet-frac-bound-for-k-sym-ext-states}
    Let $\rho_{A'A}$ be $k$-symmetrically extendible, and let $A' \cong A$. Then
    \begin{align}
        \bra{\Phi_{d}}\rho_{A'A}\ket{\Phi_{d}} \leq \frac{1}{d}\left(1 + \frac{d-1}{k}\right) \ , 
    \end{align}
    where $d = \vert A \vert$.
\end{proposition}

\begin{proof}
    Let $\rho_{A'A}$ be $k$-symmetrically extendible on $A$. Let $\rho_{A'A_{1}A_{2}\cdots A_{k}}$ be its extension. Note that by partial trace
    \begin{align}
        \wt{\rho}_{A'A_{1}A_{2}\cdots A_{k}} \coloneq \int dU \left(U \otimes \bigotimes_{i \in \{1,\cdots ,k\}} \overline{U} \right) \rho \left(U \otimes \bigotimes_{i \in \{1,\cdots ,k\}} \overline{U} \right)^{\ast} 
    \end{align}
    is a $k$-symmetric extension on $A$ of the state
    \begin{align}
        \wt{\rho}_{A'A} \coloneq \int dU (U \otimes \overline{U}) \rho (U \otimes \overline{U})^{\ast} \eqqcolon \cT(\rho_{A'A}) \ .
    \end{align}
    Thus, twirling preserves $k$-symmetric extendibility. Now, we may write
    \begin{align}
        \bra{\Phi_{d}}\rho_{A'A}\ket{\Phi_{d}} &= \Tr[\dyad{\Phi_{d}}\rho_{AA'}] \\
        &= \Tr[\dyad{\Phi_{d}}\cT(\rho_{A'A})] \\
        &\leq \frac{1}{d} \left(1 + \frac{d-1}{k} \right),
    \end{align}
    where we used the invariance of the maximally entangled state, that $\cT(\rho_{A'A})$ is an isotropic state $t\dyad{\Phi_{d}} + \frac{1-t}{d^{2}-1}(I-\dyad{\Phi_{d}})$ after twirling, and that an isotropic state has parameter $t$ at most $\frac{1}{d}(1 + \frac{d-1}{k})$ if it is $k$-symmetrically extendible~\cite{Johnson_2013,Kaur_2021}. 
\end{proof}

Note that $k$-symmetric extendibility of $\rho_{AB}$ on the $B$ system means that each $B_{i}$ is equally correlated with $A$. 
Also observe that the CMI of any extension of $\rho_{AB}$ is an upper bound on the squashed entanglement, as follows from \eqref{eq:squashed-entanglement}. These ideas are related through the fact that the squashed entanglement of a state is related to its $k$-extendibility, as captured in the following lemma.

\begin{lemma}[Theorem 1 of \cite{li2018squashed}]\label{lem:sq-ent-bound-k-sym-ext}
    Let $\rho_{AB}$ be a quantum state such that $E_{\sq}(A:B)_{\rho} \leq \varepsilon$. Then for every $k \in \mbb{N}$, there exists a $k$-extendible state $\sigma_{AB}$ such that $\left\Vert \rho - \sigma \right\Vert_{1} \leq (k-1)\sqrt{2\ln(2)}\sqrt{\varepsilon}$.
\end{lemma}

Now we prove the theorem.

\begin{proof}[Proof of \cref{thm:QSS-no-go-theorem}]
    We begin with a summary of the proof. First, we show that one can bound the $\ve$-recoverability of the encoder $\cE$ under a noise model $\cN$ by the singlet fraction of the decoded state. Second, we show that if the system after the noise is $\delta$-close to a $2$-symmetrically extendible state, then the state will still be $\delta$-close to a $2$-symmetrically extendible state after decoding. Combining these with \Cref{prop:singlet-frac-bound-for-k-sym-ext-states}, if the state after the noise is $(0,\frac{d-1}{2d}) \ni \delta$-close to a 2-symmetrically extendible state, then it is impossible for the encoding channel to be $(0, \frac{d-1}{2d}-\delta) \ni \ve$-recoverable. Finally, we use MI contraction coefficients to control the time it takes for a noise  model $\cN^{n} \coloneq \bigcirc_{i \in [n]} \bigotimes_{j \in [m]} \cN^{i}_{j}$ to make the state $\delta$-close to a symmetrically extendible state.

    First, we bound the $\ve$-recoverability of a channel $\cE_{A \to B^{m}}$ under a noise model $\cN$ in terms of the singlet fraction of the decoded subset. Fix subset $\widehat{S} \subset \{1,\ldots ,m\}$ of size $t$. Define the state
    \begin{align}
        \rho_{A'B_{\widehat{S}}B_{\widehat{S}^{c}}} = \rho_{A'B^{m}} \coloneq (\id_{A'} \otimes \cN \circ \cE)(\dyad{\Phi_{d}}_{A'A}) \ , 
    \end{align}
    where $M = \vert A \vert$.
    Then,
    \begin{align}
        &\min_{S \in \cS_{t}} \max_{\cD_{B_{S} \to A}} F(\cD \circ \Tr_{B_{S^{c}}} \circ \cN \circ \cE, \id_{A}) \notag \\
        &\leq \max_{\cD_{B_{\widehat{S}}\to A}} F(\cD \circ \Tr_{B_{\widehat{S}^{c}}} \circ \cN \circ \cE, \id_{A}) \\
        &= \max_{\cD_{B_{\widehat{S}}\to A}} \inf_{\rho_{RA}} F\left((\id_{R} \otimes \cD \circ \Tr_{B_{\widehat{S}^{c}}} \circ \cN \circ \cE)(\rho_{RA}),\rho_{RA} \right) \\
        &\leq \inf_{\rho_{RA}} \max_{\cD_{B_{\widehat{S}}\to A}} F\left((\id_{R} \otimes \cD \circ \Tr_{B_{\widehat{S}^{c}}} \circ \cN \circ \cE)(\rho_{RA}),\rho_{RA} \right) \\
        &\leq \max_{\cD_{B_{\widehat{S}}\to A}} F(\cD \circ \Tr_{B_{\widehat{S}^{c}}} \circ \cN \circ \cE(\dyad{\Phi_{d}}),\dyad{\Phi_{d}}) \\
        &= \bra{\Phi_{d}} \widehat{\cD} \circ \Tr_{B_{\widehat{S}^{c}}} \circ \cN \circ \cE(\dyad{\Phi_{d}}) \ket{\Phi_{d}} \\
        &= \bra{\Phi_{d}} (\widehat{\cD} \circ \Tr_{B_{\widehat{S}^{c}}})(\rho_{A'B_{\widehat{S}}B_{\widehat{S}^{c}}}) \ket{\Phi_{d}} \ , \label{eq:no-go-for-QSS-step-1}
    \end{align}
    where the first inequality follows from choosing the partition, the second from the max-min inequality, the third by choosing the input state to be the maximally entangled state, and the second equality is a property of the fidelity when one argument is pure as well as letting $\widehat{\cD}$ represent a maximizing decoding channel.
    
    To bound recoverability, we want to bound~\eqref{eq:no-go-for-QSS-step-1}. To that end, we now show $\rho_{A'A} \coloneq (\cD \circ \Tr_{B_{\widehat{S}^{c}}})(\rho_{A'B_{\widehat{S}}B_{\widehat{S}^{c}}})$ is $\delta$-close from a two-symmetrically extendible state so long as $\rho_{A'B_{\widehat{S}}B_{\widehat{S}^{c}}}$ is. Assume there is $\sigma_{A'B_{\widehat{S}}B_{\widehat{S}^{c}}}$ that is two-symmetrically extendible on the $B_{\widehat{S}}$ space such that $\left\Vert \rho_{A'B_{\widehat{S}}B_{\widehat{S}^{c}}} - \sigma_{A'B_{\widehat{S}}B_{\widehat{S}^{c}}} \right\Vert_{1} \leq \delta$. Denote the extension by $\sigma_{A'B_{\widehat{S}}B_{\widehat{S}^{c}}B_{\widehat{S}}'}$. By definition of two-symmetric extendibility, $\sigma_{A'B_{\widehat{S}}} \cong \sigma_{A'B_{\widehat{S}}'}$, so $\sigma_{A'B_{\widehat{S}}B_{\widehat{S}}'}$ is itself a two-symmetric extension of $\sigma_{A'B_{\widehat{S}}}$. Thus, 
    \begin{align}
        \left\Vert \rho_{A'B_{\widehat{S}}} - \sigma_{A'B_{\widehat{S}}} \right\Vert_{1} & = \left\Vert \Tr_{B_{\widehat{S}^{c}}}[\rho_{A'B_{\widehat{S}}B_{\widehat{S}^{c}}} - \sigma_{A'B_{\widehat{S}}B_{\widehat{S}^{c}}}] \right\Vert_{1} \leq \delta \ .
    \end{align}
    Thus, $\rho_{A'B_{\widehat{S}}}$ is $\delta$-close from a state that is two-symmetrically extendible. Furthermore, $\sigma_{A'AA_{2}} \coloneq (\id_{A'} \otimes \cD_{B_{\widehat{S}}\to A} \otimes \cD_{B_{\widehat{S}}' \to A})(\sigma_{A'B_{\widehat{S}}B_{\widehat{S}}'})$ satisfies $\sigma_{A'A} \cong \sigma_{A'A_{2}}$, so it is two-symmetrically extendible on the $A$ space. Thus,
    \begin{align}
        \left\Vert \rho_{A'A} - \sigma_{A'A} \right\Vert_{1} = \left\Vert \cD_{B_{\widehat{S}}}[\rho_{A'B_{\widehat{S}}} - \sigma_{A'B_{\widehat{S}}}] \right\Vert_{1} \leq \delta \ . 
    \end{align}
    It then follows by monotonicity of the Schatten $1$-norm under CPTNI maps that 
    \begin{align}
        \vert \bra{\Phi_{d}} \rho_{A'A} \ket{\Phi_{d}} - \bra{\Phi_{d}} \sigma_{A'A} \ket{\Phi_{d}} \vert \leq \delta \ . 
    \end{align}
    Applying \cref{prop:singlet-frac-bound-for-k-sym-ext-states}, 
    \begin{align}
        \vert \bra{\Phi_{d}} \rho_{A'A} \ket{\Phi_{d}} \vert \leq \delta + \frac{1}{d} + \frac{d-1}{2d} \ . 
    \end{align}
    Thus, if $\rho_{A'B_{\widehat{S}}B_{\widehat{S}^{c}}}$ is $\delta$-close to a two-symmetrically extendible state on the $A$ system, then the $\ve$-recoverability of $\cE$ is limited by
    \begin{align}
        1 - \ve \leq \delta + \frac{1}{d} + \frac{d-1}{2d} \iff \ve \geq \frac{(d-1)}{2d} - \delta \ .
    \end{align}
    Thus, the bounds on $\delta$ and $\ve$ that are specified in the theorem statement are chosen so that the right-hand side of the second inequality is positive and $\ve$ does not satisfy the inequality.

     Given the above, it suffices to guarantee that $\rho_{A'B_{\widehat{S}}B_{\widehat{S}^{c}}}$ is $\delta$-close to a two-symmetrically extendible state. First, define the extension
     \begin{align}
        \rho_{A'B_{\widehat{S}}B_{\widehat{S}^{c}}} \otimes \dyad{\psi}_{E} = \rho_{A'B^{m}E} \coloneq (\id_{A'} \otimes \cN \circ \cE \otimes \id_{E})(\dyad{\Phi_{d}}_{A'A} \otimes \dyad{\psi}_{E}) \ . 
     \end{align}
     Note that at no point does the $E$ system become correlated with the others. We use the shorthand $\cN^{i}_{S} \coloneq \bigotimes_{j \in S} \cN^{i}_{j}$. Then, using the independence of the $E$ system and the definition of the MI contraction coefficient,
    \begin{align}
        \frac{1}{2}I(B_{\widehat{S}}:AB_{\widehat{S}^{c}} \vert E)_{\rho} &= \frac{1}{2}\left[I(B_{\widehat{S}}:A'B_{\widehat{S}^{c}}E) - I(B_{\widehat{S}}:E) \right] \\
        &=\frac{1}{2}I(B_{\widehat{S}}:A'B_{\widehat{S}^{c}}) \\
        &\leq \prod_{i \in \{1,\ldots ,n\}} \left[ \eta_{\MI}(\cN^{i}_{\widehat{S}})\eta_{\MI}(\id_{A'} \otimes \cN^{i}_{\widehat{S}^{c}}) \right] \frac{1}{2}I(B_{\widehat{S}}:A'B_{\widehat{S}^{c}})_{\cE(\dyad{\Phi_{d}})} \\
        &=  \prod_{i \in \{1,\ldots ,n\}} \left[ \eta_{\MI}(\cN^{i}_{\widehat{S}}) \right] \log(\min\{\vert B_{\widehat{S}} \vert, \vert A \otimes B_{\widehat{S}^{c}} \vert \}) \ , \label{eq:CMI-upr-bnd-by-contraction}
    \end{align}
    where we used that $\eta_{\MI}(\id_{A} \otimes \cN) = 1$. 
    
    It follows by \cref{lem:sq-ent-bound-k-sym-ext} that if \eqref{eq:CMI-upr-bnd-by-contraction} is bounded from above by $\delta^{2}/2\ln(2)$, then $\rho_{A'B_{\widehat{S}}B_{\widehat{S}^{c}}}$ is $\delta$-close to a two-symmetrically extendible state. None of this relied on a specific choice of $\widehat{S}$, so one may minimize over it. It also did not depend on the channel $\cE$, so once the condition is satisfied, no encoding $\cE$ can exist.
\end{proof}

\subsection{Quantum channel capacities}

Contraction coefficients and partial orders are closely connected to the structure of channel capacities~\cite{watanabe2012private,Hirche2022contraction,hirche2022bounding,belzig2024reversetypedataprocessinginequality}. Here we briefly discuss some connections of our work to the existing literature. 

The quantum capacity is given by
\begin{align}
    Q(\cN) &= \lim_{n \to\infty} \frac{1}{n}Q^{(1)}(\cN^{\otimes n}), \\
    Q^{(1)}(\cN) &\coloneqq   \sup_{\rho} \left\{H(\mathcal{N}(\rho)) - H(\mathcal{N}^c(\rho))\right\}, 
\end{align}
where the supremum is over every channel input state $\rho$ and $\mathcal{N}^c$ denotes a complementary channel of $\mathcal{N}$ (see, e.g., \cite{devetak2005capacity} or \eqref{eq:complementary-ch} below).
Due to the regularization, an important question is for which channels the quantum capacity is additive, meaning for which channels the regularization is unnecessary. The best known condition for additivity is if the channel $\cN$ is degradable \cite{devetak2005capacity}. A potentially more general set of channels for which the quantum capacity is additive is that of informationally degradable channels, as introduced in~\cite{cross2017uniform}, sometimes also referred to as fully quantum less noisy. 

For a channel $\cN_{A \to B}$, let $V_{A \to BE}$ be an isometry such that $\cN_{A \to B}(\cdot)= \Tr_{E}[V(\cdot) V^\dag]$. Then a complementary channel of $\cN$ is given by
\begin{equation}
\cN^c_{A \to E}(\cdot) \coloneq \Tr_B[V(\cdot) V^\dag].
\label{eq:complementary-ch}
\end{equation}
A channel $\cN$ with complementary channel $\cN^c$ is called informationally degradable if, for every state $\rho_{RA}$, 
\begin{align}
    I(R:B) \geq I(R:E), 
\end{align}
or equivalently, 
\begin{align}
    \eta_{\MI}(\cN^c,\cN) \leq 1. 
\end{align}
A second criterion that guarantees additivity is if $\cN$ is regularized less noisy, which is the case if the inequality 
\begin{align}
    I(U:B^n) \geq I(U:E^n)
\end{align}
holds for all $n\geq1$.
The guiding question is now whether there exists a channel that is informationally degradable or regularized less noisy but not degradable. Ref.~\cite{belzig2024reversetypedataprocessinginequality} recently proposed  investigating this question using relative expansion and contraction coefficients, and they reported progress by constructing a channel that is less noisy but not degradable.

As identified in~\cite{Hirche2022contraction}, there is a more restrictive criterion called completely less noisy, also discussed earlier in~\cref{Sec:Partial-Orders}, that implies both informational degradability and regularized less noisy. 
Here we briefly give a variant of the corresponding result in~\cite[Proposition 6.1]{belzig2024reversetypedataprocessinginequality} that constructs a candidate channel for completely less noisy channels.

\begin{proposition}
    Let $\cN$ and $\cM$ be degradable channels. Then for all
    \begin{align}
        p\in\left[ \frac{1}{1+\cetaD_D(\cN,\cM)\left(1-\etaD_D(\cN^c,\cN) \right)},1 \right], 
    \end{align}
    where $\etaD_D(\cN^c,\cN)$ and $\cetaD_D(\cN,\cM)$ are defined as in~\eqref{eq:relative_contraction_channels_cond} and~\eqref{eq:relative_expansion_channel_cond}, the quantum channel
    \begin{align}
        \Psi_{p,\cN,\cM}=p|0\rangle\!\langle0|\otimes\cN + (1-p)|1\rangle\!\langle1|\otimes\cM^c
    \end{align}
    is completely less noisy and hence also informationally degradable and regularized less noisy. 
\end{proposition}
\begin{proof}
    The proof is essentially identical to that of~\cite[Proposition 6.1]{belzig2024reversetypedataprocessinginequality}.
\end{proof}
Note that an example of a completely less noisy, but not degradable, channel has recently been provided in~\cite{Zhu-Incapacity2026}. The construction described here could be used for a guided search for further such examples. 

\section{Applications to discrete-time quantum Markov chains}
\label{Sec:Applications2}

\subsection{On hierarchies of operational times}

\label{sec:Hierarchies-operational-times}

A canonical application of SDPI constants involves bounding mixing times. This application has been generalized to other contraction coefficients subsequently: Ref.~\cite{george2025quantumdoeblincoefficientsinterpretations} showed that one can generalize the notion of mixing time to `decoupling times' and upper bound this time with the complete SDPI constant, or in the case of a time-inhomogeneous Markov chain, the complete contraction coefficient. In a slightly different direction, Ref.~\cite{iyer2025quantum} used reverse contraction coefficients to place lower bounds on mixing times.

In this paper we have introduced various contraction coefficients, and in principle they all capture distinct alternative notions of mixing. The alternative notions of mixing that each constant captures is straightforward, and so we only focus on the definition that directly pertains to the conditional contraction coefficient and SDPI, which we call `replacing.' In particular, we  show that any notion of a time homogeneous Markov chain that implies mixing and is implied by being replacing is asymptotically equivalent (\Cref{prop:asymptotic-equiv-of-types-of-mixing}). This result follows because it turns out that these are all different ways of expressing that the channel converges to a replacer channel. Motivated by this, we show that a channel is mixing if and only if there exists a finite number of channel applications at which the quantum Doeblin coefficient becomes non-trivial (\Cref{cor:mixing-equiv-to-q-Doeblin-being-non-triv-at-some-depth}).

\subsubsection{Replacing, mixing, and their equivalence}

A quantum \textit{homogeneous} Markov chain is a process where, at each time step, the same quantum channel $\cN_{A \to A}$ occurs; i.e., for all $n \in \mbb{N}$, the total process is as follows:
\begin{equation}
\cN^{n} \coloneq \bigcirc_{i \in [n]} \cN = \underbrace{\cN \circ \cdots \circ \cN}_{n \text{ times}} \ .
\end{equation}
We thus will talk of a `homogeneous Markov chain $\cN_{A \to A}$,' as this completely specifies the Markov chain. We say a homogeneous Markov chain $\cN$ is replacing if
\begin{align} 
\lim_{n\to \infty } \left\Vert (\id_{R} \otimes \cN^{n})(\rho_{RA} - \sigma_{RA}) \right\Vert_{1} = \left\Vert \rho_{R} - \sigma_{R} \right\Vert_{1}  \quad \forall \rho_{RA},\sigma_{RA} \ .
\end{align}
We note that this definition is similar to the notion of a homogeneous Markov chain with unique fixed point $\omega$ being `mixing' \cite{Burgarth_2007,burgarth2013ergodic,george2025quantumdoeblincoefficientsinterpretations}: 
\begin{align} 
\lim_{n\to \infty } \left\Vert \cN^{n}(\rho_{A}) - \omega_{A} \right\Vert_{1}=0 \quad \forall \rho_{A} \ . 
\end{align}

The distinction between these two definitions is two-fold. First, the Markov chain being replacing involves a reference system, whereas mixing does not. Second, the Markov chain being replacing does not explicitly include a fixed point, which might make it seem that replacing ought to be measured using the conditional contraction coefficient rather than a conditional SDPI constant. However, we will show that in fact both replacing and mixing identify the same property: $\cN^{n}$ converging to the replacer channel 
\begin{align}
    \cR^{\omega}[X] \coloneqq \Tr[X]\omega \ ,
\end{align}
where $\omega$ is the unique fixed point of $\cN$. 

\paragraph{Norm equivalences for channels} To establish our asymptotic equivalence of replacing and mixing, we require a few notions of norms on channels and their equivalence in finite dimensions. For a Hermitian-preserving map $\Phi_{A \to B}$ we recall the $1 \to 1$ norm and the diamond norm (completely bounded $1 \to 1$ norm)
\begin{align}
    \left\Vert \Phi \right\Vert_{1} &\coloneqq  \max_{\rho_{A} \in \Density(A)} \left\Vert \Phi(\rho) \right\Vert_{1} , \\
    \left\Vert \Phi \right\Vert_{\diamond} &\coloneqq  \sup_{|R|, \,  \rho_{RA}} \left\Vert (\id_{R} \otimes \Phi)(\rho)\right\Vert_{1} = \max_{\dyad{\psi}_{A'A} \in \Density(A' \otimes A)} \left\Vert (\id_{A'} \otimes \Phi)(\dyad{\psi}) \right\Vert_{1} \ , \label{eq:diamond-norm-as-max-over-state}
\end{align}
where the maximizations always come from the sets being optimized over being compact and $A' \cong A$. We also consider the norm induced by the Schatten $1$-norm on the Choi operator, $\left\Vert \Gamma^{\Phi} \right\Vert_{1}$, where the Choi operator $\Gamma^{\Phi}$ is defined as 
\begin{align}\label{eq:choi-operator}
    \Gamma^{\Phi} \coloneqq \ d_{A}(\id \otimes \Phi)(\dyad{\Phi_{d}}) \coloneq (\id \otimes \Phi)(\dyad{\Gamma}) \ ,
\end{align}
i.e., the map acting on one share of the unnormalized maximally entangled state of dimension $d = \vert A \vert$.

\begin{proposition}[Equivalence of norms on Hermitian-preserving maps]\label{prop:equiv-of-chan-norms}
    Let $\cN_{A \to B}$ and $\cM_{A \to B}$ be completely positive maps. Then,
    \begin{align}
        \left\Vert \cN - \cM \right\Vert_{1} & \leq \left\Vert \cN - \cM \right\Vert_{\diamond} \leq 2d_{A} \left\Vert \cN - \cM \right\Vert_{1} \label{eq:diamond-and-one-to-one-norm-equiv}, \\
        \left\Vert \cN - \cM \right\Vert_{\diamond} & \leq \left\Vert \Gamma^{\cN} - \Gamma^{\cM} \right\Vert_{1} \leq d_{A} \left\Vert \cN - \cM \right\Vert_{\diamond}.  \label{eq:diamond-and-choi-dist-norm-equiv}
    \end{align}
\end{proposition}

\begin{proof}
    The proof of~\eqref{eq:diamond-and-one-to-one-norm-equiv} is identical to \cite[Lemma 6]{bluhm2026equivalencenonlocalcomputationtasks}, which need not make an assumption that the input and output linear operators act on the same dimension of Hilbert space (so long as they are finite). The proof of~\eqref{eq:diamond-and-choi-dist-norm-equiv} may be done identically to \cite[Lemma 7]{Wallman_2014} by noting that it need not require the input and output spaces be the same.
\end{proof}

\paragraph{The equivalence}
\begin{proposition}
\label{prop:asymptotic-equiv-of-types-of-mixing}
    For a time-homogeneous Markov chain, being mixing and being replacing are equivalent. Moreover, this means that any notion of asymptotic behaviour of a homogeneous Markov chain that is implied by replacing and implies mixing, such as being decoupling\footnote{Decoupling means that $\lim_{n \to \infty} \Vert (\id_{R} \otimes \cN^{n})(\rho_{RA}-\rho_{R} \otimes \omega_{A})\Vert_{1} = 0$} \cite{george2025quantumdoeblincoefficientsinterpretations}, is also equivalent.
\end{proposition}
\begin{proof}
    A homogeneous Markov chain $\cN$ that is replacing under trace distance implies it is mixing by choosing $\rho_{RA} = \rho_{R} \otimes \rho_{A}$ and $\sigma_{RA} = \rho_{R} \otimes \omega_{A}$. It thus suffices to prove mixing implies replacing. By the definition of mixing, 
    \begin{align}\label{eq:mixing-approx}
        \forall \delta > 0, \exists n_{\delta} \in \mbb{N} : \forall n \geq n_{\delta} \, , \,  \left\Vert \cN^{n}(\rho) - \omega \right\Vert_{1} \leq \delta \quad \forall \rho \in \Density(A) \ . 
    \end{align}
    We will use this to show that $\cN^{n}$ converges to the replacer channel $\cR^{\omega}$ in diamond norm. Note that 
    \begin{align}
        \left\Vert \cN^{n}(\rho) - \omega \right\Vert_{1} = \left\Vert (\cN^{n} - \cR^{\omega})(\rho) \right\Vert_{1} \ . 
    \end{align}
    Thus, \cref{eq:mixing-approx} implies that, for all $\delta > 0$, there exists $n_{\delta} \in \mbb{N}$ such that for all $n \geq n_{\delta}$, $\left\Vert \cN^{n} - \cR^{\omega} \right\Vert_{1} \leq \delta$ for all $n \geq n_{\delta}$. By the equivalence of channel norms (\cref{prop:equiv-of-chan-norms}), the same claim holds for diamond distance; i.e., for all $\delta > 0$, there exists $n_{\delta} \in \mbb{N}$ such that for all $n \geq n_{\delta}$, $ \left\Vert \cN^{n} - \cR^{\omega} \right\Vert_{\diamond} \leq \delta$, so 
    \begin{align}\label{eq:convergence-in-CB1-norm}
        \lim_{n \to \infty} \left\Vert \cN^{n} - \cR^{\omega} \right\Vert_{\diamond} = 0 \ . 
    \end{align}
    Then for all $\rho_{RA}$ and $\sigma_{RA}$,
    \begin{align}
        &\left\Vert (\id \otimes \cN^{n})(\rho_{RA}) - (\id \otimes \cN^{n})(\sigma_{RA}) \right\Vert_{1} - \left\Vert \rho_{R} - \sigma_{R} \right\Vert_{1} \notag \\ 
        &= \left\Vert (\id \otimes(\cN^{n} + \cR^{\omega} - \cR^{\omega}))(\rho_{RA} - \sigma_{RA}) \right\Vert_{1} - \left\Vert \rho_{R} - \sigma_{R} \right\Vert_{1} \\
        &\leq \left\Vert \id \otimes (\cN^{n} - \cR^{\omega})(\rho_{RA} - \sigma_{RA}) \right\Vert_{1} + \left\Vert (\id \otimes \cR^{\omega})(\rho_{RA} - \sigma_{RA}) \right\Vert_{1}  - \left\Vert \rho_{R} - \sigma_{R} \right\Vert_{1} \\
        &= \left\Vert \id \otimes (\cN^{n} - \cR^{\omega})(\rho_{RA} - \sigma_{RA}) \right\Vert_{1} + \left\Vert \rho_{R} \otimes \omega - \sigma_{R} \otimes \omega \right\Vert_{1} - \left\Vert \rho_{R} - \sigma_{R} \right\Vert_{1} \\
        &= \left\Vert \id \otimes (\cN^{n} - \cR^{\omega})(\rho_{RA} - \sigma_{RA}) \right\Vert_{1} \\
        &\leq \left\Vert \cN^{n} - \cR^{\omega} \right\Vert_{\diamond} \left\Vert \rho_{RA} - \sigma_{RA} \right\Vert_{1} \label{eq:replacement-time-in-1-norm}
    \end{align}
    where the first equality follows from grouping terms and adding and subtracting a channel, the first inequality from the triangle inequality, the second equality from the action of the replacer channel, the third equality from the Schatten $1$-norm being multiplicative over tensor products, $\left\Vert \omega \right\Vert_{1} = 1$, and canceling terms, and the final inequality from the definition of the completely bounded 1-norm. Combining~\eqref{eq:convergence-in-CB1-norm} and~\eqref{eq:replacement-time-in-1-norm},
    \begin{align}
        \lim_{n \to \infty} [\left\Vert (\id \otimes \cN^{n})(\rho_{RA}) - (\id \otimes \cN^{n})(\sigma_{RA}) \right\Vert_{1} - \left\Vert \rho_{R} - \sigma_{R} \right\Vert_{1}] = 0 \ . 
    \end{align}
    Thus, by definition, the homogeneous Markov chain $\cN$ is replacing.

    Finally, if some notion implies $\cN$ is mixing and is implied by replacing, then one concludes that it is equivalent as mixing implies replacing.
\end{proof}

\subsubsection{Finite-time bounds}

As shown, a homogeneous Markov chain $\cN_{A \to A}$ being replacing captures that the channel converges to the replacer channel. Thus, it is natural to measure the time it takes for the Markov chain to converge in terms of how different the output of $\cN^{n}$ is in relation to the fixed point. We call this the $\mbb{D}$-replacement time and provide the pre-existing analogous definitions for clarity.
\begin{definition}\label{def:time-hom-times}
    Given a divergence $\DD$, a quantum channel $\cN_{A \to A}$ with unique fixed point $\omega$, and $\delta \geq 0$, we define
    \begin{enumerate}[leftmargin=*]
        \item \textbf{$\DD$-Mixing Time:}
            \begin{align}
                t_{\operatorname{mix}}^{\mbb{D}}(\cN,\delta) \coloneq \inf\!\left \{n \in \mbb{N}: \mbb{D}\!\left(\cN^{n}(\rho), \omega \right) \leq \delta \quad \forall \rho \in \Density(A) \right \} \ , 
            \end{align}
        \item \textbf{$\DD$-Decoupling Time:}
            \begin{align} \label{eq:decoupling_time}
                t_{\operatorname{dec}}^{\mbb{D}}(\cN,\delta) \coloneqq \inf \left\{ n \in \mbb{N} : \begin{matrix}   \mbb{D}((\id_{R} \otimes \cN^{n})(\rho_{RA}) \Vert \rho_{R} \otimes \omega_{A}) \leq \delta \, , \\
                \forall R \, , \, \forall \rho_{RA} \in  \Density(R \otimes A) \end{matrix} \right\} \ , 
            \end{align}
        \item \textbf{$\DD$-Replacement Time:}
            \begin{align} \label{eq:replacement_time_good}
                t_{\operatorname{rep}}^{\mbb{D}}(\cN,\delta) \coloneqq \inf \left\{ n \in \mbb{N} : \begin{matrix}   \mbb{D}((\id_{R} \otimes \cN^{n})(\rho_{RA}) \Vert \sigma_{R} \otimes \omega_{A}) - \DD(\rho_{R} \Vert \sigma_{R}) \leq \delta \, , \\
                \forall R \, , \, \forall \rho_{RA} \in \Density(R \otimes A),\sigma_{R} \in  \Density(R) \end{matrix} \right\} \ .
            \end{align}
    \end{enumerate}
\end{definition}
\noindent We explicitly specify that $\cN_{A \to A}$ has a unique fixed point $\omega$, as otherwise the above times are trivially infinite by choosing $\rho_{A}$ to be a different fixed point. 

We then obtain a bound on the replacement time in terms of the conditional SDPI constant, which is analogous to previous bounds for other operational times (c.f.~\cite{george2025quantumdoeblincoefficientsinterpretations,george2025unifiedapproachquantumcontraction}).
\begin{proposition}
    Let $\DD$ be a divergence. Consider a homogeneous Markov chain $\cN_{A \to A}$ with fixed point $\omega_{A}$. Then 
    \begin{align}
        t^{\DD}_{\operatorname{rep}}(\cN,\delta) \leq \log(C^{\Delta}_{\DD}/\delta)/\log(1/\eta_{\DD}^{\Delta,\operatorname{stab}}(\cN,\omega)) ,
    \end{align}
    where 
    \begin{equation}
        C_{\mathbb{D}}^{\Delta}\coloneqq\sup_{\rho_{RA},\sigma_{R}}\mathbb{D}(\rho_{RA}\|\sigma_{R} \otimes \omega_{A})-\mathbb{D}(\rho_{R}\|\sigma_{R}).\label{eq:distin-increase-refs}
    \end{equation}
\end{proposition}
\begin{proof}
    We have the sequence of inequalities
    \begin{align}
        &\mbb{D}((\id_{R} \otimes \cN^{n})(\rho_{RA}) \Vert \sigma_{R} \otimes \omega_{A}) - \DD(\rho_{R} \Vert \sigma_{R}) \notag \\
        &=\mbb{D}((\id_{R} \otimes \cN^{n})(\rho_{RA}) \Vert \sigma_{R} \otimes \cN^{n}(\omega_{A})) - \DD(\rho_{R} \Vert \sigma_{R}) \\
        &\leq \eta_{\mbb{D}}^{\operatorname{\Delta,stab}}(\cN,\omega_{A})\left[\mbb{D}((\id_{R} \otimes \cN^{n-1})(\rho_{RA}) \Vert \sigma_{R} \otimes \omega_{A}) - \DD(\rho_{R} \Vert \sigma_{R})  \right] \\
        &\leq \eta_{\mbb{D}}^{\operatorname{\Delta,stab}}(\cN,\omega_{A})^{n} \left[\mbb{D}(\rho_{RA} \Vert \sigma_{R} \otimes \omega_{A}) - \DD(\rho_{R} \Vert \sigma_{R})  \right] \\
        &\leq \eta_{\mbb{D}}^{\operatorname{\Delta,stab}}(\cN,\omega_{A})^{n} C^{\Delta}_{\DD} \ ,
    \end{align}
    where the equality follows because $\omega_{A}$ is a fixed point of $\cN$, the first inequality from \eqref{eq:stabilized-conditional-SDPI} with the replacements $S \to R$, $R \to \mbb{C}$, $A \to A$, the second inequality from doing the previous inequality $n-1$ times, and the final inequality from the definition of $C^{\Delta}_{\DD}$. By bounding the final term from above by $\delta$, taking the logarithm, and solving for $n$ completes the proof.
\end{proof}
Here we recall that by~\cref{Prop:tensorization-stabilized-conditional-SDPI}, if we choose $\DD$ as the relative entropy, $\eta_{D}^{\operatorname{\Delta,stab}}$ tensorizes, implying stability under tensor powers of the channel $\cN$. 

Finally, we observe that, for the relative entropy, decoupling and replacing time become equivalent.
\begin{proposition}
    For every channel $\cN$ and all $\delta \in (0,\infty)$, 
    \begin{equation}
    t^{D}_{\operatorname{rep}}(\cN,\delta) = t^{D}_{\operatorname{dec}}(\cN,\delta).    
    \end{equation}
\end{proposition}

\begin{proof}
    By the relative entropy chain rule \eqref{eq:rel-ent-chain},
    \begin{align}
        &D((\id_{R} \otimes \cN^{n})(\rho_{RA}) \Vert \sigma_{R} \otimes \omega_{A}) - D(\rho_{R} \Vert \sigma_{R}) \notag \\ &= D((\id_{R} \otimes \cN^{n})(\rho_{RA}) \Vert \rho_{R} \otimes \omega_{A}) + D(\rho_{R} \Vert \sigma_{R}) - D(\rho_{R} \Vert \sigma_{R}) \\
        &= D((\id_{R} \otimes \cN^{n})(\rho_{RA}) \Vert \rho_{R} \otimes \omega_{A}) \ .
    \end{align}
    Thus
    \begin{equation}
    D((\id_{R} \otimes \cN^{n})(\rho_{RA}) \Vert \sigma_{R} \otimes \omega_{A}) - D(\rho_{R} \Vert \sigma_{R}) \leq \delta    
    \end{equation}
    if and only if
    \begin{equation}
    D((\id_{R} \otimes \cN^{n})(\rho_{RA}) \Vert \rho_{R} \otimes \omega_{A}) \leq \delta,    
    \end{equation}
    which are the conditions in the definitions. 
\end{proof}

An interesting corollary of this is that relative entropy replacement time is controlled by the product SDPI constant rather than the trivially stabilized conditional SDPI constant.
\begin{corollary}
    For a channel $\cN_{A \to A}$ with fixed point $\omega > 0$,
    \begin{align}
        t^{D}_{\operatorname{rep}}(\cN,\delta) = t^{D}_{\operatorname{dec}}(\cN,\delta) \leq \left \lceil \frac{\log(C^{p}_{D}/\delta)}{\log(1/\eta^{p}_{D}(\cN,\omega))}\right \rceil \ .
    \end{align}
    where $C^{p}_{D} \coloneq \sup_{\rho_{RA}} D(\rho_{RA} \Vert \rho_{R} \otimes \omega_{A})$.
\end{corollary}
\begin{proof}
As $\cN(\omega) = \omega$ and using the definition of the product SDPI constant (see \eqref{SDPI-product}) $n$ times,
    \begin{align}
        D((\id_{R} \otimes \cN^{n})(\rho_{RA}) \Vert \rho_{R} \otimes \omega_{A}) &= D((\id_{R} \otimes \cN^{n})(\rho_{RA}) \Vert \rho_{R} \otimes \cN^{n}(\omega_{A})) \\
        &\leq \left[\eta^{p}_{D}(\cN,\omega)\right]^{n}D(\rho_{RA} \Vert \rho_{R} \otimes \omega_{A}) \\
        &\leq \left[\eta^{p}_{D}(\cN,\omega)\right]^{n}C^{p}_{D} \ .
    \end{align}
    Bounding the right-hand side by $\delta$ and solving for $n$ completes the proof.
\end{proof}

\begin{remark}[Time-Inhomogeneous Markov Chains]
    In \cite{george2025quantumdoeblincoefficientsinterpretations}, finite time bounds on time \textit{inhomogeneous} Markov chains were developed using contraction coefficients. The same ideas may be applied for replacing times. However, this is repetitive and does not lead to a new insight beyond that the Markov chain must converge to a sequence of replacer channels rather than a fixed replacer channel. Thus, we omit a more in-depth discussion. The reader can work out the ideas straightforwardly by combining the above proof methods with the notion of replacing and the definition of a time-inhomogeneous Markov chain from \cite{george2025quantumdoeblincoefficientsinterpretations}.
\end{remark}

\subsubsection{Doeblin coefficient is eventually non-trivial for mixing channels}

In \cite{george2025quantumdoeblincoefficientsinterpretations}, the authors considered the quantum Doeblin coefficient:
\begin{align}
\label{eq:quantum-Doeblin}
    \alpha(\cN) \coloneqq  \max\{ \Tr[H] \colon I \otimes H \leq \Gamma^{\cN}, \ H \in \operatorname{Herm} \} \ . 
\end{align}
This quantity is motivated as being an efficiently computable quantity used to upper bound the trace distance contraction coefficient, $\eta_{\text{Tr}}(\cN) \leq 1 - \alpha(\cN)$. This is valuable as the trace distance contraction coefficient is NP-hard to calculate in general \cite{delsol2025computationalaspectstracenorm}, and the trace distance contraction coefficient can be used to bound the speed of mixing, i.e., to establish finite-time guarantees for mixing. 

One of the undesired aspects of the quantum Doeblin coefficient is that for a channel $\cN_{A \to A}$ such that $\vert A \vert > 2$, it can be the case $\alpha(\cN) = 0$ even when $\eta_{\text{Tr}}(\cN) < 1$; i.e., it is not faithful. Because the Choi operator of a replacer channel $\cR^{\omega}$ is $I \otimes \omega$, \cref{eq:quantum-Doeblin} in effect measures `how much' of a replacer \textit{map} that prepares a Hermitian \textit{operator} is `contained' in the Choi operator of the channel. In the previous section, we showed that all variations on mixing are asymptotically equivalent in finite dimensions because the (time-homogeneous) Markov chain converges to a replacer channel. Thus, one might hope that, for a determinable number $n$, one can guarantee $\alpha(\cN^{n}) > 0$. In this section we show that this is indeed the case as long as the channel when applied for $d^{2}$ iterations is sufficiently contractive.

\begin{proposition}
\label{prop:non-trivial-Doeblin-depth}
    Let $\cN_{A \to A}$ be a quantum channel, $d \coloneq \vert A \vert$, and $\ve > 0$. Then either
    \begin{enumerate}
        \item $\eta_{\Tr}(\cN^{d^{2}}) > 1 - \ve$ or
        \item $\alpha(\cN^{m\cdot d^{2}}) > 0$ for all $m \geq \frac{\log\left(4d^2(d+1)\right)}{\log\left(\frac{1}{1-\ve}\right)} $.
    \end{enumerate}
\end{proposition}

We remark that the first condition is natural in the sense that it has been shown that if the channel were to be mixing, then its trace distance contraction coefficient would be strictly less than one after $d^{2}$ iterations. 
\begin{proposition}
    (\cite[Theorems A.11 and B.10]{singh2024zeroerrorcommunicationdiscretetimemarkovian}) A channel $\cN_{A \to A}$ is mixing if and only if $\eta_{\Tr}(\cN^{d^{2}}) < 1$ where $d = \vert A \vert$.
\end{proposition}
\noindent In particular, this allows us to conclude the following.
\begin{corollary}\label{cor:mixing-equiv-to-q-Doeblin-being-non-triv-at-some-depth}
     There exists $n \in \mbb{N}$ such that $\alpha(\cN^{n}) > 0$ if and only if $\cN$ is mixing.
\end{corollary}

We now turn to proving \cref{prop:non-trivial-Doeblin-depth}. First, we identify bounds on the Doeblin coefficient as a function of its distance from a replacer channel.
\begin{lemma}
\label{lem:Choi-op-norm-to-Doeblin}
    Let $\cN_{A \to A}$ be a channel, and let $\omega$ be a state. If $\left\Vert \Gamma^{\cN} - I \otimes \omega \right\Vert_{\infty} \leq \delta$, then $\alpha(\cN) \geq 1 - \delta \vert A \vert $.
\end{lemma}

\begin{proof}
    As $\delta \geq \left\Vert \Gamma^{\cN} - I \otimes \omega\right \Vert_{\infty}$, $\lambda_{\min}(\Gamma^{\cN} - I \otimes \omega) \geq -\delta$. It follows that
    \begin{align}
        \Gamma^{\cN} - I \otimes \omega \geq -\delta I_{AB} \iff \Gamma^{\cN} &\geq I \otimes \omega - \delta I_{AB} \\
        &= I_{A} \otimes (\omega - \delta I_{A}) \eqqcolon I_{A} \otimes \widehat{H} \ .
    \end{align}
    As $\widehat{H}$ is Hermitian and $\Gamma^{\cN} \geq I_{A} \otimes H$, by \eqref{eq:quantum-Doeblin}, $\widehat{H}$ is a feasible point for $\alpha(\cN^{n})$. Moreover, $\Tr[H] = 1 - \delta \vert A \vert$. As \eqref{eq:quantum-Doeblin} is a maximization, this completes the proof.
\end{proof}

Now we prove the proposition.
\begin{proof}[Proof of \Cref{prop:non-trivial-Doeblin-depth}]
    First observe that for a fixed point $\omega$ of $\cN$ and all $n \in \mbb{N}$,
    \begin{align}
        \left\Vert \Gamma^{\cN^{n}} - I \otimes \omega \right\Vert_{\infty} &\leq \left\Vert \Gamma^{\cN^{n}} - I \otimes \omega \right\Vert_{1} \\
        &\leq d \left\Vert \cN^{n} - \cR^{\sigma} \right\Vert_{\diamond} \\
        &\leq 2d^2 \left\Vert \cN^{n} - \cR^{\sigma} \right\Vert_{1} \\
        &= 2d^2 \cdot \sup_{\rho} \left\Vert \cN^{n}(\rho) - \omega \right\Vert_{1} \\
        &\leq 4d^2 \cdot \eta_{\Tr}(\cN^{n}) \ , \label{eq:doeblin-non-triv-step-1}
    \end{align}
    where we used the ordering of Schatten norms, \Cref{prop:equiv-of-chan-norms}, the definition of the $1 \to 1$ norm, the fact that $\omega$ is a fixed point of $\cN^{n}$ and then the fact that $\eta_{\Tr}(\cN) \coloneq \sup_{\rho \neq \sigma} \frac{1}{2}\left\Vert \cN(\rho - \sigma) \right\Vert_{1}$.

    Now assume that $\eta_{\Tr}(\cN^{n}) \leq 1 - \ve$. By submultiplicativity of the contraction coefficient \cite{Hirche2022contraction}, $\eta_{\Tr}(\cN^{n \cdot m}) \leq \eta_{\Tr}(\cN^{n})^{m} \leq (1-\ve)^{m}$. Combining \eqref{eq:doeblin-non-triv-step-1} with this,
    \begin{align}
        \left\Vert \Gamma^{\cN^{n}} - I \otimes \omega \right\Vert_{\infty} \leq 4d^2(1-\ve)^{m} \ .
    \end{align}
    So long as $4d^2(1-\ve)^{m} \leq \frac{1}{d+1}$, we obtain $\alpha(\cN^{n\cdot m}) > 0$ by \Cref{lem:Choi-op-norm-to-Doeblin}. Re-arranging to isolate $m$ completes the proof.
\end{proof}

\subsection{Correlation decay across a \texorpdfstring{$1D$}{} open chain}

The authors of~\cite{kato2019quantum}  identified conditions on a quantum state $\rho$ representing a $1D$-chain to be approximated by the Gibbs state of a 2-local Hamiltonian:

\begin{proposition}[Theorem 1 of \cite{kato2019quantum}]
\label{prop:2-parent-hamiltonian}
    Let $\ve \geq 0$, and let $\rho_{A^{m}_{1}}$ be a quantum state on systems $A_1 \cdots A_m$ such that 
    \begin{align}
        I(A_{1}^{i-1}:A_{i+1}^{m} \vert A_{i}) \leq \ve \quad \forall i \in \{1,\ldots,m\} \ . 
    \end{align}
    Then there exist Hermitian operators $h_{i} \in \operatorname{Herm}(A_{i} \otimes A_{i+1})$ such that $D(\rho \Vert \exp(-H)/Z) \leq \ve m$ where $H = \sum_{i=1}^{m-1} h_{i}$ and $Z$ is the corresponding partition function.
\end{proposition}

Motivated by this, the authors of~\cite{chen2020matrix} showed that, given a broadcast channel $\cN_{C \to BC}$, then for $A \cong C$ and $\rho_{AB^{n}C}$  generated by applying $\cN$ first to the $\ol{A}$ system of a maximally entangled state $\dyad{\Phi_{d}}_{A\ol{A}}$ where $d = \vert A \vert$ and then to the output $C$ system the remaining times, it is the case that $I(A:C \vert B^{\ell}_{1}) = O(\eta^{\ell})$ for some $\eta < 1$ under certain conjectures including a conditional mutual information contraction coefficient being non-trivial (see \cite[Proposition III.4]{chen2020matrix}). That is, under certain conjectures, the correlation from one end of the chain to the other decays exponentially in the application of the channel.

Here we prove a similar result that does not require the input state to be a maximally entangled state. We also show that the exponential decay holds, so long as the sequence of channels have Choi operators with a minimum eigenvalue bounded from below by a number strictly greater than zero.

First we recall the following chain rule.
\begin{proposition}
    For $\rho_{ABCD}$, 
    \begin{align}\label{eq:CMI-to-CMI-chain-rule}
        I(A:BC \vert D) = I(A:C \vert D) + I(A:B \vert CD)
    \end{align}
\end{proposition}
\begin{proof}
    This follows from just using the chain rule $I(A:B \vert C) = I(A:BC) - I(A:C)$ three times:
    \begin{align}
        I(A:B \vert CD) + I(A:C \vert D) &= I(A:BCD) - I(A:CD) + I(A:CD) - I(A:D) \\
        &= I(A:BCD) - I(A:D) \\
        &= I(A:BC \vert D) \ , 
    \end{align}
    thus concluding the proof.
\end{proof}

Now we prove the generic decay (see also the proof of \cite[Proposition III.4]{chen2020matrix}).
\begin{proposition}
\label{prop:CMI-decay}
Let $\left(\cN^{i}_{\ol{A}_{i} \to A_{i}\ol{A}_{i+1}}\right)_{i \in \{1,...,m-1\}}$ be a tuple of quantum channels, and let $\cN^{m}_{\ol{A}_{m} \to A_{m}}$ be a quantum channel. Define
\begin{align}
    \rho_{A^{m}_{0}} \coloneq \bigcirc_{i \in \{1,\ldots,m\}} (\id_{A^{i-1}_{1}} \otimes \cN^{i})(\rho_{A_{0}\ol{A}_{1}}) \ . 
\end{align}
Then for every $k \in \{1,\ldots ,m\}$
\begin{align}
    I(A_{0}:A_{k} \vert A_{1}^{k-1}) &\leq \left[\prod_{i \in \{1,\ldots,k\}} \eta_{\CMI}(\cN^{i}) \right] \cdot I(A_{0}:\ol{A}_{1}) \\
    & \leq \left[\prod_{i \in \{1,\ldots,k\}} \eta_{\CMI}(\cN^{i}) \right] 2 \log \min\{\vert A_{0} \vert, \vert \ol{A}_{1} \vert \} \ . 
\end{align}
\end{proposition}
\begin{proof}
    First, $I(A_{0}:A_{m} \vert A_{1}^{m-1}) \leq \eta_{\CMI}(\cN^{m})I(A_{0}:\ol{A}_{m} \vert A_{1}^{m-1})$. Now to apply data processing again, we need the $A^{m-1}$ register together with $\ol{A}_{m}$. By using~\eqref{eq:CMI-to-CMI-chain-rule} with $A \to A_{0}$, $B \to \ol{A}_{m}$, $C \to A_{m-1}$, and $D \to A^{m-2}_{1}$,
    \begin{align}
        I(A_{0}:\ol{A}_{m} \vert A^{m-1}_{1}) &= I(A_{0}:\ol{A}_{m}A_{m-1}\vert A^{m-2}_{1}) - I(A_{0}:A_{m-1}\vert A^{m-2}_{1}) \\
        &\leq I(A_{0}:\ol{A}_{m}A_{m-1}\vert A^{m-2}_{1}) \\
        &\leq \eta_{\CMI}(\cN^{m-1})I(A_{0}:\ol{A}_{m-1} \vert A^{m-2}_{1}) \ .
    \end{align}
    By applying the same argument iteratively, we obtain the mutual information bound. The subsequent bound follows from \cite[Exercise 11.7.9]{Wilde-Book}. Lastly, we can start this argument for all $k \leq m$, which completes the proof.
\end{proof}

To apply this decay, one needs some guarantee that $\eta_{\CMI}(\cN^{i})$ is not equal to one. By \Cref{Lem:a-Doeblin}, it would suffice for $\alpha_{+}(\cN^{i})$ to be non-zero. A sufficient condition for this is the following.
\begin{proposition}
    For every quantum channel $\cN_{A \to B}$,
    \begin{equation}
    \alpha_{+}(\cN) \geq \vert B \vert \lambda_{\min}(\Gamma^{\cN}).    
    \end{equation}
    In particular, $\alpha_{+}(\cN) > 0$ for every channel with a full rank Choi operator.
\end{proposition}

\begin{proof}
    The proof is identical to \cite[Proposition 38]{george2025quantumdoeblincoefficientsinterpretations}.
\end{proof}

Combining this with \Cref{prop:CMI-decay} allows us to conclude the following.
\begin{proposition}
    Consider sequences of channels $(\cN^{i}_{\overline{A}_{i} \to A_{i}\overline{A}_{i+1}})$ where $\lambda_{\min}(\cN^{i}) \geq \epsilon > 0$ for all $i$. Define $(\wt{\cN}^{i} \coloneq \Tr_{\ol{A}_{i+1}} \circ \cN^{i})$ and, for every $m > 1$, 
    \begin{align}
    \rho_{A^{m}_{0}} \coloneq \wt{\cN}^{m} \circ \bigcirc_{i \in \{1,\ldots,m-1\}} (\id_{A^{i-1}_{1}} \otimes \cN^{i})(\rho_{A_{0}\ol{A}_{1}}) \ . 
\end{align}
    If all of the channels are full rank, then $I(A_{0}:A_{m} \vert A_{1}^{m-1})$ decays exponentially fast in the depth.
\end{proposition}

We remark that, in comparison to \cite{chen2020matrix}, our result does not require the initial state to be one share of an entangled state, nor do we require a unitality property of the channel. It does however require the channels to have full-rank Choi states. We also highlight that since $\alpha_{+}$ is a semidefinite program, one can in principle compute a bound on the exponential decay.

\section{Applications to quantum machine learning}

\label{Sec:Applications3}

In this section, we discuss how some of the results developed in this work provide insights into various performance indicators of quantum machine learning models, including stability, generalization, and training due to the application of noisy quantum channels and algorithms. 

\subsection{Stability of quantum learning algorithms}

A learning algorithm is known to
be stable if its output does not depend too much on the input data~\cite{bousquet2002stability,raginsky2016information}. It is also known that stability and generalization of a classical learning algorithm to new inputs are closely related, where stability implies generalization using information-theoretic tools~\cite{raginsky2016information}.

It is an interesting research question to explore whether noisy quantum learners also provide stability and generalization in quantum learning settings. In~\cite{caro2023information}, it was shown that the generalization error of classical-quantum learners is bounded from above by a function of the mutual information between the input and output spaces, plus a Holevo information term. 
In many noisy quantum learning algorithms, stability is primarily considered in differentially private (a form of noisy) quantum models~\cite{nuradha2024contraction,dasgupta2026privacyimpliesstabilityinformationtheoretic,nuradha2023quantum}.

\textit{Setup}: We focus on the setup considered in~\cite[Section~III]{dasgupta2026privacyimpliesstabilityinformationtheoretic}. Fix a finite classical dataset such that 
$s\coloneqq (z_1,\ldots,z_n)\in\cS\coloneqq\cZ^n$, $z_i\coloneqq (x_i,y_i)$, and $s$ is drawn from a  data distribution
$p_S$. Each data sample $z_i\in\cZ$ is encoded into a quantum state as $\rho_{z_i} \in \cD(\hat{T_1} \otimes \hat{T_2})$, where $\hat{T_1}$ and $ \hat{T_2}$ are quantum systems. The aggregate state for $n$ samples of data in $s$ is
\begin{equation}
\rho_s \coloneq \bigotimes_{i=1}^n \rho_{z_i} \in \cD( T_1 \otimes T_2),    
\end{equation}
where $T_1 \equiv \hat{T}_1^{\otimes n}$ and $T_2 \equiv \hat{T}_2^{\otimes n}$.

The quantum algorithm is carried out by a party known as the Data Processor.
The Data Processor executes a learning algorithm modeled as a
collection of quantum channels
$\cN\coloneqq\{\cN^{(s)}\colon T_1 \to B\}_{s\in\cS}$,
where $B\equiv WB'$ consists of a classical hypothesis register $W$ and a
quantum system $B'$, and the register $T_2$ cannot be accessed by the algorithm. Then,
the action on the full encoded state is $\id_{T_2}\otimes
\cN^{(s)}$. Denoting $\sigma_s^{\cN}\coloneqq(\id_{T_2}\otimes
\cN^{(s)})(\rho_s)$, the joint classical-quantum state after the Data Processor acts is
\begin{equation} \label{eq:sigma_N_output}
    \sigma^{\cN}\coloneqq \sum_{s\in\cS} p_S(s) |s\rangle\!\langle s|\otimes \sigma_s^{\cN}.
\end{equation}

\medskip
In~\cite[Section VI.B]{nuradha2024contraction}, stability for a quantum model with classical data is quantified by Holevo information. Here, we consider the stability notion defined for a general quantum learning system in~\cite[Definition~1]{dasgupta2026privacyimpliesstabilityinformationtheoretic} specialized for a fixed data distribution $p_S$ (note that ~\cite[Definition~1]{dasgupta2026privacyimpliesstabilityinformationtheoretic} has a maximization over $p_S$).

\begin{definition}[Stability] \label{def:stability_def}
    For $\gamma \geq 0$, a quantum learning algorithm $\cN$ is
    $\gamma$-stable with respect to $p_S$ and the collection $\left(\rho_s\right)_{s\in \mathcal{S}}$ of states if
    \begin{equation}
         I(S T_2 ;WB')_{\sigma^\cN} \le \gamma,
    \end{equation}
    where $\sigma^\cN$ is given by~\eqref{eq:sigma_N_output}.
\end{definition}

In prior work on the quantum setting, stability is primarily considered in differentially private quantum models~\cite{nuradha2024contraction,dasgupta2026privacyimpliesstabilityinformationtheoretic,nuradha2023quantum}.
Next, we discuss the stability of the quantum learning model where $\cN^{(s)}=\cM$ for all $s \in \cS$; i.e., $\mathcal{N}_{(s)}$ is equal to the same quantum channel $\mathcal{M}$. This is relevant to the scenario where the learning model is fixed or pre-trained. Also, note that the setting  $\cN^{(s)} \neq \cN^{(s')}$ for $s \neq s'$ is the general setting; we leave the analysis of this more general scenario for future work. 
\begin{proposition}
    Let $\cM$ be a quantum channel, and let $\cN$ be a quantum learning algorithm such that $\cN^{(s)}= \cM$ for all $s \in \cS$.
    The algorithm $\cN$ is $\gamma'$-stable if
    \begin{equation}
\eta^p_D(\cM) \ I(S T_2;T_1) \leq \gamma'.
    \end{equation}
Furthermore, it is $ \gamma'$-stable if
\begin{equation}
   (1-\alpha_+(\cM)) \ I(S T_2; T_1)  \leq \gamma'. 
\end{equation}

\label{prop:Stability_QML_noise}
\end{proposition}
\begin{proof}
Consider the following:
\begin{align}
    I(ST_2:WB') &= D(\sigma^\cN_{ST_2WB'}\| \sigma^\cN_{ST_2}\otimes \sigma^\cN_{WB'}) \\
    &= \sum_s p_S(s) D(\sigma^{\cN,s}_{T_2WB'}\| \sigma^{\cN,s}_{T_2}\otimes \sigma^\cN_{WB'}) \\
    &\leq \sum_s p_S(s) \eta^p_D(\cM) D(\sigma^{\cN,s}_{T_2T1}\| \sigma^{\cN,s}_{T_2}\otimes \sigma^\cN_{T1}) \\
    &= \eta^p_D(\cM)  \sum_s p_S(s)  D(\sigma^{\cN,s}_{T_2T1}\| \sigma^{\cN,s}_{T_2}\otimes \sigma^\cN_{T1}) \\
    &= \eta^p_D(\cM) I(ST_2:T_1),
\end{align}
where the first inequality follows from the application of~\eqref{SDPI-product} with the system choice being $R\equiv T_2$ and $A\equiv T_1$ along with $\cN_{A \to B} \coloneq \cM_{T_1 \to WB'}$.
With that, we arrive at the first statement. 

Next, by the use of 
$\eta^p_D(\cM)\leq \etaD_D(\cM) \leq  1-\alpha_+(\cM)$, which follow from \cref{Cor:Summary1} and \eqref{eq:CMI-contract-doeblin}, we conclude the second statement as well.
\end{proof}

\begin{remark}
\cref{prop:Stability_QML_noise} provides an approach to understanding how the encoding procedure and the design of the learning model impact the stability of the trained learning algorithm. The contraction coefficients cover the impact of the noisy learning model, while the mutual information terms capture the dependencies and correlations in the dataset and the encoding protocol.

If $\cN=\cP_d \circ \cU_d \circ \cdots \cU_1 \circ \cP_1 $, where each $\cP_i$ is a  quantum channel and each $\cU_i$ is a unitary channel,
then $\cN$ is $\prod_{i=1}^d (1-\alpha_+(\cP_i)) I(ST_2;T_1)$-stable, in such a way that the terms in front of the mutual information term $I(ST_2;T_1)$ decay more with increasing $d$.
\end{remark}

\begin{remark}[Connection to generalization error]
    We see that the stability defined in~\cref{def:stability_def} has a direct connection to the generalization error from~\cite[Theorem~1]{dasgupta2026privacyimpliesstabilityinformationtheoretic}. Generalization error quantifies how well a model provides inference for fresh input data. Thus, generalization error in the setup we consider in~\cref{prop:Stability_QML_noise} scales as $\sqrt{\gamma'}$ with 
    \begin{equation}
    \eta^p_D(\cM) I(ST_2:T_1) \leq \gamma',    
    \end{equation}
    as a consequence of~\cite[Theorem~1]{dasgupta2026privacyimpliesstabilityinformationtheoretic}.
\end{remark}

\subsection{Information decay of layered quantum circuits}

In training and using noisy quantum circuits for inference, it is important to understand how the output information could be utilized to infer information about the inputs or how outputs of different inputs can be distinguished.
In~\cite{george2025quantumdoeblincoefficientsinterpretations}, the existence of noise-induced barren-plateaus and the sample complexity of error mitigation  were discussed for composite noisy channels of the form:
\begin{equation}
    \cM_d= \id_R \otimes (\cN_d \circ \cU_d \circ \cdots \cN_1 \circ \cU_1),
\end{equation}
where the subsystem $R$ is noiseless.
This way when we input states $\rho_{RA},\sigma_{RA}$, we get $\Tr_A[\cM_d(\rho)]= \Tr_A[\cM_d(\sigma)]$ whenever $\rho_R=\sigma_R$. However, this may not always be the case with quantum circuits. For instance, there may be noisy circuits (channels) that are of the form:
\begin{equation}
    \cM_d=(\id_R \otimes (\cN_d \circ \cU_d \circ \cdots \cN_1 \circ \cU_1)) \circ \cE, 
\end{equation}
where $\cE$ represents a noisy encoding channel or a noisy pre-trained quantum circuit that involves decoherence, randomness, correlations, etc. In these cases we may have $ \rho_R^\cE \coloneq \Tr_A[\cE(\rho)] \neq \Tr_A[\cE(\sigma)] \eqcolon \sigma_R^\cE$, which makes the previous analysis in~\cite[Section~8]{george2025quantumdoeblincoefficientsinterpretations} not applicable. In such scenarios, the tools developed in this work will provide useful insights on the impact of the channel $\cE$.

Consider the following:
\begin{align}
    &T(\cM_d(\rho_{RA}), \cM_d(\sigma_{RA})) \nonumber \\ 
    &\leq \etaD_{\Tr}(\cN_d) \left( T(\cM_{d-1}(\rho_{RA}), \cM_{d-1}(\sigma_{RA})) -T(\rho_R^\cE,\sigma_R^\cE )\right) +T(\rho_R^\cE,\sigma_R^\cE ) \\
    &\leq \prod_{i=1}^{d} \etaD_{\Tr}(\cN_i) \left( T(\cE(\rho_{RA}), \cE(\sigma_{RA})) -T(\rho_R^\cE,\sigma_R^\cE )\right) +T(\rho_R^\cE,\sigma_R^\cE ) \\ 
    &=\left(\prod_{i=1}^{d} \etaD_{\Tr}(\cN_i) \right)  T(\cE(\rho_{RA}), \cE(\sigma_{RA})) +\left(1- \prod_{i=1}^{d} \etaD_{\Tr}(\cN_i) \right)T(\rho_R^\cE,\sigma_R^\cE ). 
\end{align}
If $T(\rho_R^\cE,\sigma_R^\cE )$ is also quite small, then similar conclusions as in~\cite[Section~8]{george2025quantumdoeblincoefficientsinterpretations} can be drawn. Under that setting, an exponential number of data samples are required for error mitigation, and noisy barren plateaus also exist in the considered noisy circuit setup (see Section~8.1 and Section~8.2~in~\cite{george2025quantumdoeblincoefficientsinterpretations}). More generally, this would give us some insights on how we can decide on choosing circuit depth, encoder based on noisy channels to minimize information decay.

\subsection{Information scrambling in quantum circuits with prior correlations} 

Quantum information scrambling refers to the delocalization of initially localized information across other subsystems or over the entire quantum circuit.
An important distinction is that scrambling does not necessarily destroy information. For an isolated system undergoing unitary evolution, the information remains globally recoverable, even though it may become inaccessible to observers restricted to small subsystems. This observation means that applying a contraction coefficient directly to the full unitary circuit does not provide any useful insight regarding scrambling. That is because, for every unitary channel $\cU$ and every divergence that satisfies unitary invaraince,
\begin{align}
\etaD_{\DD}(\cU)=\etaC_{\DD}(\cU)=1.
\end{align}
With that, we investigate scrambling through reduced channels describing the information accessible to restricted observers, which may be of relevance for distributed quantum circuits. Approaches along these lines have been studied for example in~\cite{ding2016conditional}. 

Let us consider the following setting: Consider a quantum state $\rho_{RAC}$. System $A$ is passed through the quantum circuit (channel) denoted by $\cE_{A \to B_1\cdots B_n}$. System $C$ is inaccessible at the output of the quantum circuit, while system $R$ corresponds to prior correlations within subsystems. We want to study the information accessible to a collection of $B_i$ about $C$ given access to $R$. With that in mind we define the following: 
\begin{align}
    \cN_{S} & \coloneq \Tr_{S^c} \circ \cE , \\
    \Gamma_m & \coloneq \sup_{S \in \cR_m} \etaD_{D}(\cN_S),
\end{align}
where $\cR_m \coloneq \{ S \subseteq \{1,\ldots,n\}: |S|=m   \}$ for $m \leq n$.
With that, we arrive at the following for all $S \in \cR_m$:
\begin{equation}
    I(C:\cup_{i\in S}  B_i|R) \leq \Gamma_m \ I(C:A|R).
\end{equation}
by application of the conditional contraction coefficient. 

To differentiate the information loss and information scrambling across subsystems we could use the assistance of the two quantities $\etaD_D (\cE)$ and $ \Gamma_m$. Note that 
\begin{equation}
    \etaD_{D}(\cN_S) \leq \etaD_D (\cE) 
\end{equation}
by the concatenation property in~\eqref{eq:concatenation_etaD} since $\cN_S= \Tr_{S^c} \circ \cE$ and $\etaD_D(\cdot) \leq 1$.
If $\Gamma_m$ is small and $\etaD_D (\cE) $ is closer to one, then information survives globally for some input states but is hidden from small regions corresponding to $\cR_m$, showcasing a regime where scrambling occurs. If both of them are small, then the circuit has destroyed the information rather than delocalizing it.
Furthermore, if  $\Gamma_m$ is small and $\cetaD_D (\cE) $ is closer to one, the information survives globally for all input states but is hidden from all small regions corresponding to $\cR_m$. However, it is not always the case that $\cetaD_D (\cE) $ is closer to one.

\section{Conclusion}

\label{Sec:Conclusions}

Our paper provided an in-depth investigation of contraction of quantum divergences in settings with reference systems. Such problems naturally appear in network and multi-user settings. While many of the quantities considered here reduce to the usual contraction coefficient in the classical setting, quantum references exhibit a rich structure that shows a clear departure from the classical setting.

We also reported progress in the study of SDPI constants. Interestingly, while relative-entropy SDPI constants, and more generally many $f$-divergence SDPI constants, tensorize classically, the commonly considered classical-quantum SDPI constants do not tensorize in general. In contrast, fully quantum mutual information and conditional mutual information SDPI constants recover the tensorization property. This suggests that tensorization in the quantum setting depends crucially on how quantum side information is incorporated, and that fully quantum formulations can possess favorable compositional properties absent from their classical-quantum counterparts (with related tensor-stable constructions appearing, e.g., in~\cite{gao2020fisher,gao2022complete}).

Finally, we presented several applications. Although our list is extensive, it is certainly not exhaustive, and we expect our framework to find use in many places where one cares about information decay in multi-user quantum settings. More broadly, these results point towards a compositional theory of information loss in quantum networks, in which contraction properties of individual noisy components can be used to control the propagation and lifetime of information in larger quantum systems. A central question going forward is therefore to understand precisely when such local-to-global principles hold in the presence of quantum side information, in particular also in the setting of SDPI constants.

To end, we briefly state some open problems:
\begin{itemize}
\item When considering the hockey-stick divergence, can the optimization in the contraction coefficient be restricted to orthogonal states, similar to~\cref{Prop:sup-ort-enough}? This would directly imply that the associated contraction coefficient is monotone in $\gamma$, which in turn would have applications similar to the results in~\cite{hirche2024quantumDivergences}.
\item Can we find a concrete example such that $\etaC_\DD(\cN) < \etaD_\DD(\cN)$ for some faithful divergence $\DD$? A particularly interesting case would be that of the relative entropy. Although it would be natural to conjecture that such an example exists, the unbounded reference systems make it difficult to confirm.
\item Although most basic questions about tensorization now appear to be resolved, it would still be interesting to characterize more precisely when quantities such as the relative entropy SDPI tensorize or fail to tensorize, and which structural properties of the divergence and the allowed side information determine this behavior.
\item We did not touch on non-linear strong data-processing inequalities with reference systems (e.g., as studied in~\cite{nuradha2025nonlinearSDPI}). This could be a natural future direction.
\end{itemize}

\section*{Acknowledgments}

CH received funding by the Deutsche Forschungsgemeinschaft (DFG, German
Research Foundation) – 550206990, the Federal Ministry of Research, Technology and Space (BMFTR), Germany, under the QC service center QUICS (grant no.~13N17418), and by Germany's Excellence Strategy – EXC 2123/2 QuantumFrontiers – 390837967. This project is supported by the Ministry of Education, Singapore, through grant T2EP20124-0005. This project is supported by the National Research Foundation, Singapore under the NRF Postdoctoral award. TN acknowledges support from the
Department of Mathematics and the IQUIST Postdoctoral Fellowship from
the Illinois Quantum Information Science and Technology Center at
the University of Illinois Urbana-Champaign.
MMW acknowledges support from the National Science Foundation under grant
no.~2329662.

\section*{Statement on AI usage}

We acknowledge the use of ChatGPT (GPT-5.6 Sol) in the conception of~\cref{lem:Choi-op-norm-to-Doeblin} and~\cref{Prop:no-tensorization-stabilized-complete-SDPI}, as well as for polishing the figures, correcting some typos, and detecting an error in an earlier draft. 

\bibliographystyle{quantum}
\bibliography{lib}

\appendix

\section{Contraction coefficients with classical reference systems}

\label{app:classical-ref}

In this appendix, we begin by defining various contraction coefficients that involve classical reference systems only, and after that, we prove that the resulting contraction coefficients are actually equal to the standard one. As such, there is no need to consider complete or conditional contraction coefficients when the reference system is classical.

\begin{definition}[Generalized contraction coefficient with classical reference]
    For a quantum channel $\cN$ and a generalized divergence $\DD$, we define the complete contraction coefficient with classical reference
    \begin{align}
        \etaC_{\DD,\operatorname{CQ}}(\cN) \coloneqq \sup_{\substack{\rho_{RA}\neq\sigma_{RA},\\ \rho_R=\sigma_R, 
        |R|,\\
        \rho_{RA},\sigma_{RA} \in \operatorname{CQ}}} \frac{\DD((\id\otimes\cN)(\rho_{RA})\|(\id\otimes\cN)(\sigma_{RA}))}{\DD(\rho_{RA}\|\sigma_{RA})}, \label{eq:Comp_CC_CQ}
    \end{align}
    where CQ denotes the set of classical--quantum states.
    We also define the conditional contraction coefficient, 
    \begin{align}
        \etaD_{\DD,\operatorname{CQ}}(\cN) \coloneqq \sup_{\substack{|R|,\,\rho_{RA}\neq\sigma_{RA},\\
        \rho_{RA},\sigma_{RA} \in \operatorname{CQ}}} \frac{\DD((\id\otimes\cN)(\rho_{RA})\|(\id\otimes\cN)(\sigma_{RA})) - \DD(\rho_R\|\sigma_R)}{\DD(\rho_{RA}\|\sigma_{RA})- \DD(\rho_R\|\sigma_R)}.\label{Eq:conditional-def_CQ}
    \end{align}
\end{definition}

\begin{lemma} \label{lem:extended_Direct_sum}
Let $\rho_{RA}$ and $\sigma_{RA}$ be classical-quantum states of the following form:
\begin{align}
    \rho_{RA} & = \sum_r p(r) |r\rangle\!\langle r|_R \otimes \rho^r_A, \\
    \sigma_{RA} & = \sum_r q(r) |r\rangle\!\langle r|_R \otimes \sigma^r_A, 
\end{align}
where $p(r)$ and $q(r)$ are probability distributions and $\left(\rho^r_A\right)_r$ and $\left(\sigma^r_A\right)_r$ are tuples of quantum states.
    Let $\DD$ be a divergence with the extended direct-sum property, 
    \begin{align}
        \DD(\rho_{RA}\|\sigma_{RA}) = \sum_r g(p(r),q(r)) \DD(\rho_A^r \|\sigma_A^r) + \DD(p\|q),
    \end{align}
    where $g(\cdot,\cdot)$ is a positive function. Then
   the following equalities hold:
    \begin{align}
        \eta_{\DD,\operatorname{CQ}}(\cN) = \eta^c_{\DD,\operatorname{CQ}}(\cN) = \etaD_{\DD,\operatorname{CQ}}(\cN) . 
    \end{align}
\end{lemma}

\begin{proof}
    The result follows from applying the extended direct-sum property and the fact that $\DD(p\|q) = \DD(\rho_R\|\sigma_R)$: 
    \begin{align}
        &\DD((\id\otimes\cN)(\rho_{RA})\|(\id\otimes\cN)(\sigma_{RA})) - \DD(\rho_R\|\sigma_R) \notag \\
        &= \sum_r g(p(r),q(r)) \DD(\cN(\rho^r_{A})\|\cN(\sigma^r_{A})) + \DD(p\|q) - \DD(\rho_R\|\sigma_R) \\
        &= \sum_r g(p(r),q(r)) \DD(\cN(\rho^r_{A})\|\cN(\sigma^r_{A})) \\
        &\leq \sum_r g(p(r),q(r)) \eta_\DD(\cN) \DD(\rho^r_{A}\|\sigma^r_{A}) \\
        &= \eta_\DD(\cN) \left[ \DD(\rho_{RA}\|\sigma_{RA}) - \DD(\rho_R\|\sigma_R) \right].
    \end{align}
     Thus,
    \begin{align}
    \frac{\DD((\id\otimes\cN)(\rho_{RA})\|(\id\otimes\cN)(\sigma_{RA})) - \DD(\rho_R\|\sigma_R)}{\DD(\rho_{RA}\|\sigma_{RA}) - \DD(\rho_R\|\sigma_R)} \leq \eta_\DD(\cN).
    \end{align}
    Since this inequality holds for all classical-quantum states $\rho_{RA} $ and $\sigma_{RA}$, we conclude that $\etaD_{\DD,\operatorname{CQ}}(\cN)\leq \eta_\DD(\cN)$. We conclude the proof after combining this inequality with $\eta_\DD(\cN) \leq \eta^c_{\DD,\operatorname{CQ}}(\cN) \leq \etaD_{\DD,\operatorname{CQ}}(\cN)$, which are direct consequences of definitions.
\end{proof}

Note that the extended direct sum property required in~\cref{lem:extended_Direct_sum} is satisfied by Umegaki and Belavkin–Staszewski quantum relative, with $g(p(r),q(r))=p(r)$, and many quantum $\alpha$-Hellinger divergences, with $g(p(r),q(r))=p(r)^\alpha q(r)^{1-\alpha}$, see for example~\cite{khatri2024principlesquantumcommunicationtheory}. 
However, it does not include other common divergences, such as the trace distance or the hockey stick divergences.

\section{Limit property of relative entropy and generalized divergences}

\label{app:prop-rel-ent}

\begin{lemma}
\label{lem:prop-rel-ent}
Let $\rho$ and $\sigma$ be full-rank states, define
\begin{equation}
\rho(\lambda)\coloneqq\frac{\sigma-\lambda\rho}{1-\lambda},
\label{eq:rho-lam-state-def}
\end{equation}
where $\lambda\in\left[0,1\right]$ is chosen such that $\rho(\lambda)$
is a full-rank state. Then
\begin{equation}
\left.\frac{\partial}{\partial\lambda}D(\rho(\lambda)\|\sigma)\right|_{\lambda=0}=0.
\end{equation}
\end{lemma}

\begin{proof}
Using the facts that
\begin{align}
\frac{\partial}{\partial x}\Tr[f(A(x))] & =\Tr\!\left[f'(A(x))\frac{\partial}{\partial x}A(x)\right],\\
\frac{\partial}{\partial x}\left(x\ln x\right) & =\ln x+1,\\
\frac{\partial}{\partial\lambda}\rho(\lambda) & =\frac{\sigma-\rho}{\left(1-\lambda\right)^{2}},
\end{align}
consider that
\begin{align}
\frac{\partial}{\partial\lambda}D(\rho(\lambda)\|\sigma) & =\frac{\partial}{\partial\lambda}\Tr[\rho(\lambda)\ln\rho(\lambda)]-\frac{\partial}{\partial\lambda}\Tr[\rho(\lambda)\ln\sigma]\\
 & =\Tr\!\left[\left(\ln\rho(\lambda)+I\right)\frac{\partial}{\partial\lambda}\rho(\lambda)\right]-\Tr\!\left[\left(\frac{\partial}{\partial\lambda}\rho(\lambda)\right)\ln\sigma\right]\\
 & =\Tr\!\left[\left(\ln\rho(\lambda)+I\right)\left(\frac{\sigma-\rho}{\left(1-\lambda\right)^{2}}\right)\right]-\Tr\!\left[\left(\frac{\sigma-\rho}{\left(1-\lambda\right)^{2}}\right)\ln\sigma\right]\\
 & =\frac{1}{\left(1-\lambda\right)^{2}}\Tr\!\left[\left(\ln\rho(\lambda)-\ln\sigma\right)\left(\sigma-\rho\right)\right].
\end{align}
After observing that $\lim_{\lambda\to0}\rho(\lambda)=\sigma$, we
conclude the proof.
\end{proof}

In the main text, two different manifestations of the above condition appear. Here we clarify that one does imply the other, even for general divergences. 

\begin{lemma}
\label{Lem:partial-equiv}
Let $\rho$ and $\sigma$ be full-rank states and suppose that
\begin{align}
\left.
\frac{\partial}{\partial\lambda}
\DD(\lambda\tau+(1-\lambda)\sigma\|\sigma)
\right|_{\lambda=0^+}
=0
\end{align}
for every state $\tau$.
Then, for
\begin{align}
\rho(\lambda)
\coloneqq
\frac{\sigma-\lambda\rho}{1-\lambda},
\end{align}
we have
\begin{align}
\left.
\frac{\partial}{\partial\lambda}
\DD(\rho(\lambda)\|\sigma)
\right|_{\lambda=0^+}
=0.
\end{align}
\end{lemma}

\begin{proof}
Since $\sigma$ is full rank, there exists $\varepsilon>0$ sufficiently
small such that
\begin{align}
\tau\coloneqq(1+\varepsilon)\sigma-\varepsilon\rho
\end{align}
is a quantum state. Then
\begin{align}
\rho(\lambda)
&=
\sigma+\frac{\lambda}{1-\lambda}(\sigma-\rho)\\
&=
\sigma+
\frac{\lambda}{\varepsilon(1-\lambda)}(\tau-\sigma)\\
&=
(1-\mu(\lambda))\sigma+\mu(\lambda)\tau,
\end{align}
where
\begin{align}
\mu(\lambda)
\coloneqq
\frac{\lambda}{\varepsilon(1-\lambda)}.
\end{align}
For all sufficiently small $\lambda\geq0$, also $\mu(\lambda)\geq0$,
with $\mu(0)=0$ and $\mu'(0)=1/\varepsilon$. Hence, by the chain rule,
\begin{align}
\left.
\frac{\partial}{\partial\lambda}
\DD(\rho(\lambda)\|\sigma)
\right|_{\lambda=0^+}
=
\frac1{\varepsilon}
\left.
\frac{\partial}{\partial\mu}
\DD((1-\mu)\sigma+\mu\tau\|\sigma)
\right|_{\mu=0^+}
=0,
\end{align}
where the final equality follows from the assumed vanishing first
derivative property.
\end{proof}

\section{Failure of tensorization for the stabilized complete SDPI}
\label{Sec:no-tensor-for-stabC}
In this section, we prove that the stabilized complete SDPI constant does not tensorize. 
\begin{proposition}
\label{Prop:no-tensorization-stabilized-complete-SDPI}
There exist a quantum channel $\cN\colon A\to B$ and a state
$\sigma_{RA}$ such that
\begin{align}
    \eta_D^{c,\mathrm{stab}}(\cN,\sigma_{RA})
    =
    \frac12
    <
    \eta_D^{c,\mathrm{stab}}
    \!\left(
        \cN^{\otimes2},
        \sigma_{RA}^{\otimes2}
    \right).
\end{align}
Consequently, the stabilized complete relative-entropy SDPI constant
does not tensorize in general.
\end{proposition}

\begin{proof}
Let $R$ be a qubit and let $A$ have orthonormal basis
$\{\ket{x}\}_{x=0}^3$. Define
\begin{align}
    \omega_0&\coloneqq\ket{+}\!\bra{+},
    &
    \omega_1&\coloneqq\ket{-}\!\bra{-},
    &
    \omega_2&\coloneqq\ket{0}\!\bra{0},
    &
    \omega_3&\coloneqq\ket{1}\!\bra{1},
\end{align}
where $\ket{\pm}=(\ket0\pm\ket1)/\sqrt{2}$. The crucial relation is
\begin{align}
    \omega_0+\omega_1
    =
    \omega_2+\omega_3
    =
    I_R.
    \label{eq:four-state-relation-counterexample}
\end{align}
Set
\begin{align}
    \xi_x
    \coloneqq
    \omega_x\otimes\ket{x}\!\bra{x}_A,
    \qquad
    \sigma_{RA}
    \coloneqq
    \frac14\sum_{x=0}^3\xi_x,
    \label{eq:fixed-state-counterexample}
\end{align}
so that $\sigma_R=I_R/2$. Let $\cN\colon A\to B$ record whether the input is equal to $2$:
\begin{align}
    \cN(\rho)
    \coloneqq
    \Tr[\ket{2}\!\bra{2}\rho] \ \ket{1}\!\bra{1}_B
    +
    \Tr[(I-\ket{2}\!\bra{2})\rho] \ \ket{0}\!\bra{0}_B.
    \label{eq:channel-counterexample-stabilized-complete}
\end{align}
We use the following standard identity for states that are block
diagonal in the same classical basis:
\begin{align}
    &D\!\left(
        \sum_i p_i\rho_i\otimes\ket{i}\!\bra{i}
        \middle\|
        \sum_i q_i\sigma_i\otimes\ket{i}\!\bra{i}
    \right)=
    D(p\|q)+\sum_i p_iD(\rho_i\|\sigma_i).
    \label{eq:relative-entropy-direct-sum-used}
\end{align}

\paragraph{One copy.}
Fix an arbitrary stabilizing system $S$ and state $\tau_S$. Since
$\cN$ depends only on the diagonal entries of $A$, dephasing $A$
leaves the numerator of the complete-SDPI ratio unchanged and cannot
increase its denominator. Hence it suffices to consider a dephased
state.

The $x$-th block of the fixed second state
$\tau_S\otimes\sigma_{RA}$ is
$\frac14\tau_S\otimes\omega_x$. Finiteness of the relative entropy
therefore requires the corresponding $SR$-block of the first state to
be supported on
\begin{align}
    \operatorname{supp}(\tau_S)
    \otimes
    \operatorname{supp}(\omega_x).
\end{align}
Since $\omega_x$ is pure, its support on $R$ is one-dimensional.
Thus the $SR$-block is necessarily a product state, and we may write
\begin{align}
    \rho_{SRA}
    =
    \sum_{x=0}^3
    p_x\rho_S^x\otimes\xi_x,
    \label{eq:block-form-counterexample}
\end{align}
where $(p_x)_x$ is a probability distribution and the $\rho_S^x$ are
normalized states.
The equal-reference-marginal condition is
\begin{align}
    \sum_{x=0}^3
    p_x\rho_S^x\otimes\omega_x
    =
    \tau_S\otimes\frac{I_R}{2}.
    \label{eq:equal-reference-marginal-counterexample}
\end{align}
Writing the four states in the Pauli basis,
\begin{align}
    \omega_0&=\frac12(I_R+X_R),
    &
    \omega_1&=\frac12(I_R-X_R),\\
    \omega_2&=\frac12(I_R+Z_R),
    &
    \omega_3&=\frac12(I_R-Z_R),
\end{align}
the left-hand side of
\eqref{eq:equal-reference-marginal-counterexample} becomes
\begin{align}
    &\frac12
    \left(\sum_xp_x\rho_S^x\right)\otimes I_R
    +
    \frac12
    (p_0\rho_S^0-p_1\rho_S^1)\otimes X_R
    \nonumber\\
    &\qquad+
    \frac12
    (p_2\rho_S^2-p_3\rho_S^3)\otimes Z_R.
\end{align}
Here $X_R$ and $Z_R$ are the usual Pauli matrices acting on $R$.
Since $I_R$, $X_R$, and $Z_R$ are linearly independent, their
operator-valued coefficients on $S$ must agree separately with those
on the right-hand side of
\eqref{eq:equal-reference-marginal-counterexample}. Hence
\begin{align}
    p_0\rho_S^0&=p_1\rho_S^1,
    &
    p_2\rho_S^2&=p_3\rho_S^3,
    &
    \sum_xp_x\rho_S^x&=\tau_S.
\end{align}
Taking traces, we can therefore write, for some $\alpha\in[0,1]$,
\begin{align}
    \rho_{SRA}
    &=
    \frac{\alpha}{2}\rho_S^+
        \otimes(\xi_0+\xi_1)
    +
    \frac{1-\alpha}{2}\rho_S^-
        \otimes(\xi_2+\xi_3),
    \label{eq:one-copy-normal-form}\\
    \tau_S
    &=
    \alpha\rho_S^+
    +(1-\alpha)\rho_S^-.
    \label{eq:one-copy-tau}
\end{align}
Applying~\eqref{eq:relative-entropy-direct-sum-used} to the four
orthogonal values of $A$ gives
\begin{align}
    &D(\rho_{SRA}\|\tau_S\otimes\sigma_{RA})=
    \alpha D(\rho_S^+\|\tau_S)
    +(1-\alpha)D(\rho_S^-\|\tau_S)+
    D\!\left(
        (\alpha,1-\alpha)
        \middle\|
        \left(\frac12,\frac12\right)
    \right).
    \label{eq:one-copy-input-counterexample}
\end{align}
Using~\eqref{eq:four-state-relation-counterexample}, the two output
states can be written as
\begin{align}
    \cN(\rho_{SRA})
    &=
    \frac{\alpha}{2}
    \rho_S^+\otimes\omega_2\otimes\ket0\!\bra0_B
    +
    \frac{1-\alpha}{2}
    \rho_S^-\otimes\omega_2\otimes\ket1\!\bra1_B
    \nonumber\\
    &\quad+
    \frac12
    \tau_S\otimes\omega_3\otimes\ket0\!\bra0_B,
    \label{eq:one-copy-first-output}\\
    \cN(\tau_S\otimes\sigma_{RA})
    &=
    \frac14\tau_S\otimes\omega_2\otimes\ket0\!\bra0_B
    +
    \frac14\tau_S\otimes\omega_2\otimes\ket1\!\bra1_B
    \nonumber\\
    &\quad+
    \frac12\tau_S\otimes\omega_3\otimes\ket0\!\bra0_B.
    \label{eq:one-copy-second-output}
\end{align}
The three terms in each expression have mutually orthogonal supports.
Applying~\eqref{eq:relative-entropy-direct-sum-used} to these three
orthogonal sectors gives
\begin{align}
    D\!\left(
        \cN(\rho_{SRA})
        \middle\|
        \cN(\tau_S\otimes\sigma_{RA})
    \right)
    =
    \frac12
    D(\rho_{SRA}\|\tau_S\otimes\sigma_{RA}).
\end{align}
The value $1/2$ is attained by
\begin{align}
    \rho_{SRA}
    =
    \frac12\tau_S\otimes(\xi_0+\xi_1).
\end{align}
Since $S$ and $\tau_S$ were arbitrary,
\begin{align}
    \eta_D^{\mathrm{C,stab}}(\cN,\sigma_{RA})
    =
    \frac12.
    \label{eq:one-copy-value-counterexample}
\end{align}

\paragraph{Two copies.}
For $0<\varepsilon<1/2$, consider the state
\begin{align}
    \rho^\varepsilon
    \coloneqq
    \sigma_{RA}^{\otimes2}
    +
    \frac{\varepsilon}{16}
    \Big[
        (\xi_0+\xi_1-\xi_2-\xi_3)\otimes\xi_2
          +
        \xi_2\otimes
        (\xi_0+\xi_1-\xi_2-\xi_3)
    \Big].
    \label{eq:two-copy-small-perturbation}
\end{align}
The reference marginal is unchanged because
\begin{align}
    \omega_0+\omega_1-\omega_2-\omega_3=0.
\end{align}
Hence
\begin{align}
    \rho^\varepsilon_{R_1R_2}
    =
    \sigma_{R_1}\otimes\sigma_{R_2},
\end{align}
so $\rho^\varepsilon$ is admissible for the two-copy complete SDPI.
The state $\sigma_{RA}^{\otimes2}$ corresponds to the uniform
distribution over the $16$ classical label pairs $(x,y)$. The
perturbation in~\eqref{eq:two-copy-small-perturbation} changes this
distribution by the relative amounts
\begin{align}
    +\varepsilon
    &\quad\text{for }(0,2),(1,2),(2,0),(2,1),\\
    -\varepsilon
    &\quad\text{for }(3,2),(2,3),\\
    -2\varepsilon
    &\quad\text{for }(2,2),
\end{align}
and leaves all other probabilities unchanged. This also makes
positivity clear for $0<\varepsilon<1/2$.
Since the conditional quantum state
$\omega_x\otimes\omega_y$ associated with each label pair is the same
in $\rho^\varepsilon$ and $\sigma_{RA}^{\otimes2}$, their relative
entropy is simply the classical relative entropy of these label
distributions. We use the elementary expansion
\begin{align}
    D\!\left(
        \big(q_i(1+\varepsilon a_i)\big)_i
        \middle\|
        (q_i)_i
    \right)
    =
    \frac{\varepsilon^2}{2}
    \sum_iq_i a_i^2
    +O(\varepsilon^3),
    \label{eq:classical-relative-entropy-expansion}
\end{align}
valid when $\sum_iq_i a_i=0$. Since every label pair has probability
$1/16$ under the uniform distribution, we obtain
\begin{align}
    D(\rho^\varepsilon\|\sigma_{RA}^{\otimes2})
    &=
    \frac{\varepsilon^2}{32}
    \left(
        4\cdot1^2+2\cdot(-1)^2+(-2)^2
    \right)
    +O(\varepsilon^3)\\
    &=
    \frac{5}{16}\varepsilon^2
    +O(\varepsilon^3).
    \label{eq:two-copy-input-expansion}
\end{align}

At the output,
\begin{align}
    \cN(\xi_0+\xi_1-\xi_2-\xi_3)
    =
    \omega_2\otimes
    \left(
        \ket0\!\bra0-\ket1\!\bra1
    \right)_B,
    \label{eq:one-copy-perturbation-output}
\end{align}
while
\begin{align}
    \cN(\xi_2)
    =
    \omega_2\otimes\ket1\!\bra1_B.
\end{align}
Therefore,
\begin{align}
    &\cN^{\otimes2}(\rho^\varepsilon)
    -
    \cN(\sigma_{RA})^{\otimes2}=
    \frac{\varepsilon}{16}
    \omega_2^{\otimes2}
    \otimes
    \left(
        \ket{01}\!\bra{01}
        +
        \ket{10}\!\bra{10}
        -
        2\ket{11}\!\bra{11}
    \right)_{B_1B_2}.
    \label{eq:two-copy-output-perturbation}
\end{align}
Furthermore,
\begin{align}
    \cN(\sigma_{RA})
    &=
    \frac14\omega_2\otimes\ket0\!\bra0_B
    +
    \frac14\omega_2\otimes\ket1\!\bra1_B
    +
\frac12\omega_3\otimes\ket0\!\bra0_B.
\end{align}
Thus the three mutually orthogonal states appearing in
\eqref{eq:two-copy-output-perturbation} each have probability $1/16$
under the unperturbed two-copy output. Their relative perturbations are
$+\varepsilon,+\varepsilon,-2\varepsilon$, respectively. Applying
\eqref{eq:classical-relative-entropy-expansion} again gives
\begin{align}
    D\!\left(
        \cN^{\otimes2}(\rho^\varepsilon)
        \middle\|
        \cN(\sigma_{RA})^{\otimes2}
    \right) &=
    \frac{\varepsilon^2}{32}
    \left(
        1^2+1^2+(-2)^2
    \right)
    +O(\varepsilon^3)\\
    &=
    \frac{3}{16}\varepsilon^2
    +O(\varepsilon^3).
    \label{eq:two-copy-output-expansion}
\end{align}
Consequently,
\begin{align}
    \lim_{\varepsilon\searrow0}
    \frac{
        D\!\left(
            \cN^{\otimes2}(\rho^\varepsilon)
            \middle\|
            \cN(\sigma_{RA})^{\otimes2}
        \right)
    }{
        D(\rho^\varepsilon\|\sigma_{RA}^{\otimes2})
    }
    =
    \frac35
    >
    \frac12.
\end{align}
Hence, for all sufficiently small $\varepsilon>0$,
\begin{align}
    \eta_D^{\mathrm{C,stab}}
    \!\left(
        \cN^{\otimes2},
        \sigma_{RA}^{\otimes2}
    \right)
    >
    \frac12
    =
    \eta_D^{c,\mathrm{stab}}(\cN,\sigma_{RA}),
\end{align}
which completes the proof.
\end{proof}

\end{document}